\documentclass[onecolumn,notitlepage,nofootinbib,superscriptaddress,floatfix]{revtex4-2}
 
\usepackage{float}
\makeatletter
\let\newfloat\newfloat@ltx
\makeatother

\usepackage{algorithm}
\usepackage{algorithmic}
\usepackage{array}

\newcommand{\fu}{
Dahlem Center for Complex Quantum Systems,
Freie Universität Berlin, 14195 Berlin, Germany
}

\newcommand{\caltech}{Institute for Quantum Information and Matter, California Institute of Technology, Pasadena, CA 91125, USA}

\newcommand{\gqaivenice}{
Google Quantum AI, Venice, CA 90291, USA
}

\newcommand{\equalcontrib}{
\thanks{These authors contributed equally to this work.}
}

\usepackage[utf8]{inputenc}       
\usepackage[english]{babel}

\usepackage{amsmath,amssymb,amsthm,mathtools}
\usepackage{physics}

\usepackage{graphicx}

\usepackage[shortlabels]{enumitem}

\usepackage[caption=false]{subfig} 
\usepackage{xcolor}
\usepackage{booktabs}
\usepackage{adjustbox}

\usepackage{array}

\usepackage{tikz}
\usetikzlibrary{positioning,decorations.pathreplacing}
\usetikzlibrary{decorations.pathreplacing}
\usepackage{tikz-cd}
\usepackage{pgfplots}
\usepgfplotslibrary{groupplots,fillbetween}
\pgfplotsset{compat=1.18}

\usepackage{csquotes}
\usepackage{IEEEtrantools}
\usepackage{relsize}
\usepackage{placeins}
\usepackage{comment}
\usepackage[normalem]{ulem}
\usepackage{systeme}
\usepackage{pgffor}

\definecolor{lb}{rgb}{.5,.1,-.7}

\usepackage[pdfpagelabels]{hyperref}
\hypersetup{
  colorlinks=true,
  linkcolor  = [rgb]{0.70,0.13,0.13},
  citecolor  = [rgb]{0.13,0.55,0.13},
  urlcolor   = [rgb]{0.25,0.41,0.88}
}

\theoremstyle{plain}
\newtheorem{theorem}{Theorem}[section]
\newtheorem{lemma}[theorem]{Lemma}
\newtheorem{proposition}[theorem]{Proposition}
\newtheorem{corollary}[theorem]{Corollary}

\theoremstyle{definition}
\newtheorem{definition}[theorem]{Definition}

\newtheorem{remark}[theorem]{Remark}

\newcommand{\E}{\mathbb{E}}

\providecommand{\Tr}{\operatorname{Tr}}
\DeclareMathOperator{\Var}{Var}

\DeclareMathOperator{\dist}{dist}
\DeclareMathOperator{\Haar}{Haar}

\begin{document}

\author{Antonio A. Mele}
\equalcontrib
\affiliation{\fu}
\affiliation{\gqaivenice}

\author{Francesco A. Mele}
\equalcontrib
\affiliation{\caltech}
\affiliation{\gqaivenice}

\author{Jarrod R. McClean}
\affiliation{\gqaivenice}

\author{Thomas E. O'Brien}
\affiliation{\gqaivenice}





\title{Towards verifiable quantum advantage with random circuits:\\
Observables that survive concentration}
\begin{abstract}
Demonstrating quantum advantage on current quantum hardware is a central goal of quantum computing, and random quantum circuits underpin many leading proposals. Yet sampling-based demonstrations are often difficult to verify, while observable-based approaches face a different challenge: concentration can suppress differences between circuit instances. Recent experiments have put forward the estimation of out-of-time-order correlators (OTOCs) in random circuits as a promising task for verifiable quantum advantage, yet whether their circuit-to-circuit fluctuations survive concentration as system size grows has remained open. Here we show that they do. For broad classes of local random circuits in any fixed spatial dimension, we prove inverse-polynomial fluctuations of fixed-order OTOCs at system-scale depths. In one-dimensional Haar-random brickwork circuits, we further show that macroscopically many gates contribute to these fluctuations, yet in each layer they remain confined to a sublinear-width region. As a byproduct, we develop a classical algorithm exhibiting the first rigorous improvement over brute-force classical simulation of OTOCs in 1D. Although classical hardness remains open, our results rule out strong concentration as an obstruction to OTOC-based proposals for verifiable quantum advantage.
\end{abstract}
\maketitle
\vspace*{-0.7cm}
\section{Introduction}
Recent years have witnessed intense efforts to demonstrate quantum advantage
using present-day quantum devices~\cite{
boixo2018characterizing,
bouland2019complexity,
Arute_2019}.
Such demonstrations would be especially compelling if the task were
efficiently \emph{verifiable}~\cite{
bouland2019complexity,
mahadev2018classical,
yamakawa2022verifiable,
aaronson2024verifiable,
MartielEtAl2026,BermejoVegaEtAl2018,
HangleiterEtAl2019,LanesEtAl2025}: the output of the quantum computation should be
independently checkable, for example by reproducing the same quantity on
another quantum processor or by comparing it with an independently
accessible physical observable in nature~\cite{king2025simplifiedversionquantumotoc2,BarronEtAl2026}. Ideally, one would also like the advantage
to be \emph{generic}, meaning that it persists for typical circuits drawn
from a simple and experimentally natural family rather than relying on
specially engineered instances.

Random quantum circuit sampling has been the paradigmatic proposal for
demonstrating generic near-term quantum advantage~\cite{HangleiterEtAl2019,MartielEtAl2026}. Sufficiently deep random
circuits are indeed believed to be average-case hard to sample from
classically~\cite{
boixo2018characterizing,
bouland2019complexity,
Arute_2019}; however, the sampling task is not efficiently verifiable, even by another quantum computer: rerunning the same circuit produces another random sample,
which with overwhelming probability differs from the original output and
therefore does not verify it.

A complementary difficulty arises from \emph{concentration}, meaning that the quantity of interest becomes nearly constant across instances. A prominent
example is provided by variational quantum algorithms. Their outputs are
expectation values and are therefore naturally reproducible across devices,
but for sufficiently expressive parametrized circuits cost landscapes
concentrate and gradients can vanish exponentially with system size, leading
to barren plateaus~\cite{
McCleanEtAl2018,
CerezoEtAl2021,
ArrasmithEtAl2022}. Thus, experimental accessibility of an observable alone does not guarantee
that typical large-scale instances retain a nontrivial measurable signal.

A recently proposed task that may combine the average-case nature of random circuits with the verifiability of observable estimation is the estimation of out-of-time-order correlators (OTOCs), recently implemented by Google
Quantum AI~\cite{
abanin2025constructiveinterferenceedgequantum,
king2025simplifiedversionquantumotoc2}.
OTOCs are quantities of fundamental interest in many-body physics, where
they probe operator spreading, information scrambling, and quantum
chaos~\cite{
LarkinOvchinnikov1969,
ShenkerStanford2014,
MaldacenaShenkerStanford2016,
HosurEtAl2016,
Swingle2018,
XuSwingle2024}.

For an \(n\)-qubit circuit \(U\), an input state \(\rho\), and local Pauli
observables \(B\) and \(M\), the \(k\)th-order OTOC is
\begin{align}
\mathrm{OTOC}^{(k)}_{\rho}(U)
\coloneqq
\Tr\!\left[
\rho\left(U^\dagger B U M\right)^{2k}
\right].
\end{align}
We consider ensembles \(\mathcal E_{n,d}\) of \(n\)-qubit random circuits of
depth \(d\). Throughout, \(k=O(1)\) is fixed independently of \(n\), while
\(B\) and \(M\) have constant-size support and are separated by a distance
that grows with the system size; in the one-dimensional geometry considered
below, this distance is \(\Theta(n)\).

The computational task (which is naturally verifiable) is the following: draw a circuit
\(U\sim\mathcal E_{n,d}\) and, given its description, estimate
\(\mathrm{OTOC}^{(k)}_{\rho}(U)\) to additive accuracy
\(\varepsilon=1/\operatorname{poly}(n)\), with high probability over the
choice of \(U\). For every fixed \(k\), this can be done efficiently on a
quantum computer: each experimental repetition requires only \(O(k)\)
applications of \(U\) and \(U^\dagger\), interspersed with the local Pauli
operations \(B\) and \(M\). For example, for \(\mathrm{OTOC}^{(2)}\), with
\(\rho=\ketbra{0^n}{0^n}\) and \(M\ket{0^n}=\ket{0^n}\), one applies
\((U^\dagger B U)M(U^\dagger B U)\) to \(\ket{0^n}\) and measures \(M\).
Repeating this procedure \(O(1/\varepsilon^2)\) times yields an
additive-\(\varepsilon\) estimate
~\cite{king2025simplifiedversionquantumotoc2}. 

At the same time, no efficient classical algorithm is known for estimating \(\mathrm{OTOC}^{(k)}\), for any fixed \(k\).
For the first-order OTOC, Monte-Carlo and tensor-network methods
can perform well on finite-size instances
~\cite{
abanin2025constructiveinterferenceedgequantum,bermejo2026tensornetworks,XuSwingle2020},
but no rigorous asymptotic runtime guarantee is known.
For higher-order OTOCs, the available evidence points to substantially
greater classical difficulty: the Google experiment found that the tested
classical methods could not reproduce its largest
\(\mathrm{OTOC}^{(2)}\) instances at feasible cost, while subsequent work
provided further evidence that tensor-network methods based on belief
propagation do not circumvent this difficulty
~\cite{
abanin2025constructiveinterferenceedgequantum,
bermejo2026tensornetworks}.

However, for an average-case quantum-advantage proposal based on random
circuits, outperforming the best classical algorithms currently known is
only one part of the story. The task must also remain inherently
instance-dependent: if almost all circuits produce essentially the same
value, then the circuit instance becomes irrelevant and the problem reduces
to estimating the ensemble mean.

For the OTOC task considered here, the circuit depth constrains the regime in which such instance dependence can survive.
Before the
Heisenberg light cone of \(B\) reaches the support of \(M\),
\(U^\dagger B U\) and \(M\) commute, and hence
\(\mathrm{OTOC}^{(k)}_{\rho}(U)=1\) for every circuit realization.
At sufficiently large depths, local random circuits approach the relevant
Haar moments through approximate unitary-design convergence~\cite{
brandao2016local,
haferkamp2022random,
harrowmehraban2023approximate,
mittal2023local},
and fixed-order OTOCs become strongly concentrated. Any
nontrivial circuit-to-circuit variation must therefore arise between
these two regimes.

Previous work has provided heuristic, numerical, and experimental evidence that such fluctuations persist in random quantum circuits~\cite{nahum2018operator,mi2021information,abanin2025constructiveinterferenceedgequantum}. Of particular relevance, recent experiments by Google Quantum AI observed circuit-to-circuit fluctuations that decrease only polynomially with system size for both first- and higher-order OTOCs~\cite{abanin2025constructiveinterferenceedgequantum}. Despite this progress, a rigorous proof remained open: it was not known whether these fluctuations remain resolvable as the system size grows, even for \(\mathrm{OTOC}^{(1)}\)~\cite{king2025simplifiedversionquantumotoc2}. This leads to the central question of our work:

\begin{quote}
\emph{Do local random circuits retain resolvable circuit-to-circuit OTOC fluctuations at system-scale depths?}
\end{quote}
More precisely, we ask whether there exist depths \(d=d(n)\) scaling with the system size for which
\begin{align}
\Var_{U\sim\mathcal E_{n,d}}\!\left[\mathrm{OTOC}^{(k)}_{\rho}(U)\right] \geq \frac{1}{\operatorname{poly}(n)}.
\end{align}
Establishing such a variance bound might appear particularly challenging for higher-order OTOCs, as it requires control of the second moment of the circuit, and no known classical methods exist, to our knowledge, to evaluate even the first circuit moment of $\mathrm{OTOC}^{(2)}$.

We answer this question affirmatively without explicitly evaluating these higher-order moments. Instead, we relate changes in the ensemble-averaged OTOC across circuit depth to lower bounds on its circuit-to-circuit variance. For every fixed OTOC order, we prove inverse-polynomial fluctuations at system-scale depths for broad classes of local random circuits, covering local architectures in any fixed spatial dimension and arbitrary initial states. To our knowledge, OTOCs provide the first experimentally accessible observables for which inverse-polynomial circuit-to-circuit fluctuations are rigorously established in broad classes of local random circuits at system-scale depths.

Beyond this general result, we obtain a much sharper picture for
\(\mathrm{OTOC}^{(1)}\) in one-dimensional Haar-random brickwork circuits.
When the two local Pauli operators are placed at opposite ends of the chain,
the circuit-averaged OTOC changes from near \(1\) to near \(0\) around depth
\(d_\star=5n/3\), over a window of width \(\Theta(\sqrt n)\). Throughout
this window, we prove
\(\Var_U[\mathrm{OTOC}^{(1)}_{\rho}(U)]=\Omega(n^{-1/2})\),
substantially strengthening the lower bound from our general theorem.

We then determine how many gates are responsible for these fluctuations and
where they are located. A large variance could, in principle, be caused by
only a few exceptional gates. We show that this is not the case:
macroscopically many gates indeed contribute. At the same time, these gates are not
spread throughout the circuit; in each layer, they are confined to a region
of width \(O(\sqrt n)\). Our analysis makes this localization picture rigorous
at the level of individual gate contributions, complementing the numerical
tensor-network evidence and geometric arguments of
Ref.~\cite{bermejo2026tensornetworks} with an analytical proof for the
Haar-random brickwork ensemble.

Finally, for the infinite-temperature \(\mathrm{OTOC}^{(1)}\) at the critical depth \(d_\star=5n/3\), we exploit this structure for classical simulation. Gates outside the relevant region can be averaged over, so that only a sublinear-width part of the circuit must be treated explicitly. This yields a subexponential-time classical estimator to inverse-polynomial accuracy, providing, to our knowledge, the first provable improvement over brute-force simulation for an OTOC-estimation task.

\begin{figure}[h]
    \centering
    \includegraphics[width=0.85\linewidth]{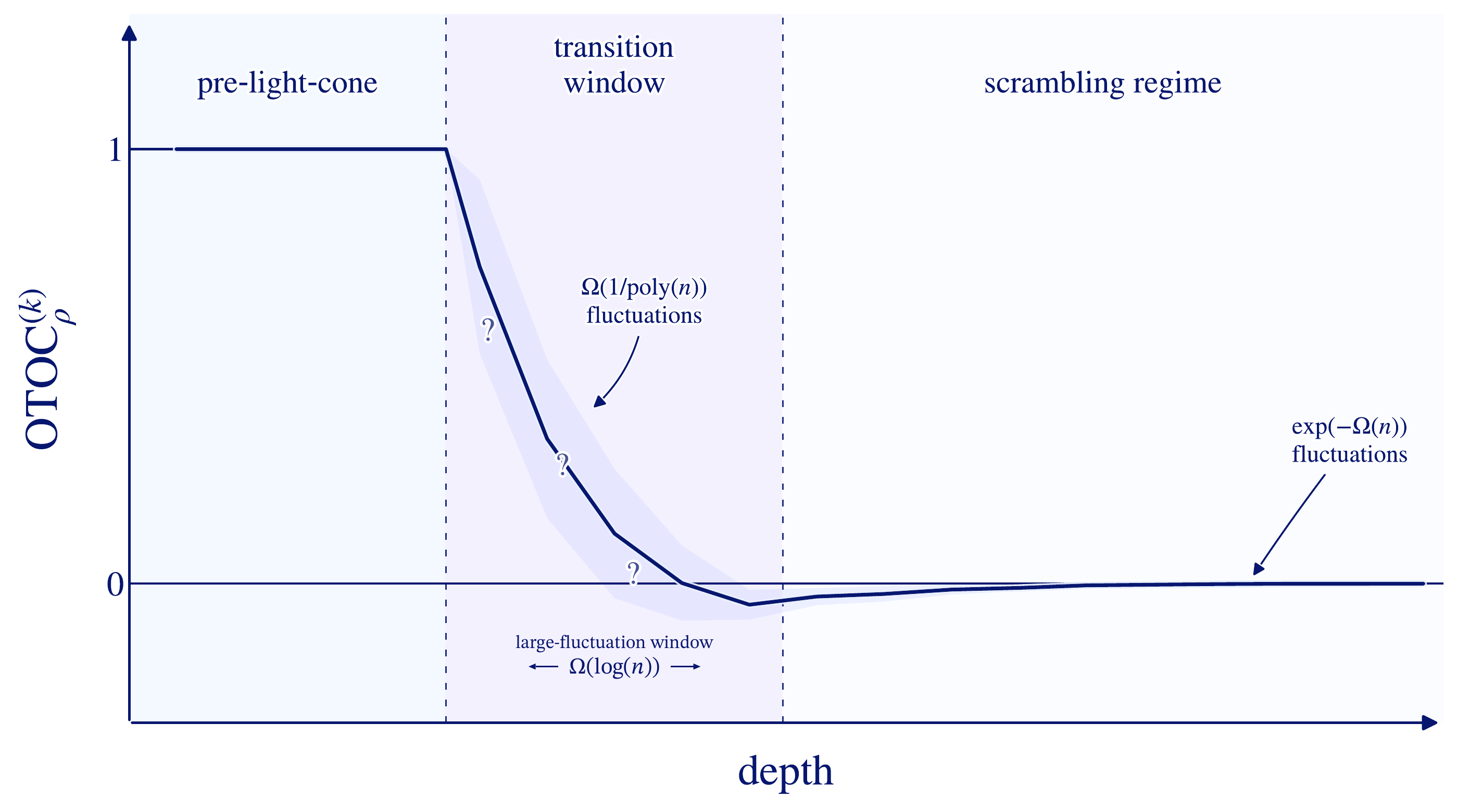}
    \caption{Representative depth dependence of
\(\mathrm{OTOC}^{(k)}_{\rho}\) in a local random circuit.
Before the Heisenberg light cone reaches the probe, the OTOC equals
\(1\) for every circuit realization. At intermediate depths,
circuit-to-circuit fluctuations become visible, before the observable
approaches its deep-circuit behavior.
The data show \(\mathrm{OTOC}^{(2)}_{\rho}\) for
\(\rho=\ketbra{0^n}{0^n}\) in a one-dimensional open-boundary
Haar-random brickwork circuit with \(n=18\) qubits,
\(B=Z_4\), and \(M=Z_{14}\).
For general fixed \(k\), we locate a growing range of $\Omega(\log n)$ consecutive depths with large fluctuations between the light-cone and design regimes,
without determining the full transition width; this interval can be
scanned efficiently in a quantum experiment. For the one-dimensional
endpoint \(\mathrm{OTOC}^{(1)}\), we characterize the transition completely:
it is centered at depth \(5n/3\) and has diffusive width
\(\Theta(\sqrt{n})\).}
    \label{fig:otoc-schematic}
\end{figure}


\section{Fluctuations of \texorpdfstring{\(\mathrm{OTOC}^{(k)}\)}{OTOC(k)} in local random circuits}
We begin by outlining the main ideas behind our result showing the absence of strong concentration for fixed-order OTOCs. 

Throughout this section, the butterfly operator \(B\) and measurement operator \(M\) are local Pauli observables with constant-size support, separated by a distance that scales with the system size in the interaction graph of the circuit architecture. For example, in the one-dimensional setting considered below, their separation is \(\Theta(n)\). We consider local random circuits composed of independent random gates. Such circuits provide standard models of operator spreading and OTOC dynamics~\cite{nahum2018operator,keyserlingk2018operator,harrow2019separation}, and their convergence toward approximate unitary designs is well established across a broad range of architectures and gate ensembles~\cite{brandao2016local,haferkamp2022random,harrowmehraban2023approximate,mittal2023local,yada2026nonhaar}. 


\begin{theorem}[Informal]
\label{thm:informalOTOCk}
For every fixed OTOC order \(k\), local random circuits exhibit inverse-polynomial circuit-to-circuit fluctuations over a nontrivial range of depths.

More precisely, for broad classes of local random circuits in any fixed spatial dimension and for arbitrary initial states, there exists a depth regime between the light-cone and deep-scrambling limits in which
\[
\Var_U\!\left[\mathrm{OTOC}^{(k)}_\rho(U)\right]
\geq
\frac{1}{\operatorname{poly}(n)}.
\]
The result applies, in particular, to Haar-random local circuits and to suitable hardware-motivated local gate ensembles.
\end{theorem}
A complete proof is given in Sec.~\ref{app:fluct} of the SM; see in particular Theorem~\ref{thm:fixed-k-transition-window-fluctuation-bound}. The proof sketched below is based on a simple slope-to-variance principle. If the ensemble-averaged OTOC changes by an amount \(\Delta\) over a depth interval of width \(W\), then there exists a depth in that interval, together with a fixed number of subsequent depths, where the variance is at least \(\Omega(\Delta^2/W^2)\). Thus, a constant change of the mean over a polynomial-width transition already implies inverse-polynomial circuit-to-circuit fluctuations. 

\begin{proof}[Proof sketch]
Let \(d_{\mathrm{lc}}\) denote the last depth before the Heisenberg light cone of \(B\) reaches the support of \(M\). For every \(d\leq d_{\mathrm{lc}}\), the operators \(U_d^\dagger B U_d\) and \(M\) commute, and hence \(\mathrm{OTOC}^{(k)}_\rho(U_d)=1\) for every circuit instance and every input state.
Now suppose that the ensemble-averaged OTOC changes by \(\Delta\) over a depth interval \([a,b]\) of width \(W=b-a\). Since the total change is \(\Delta\), at least one layer must change the mean by order \(\Delta/W\). Within that layer, the change can be decomposed into the contributions of its local gates, so that at least one constant-size random gate or block produces a non-negligible change in the corresponding conditional mean. The local reverse-variance bound of Appendix~\ref{app:local-reverse-variance} then converts this mean response into conditional variance of order \(\Delta^2/W^2\). Averaging over the remaining circuit randomness and using the law of total variance gives circuit-to-circuit variance of the same order at some depth \(d_\star\in\{a+1,\ldots,b\}\). The forward-persistence lemma of Appendix~\ref{app:fluct} shows that this lower bound persists for a fixed number of subsequent depths (at least $\Omega(\log(n))$ depths).
It remains only to certify that the ensemble mean changes by a non-negligible amount. Before the light cone reaches \(M\), the averaged OTOC is exactly \(1\). At sufficiently large depth, convergence to the relevant unitary-design moment implies that the averaged OTOC approaches its Haar value. Whenever this Haar-averaged value differs from \(1\) by a constant, we have \(\Delta=\Omega(1)\). Since the corresponding design depth is polynomial in \(n\) for the circuit families considered here~\cite{brandao2016local,haferkamp2022random,harrowmehraban2023approximate,yada2026nonhaar,oszmaniec2022epsilon}, the claimed inverse-polynomial variance lower bound follows.
\end{proof}

This existential localization in depth is not an obstruction for a quantum experiment: since the relevant interval contains only polynomially many candidate depths, one can efficiently scan through them and identify a depth exhibiting inverse-polynomial circuit-to-circuit fluctuations.
The role of design convergence is only to certify that the averaged OTOC changes by a constant amount. The argument itself is more general: any control of the mean profile immediately translates into a fluctuation bound through the slope-to-variance relation above. In particular, sharper information about the transition window leads directly to stronger bounds. This is precisely what happens and show below for \(\mathrm{OTOC}^{(1)}\) in one dimension, where we can characterize the transition in much greater detail.

The theorem is not restricted to Haar-random two-qubit gates. It also applies to hardware-motivated ensembles constructed from a fixed entangling gate---for example CNOT, CZ, iSWAP, or generic fSim gates---interleaved with Haar random single-qubit rotations~\cite{brylinski2002universal,bremner2002practical}. The proof requires only sufficient local moment mixing in Appendix~\ref{app:moment-control-design-depth} and the local reverse-variance condition of Appendix~\ref{app:local-reverse-variance}; both are properties of constant-size local blocks and are independent of the total number of qubits. This makes the result directly relevant to the gate architecture used in the Google Quantum AI OTOC experiments, where the two-dimensional random circuits are built from fixed iSWAP-like entangling gates interleaved with randomized single-qubit rotations~\cite{abanin2025constructiveinterferenceedgequantum,king2025simplifiedversionquantumotoc2}. 


\section{Are there actually many gates that significantly influence the OTOC? A one-dimensional analysis}
\label{subsec:main-otoc1-finite-size}

The general result above establishes inverse-polynomial circuit-to-circuit
OTOC fluctuations at suitable depths. By itself, however, it leaves open the
possibility that these fluctuations are generated by only a small number of
exceptional gates. If that were the case, one might hope to simulate the OTOC
efficiently by identifying those gates and averaging over the rest.

We now rule out this possibility in the concrete setting of the first-order OTOC in
a one-dimensional Haar-random brickwork circuit. We show that the
fluctuations have a collective origin, with \(\Theta(n^{3/2})\) gates making
inverse-polynomial contributions. Interestingly, these gates are organized
within a region whose width in each circuit layer is only \(O(\sqrt n)\).
This structure will also lead to the subexponential classical simulation
algorithm described below.

Consider an open chain of \(n\) qubits and a depth-\(d\) brickwork circuit,
as depicted in Fig.~\ref{fig:otoc1-gate-sensitivity-lens}. Each two-qubit
gate is sampled independently from \(\Haar(\mathrm U(4))\), and odd and even
layers alternate between the two nearest-neighbor matchings. We place the
butterfly operator on the first qubit and the probe on the last, taking
\(B=Z_1\) and \(M=Z_n\). We refer to this choice as the
\emph{endpoint geometry}.

Before the Heisenberg light cone reaches the last qubit,
\(\mathrm{OTOC}^{(1)}_\rho(U_d)=1\) for every circuit realization, and at larger
depths the value begins to depend significantly on the particular circuit. We first
identify the depth range in which this happens and then study the size and
origin of the resulting circuit-to-circuit fluctuations.

\begin{figure*}[t]
    \centering
    \makebox[\textwidth][c]{%
        \begin{minipage}[c]{0.5\textwidth}
            \centering
            \subfloat[\label{fig:otoc1-gate-sensitivity-lens-a}]{%
                \includegraphics[width=\linewidth]
                {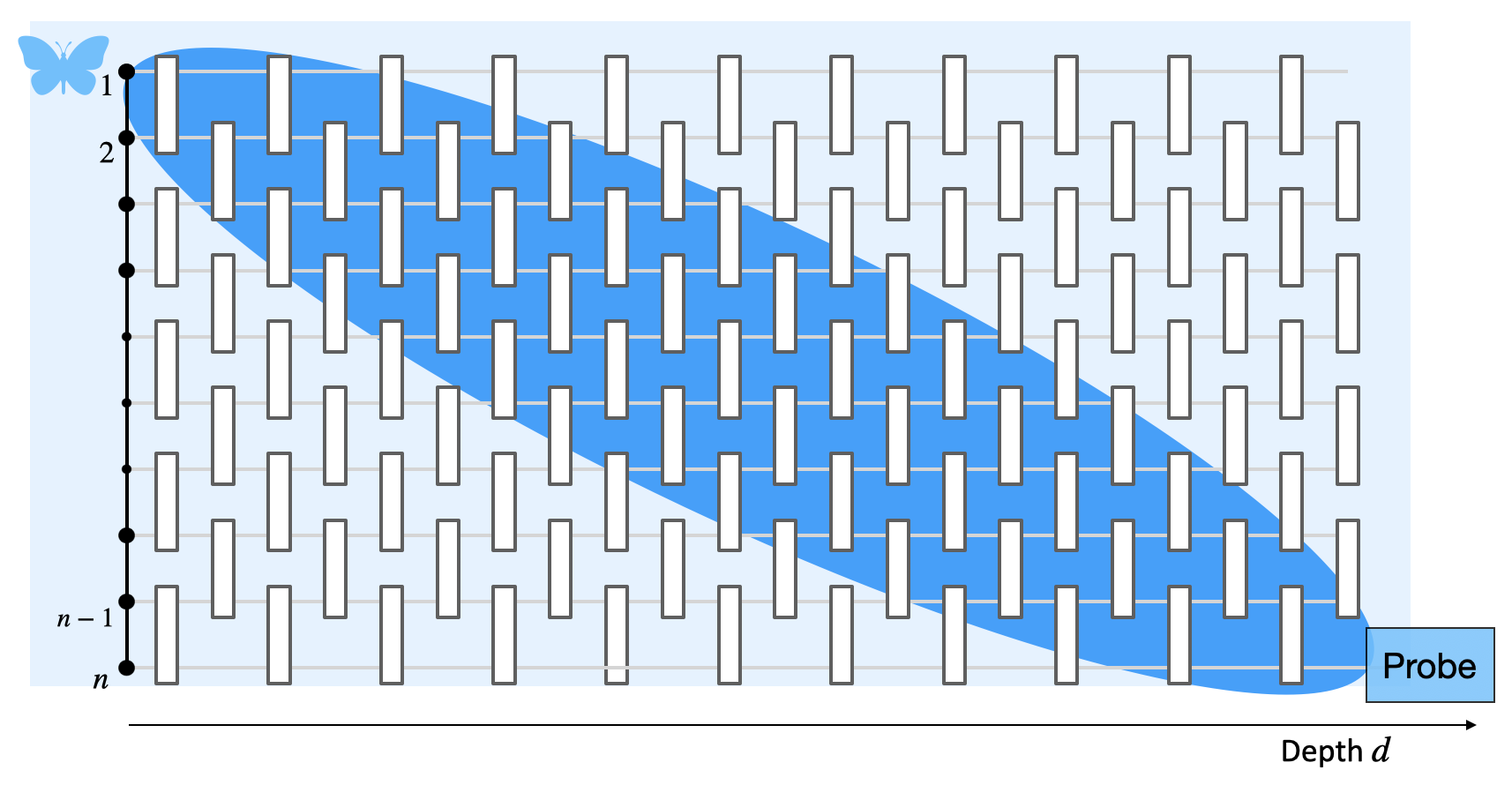}%
            }
        \end{minipage}%
        \hspace{0.0025\textwidth}%
        \begin{minipage}[c]{0.54\textwidth}
            \centering
            \subfloat[\label{fig:otoc1-gate-sensitivity-lens-b}]{%
                \includegraphics[width=\linewidth]
                {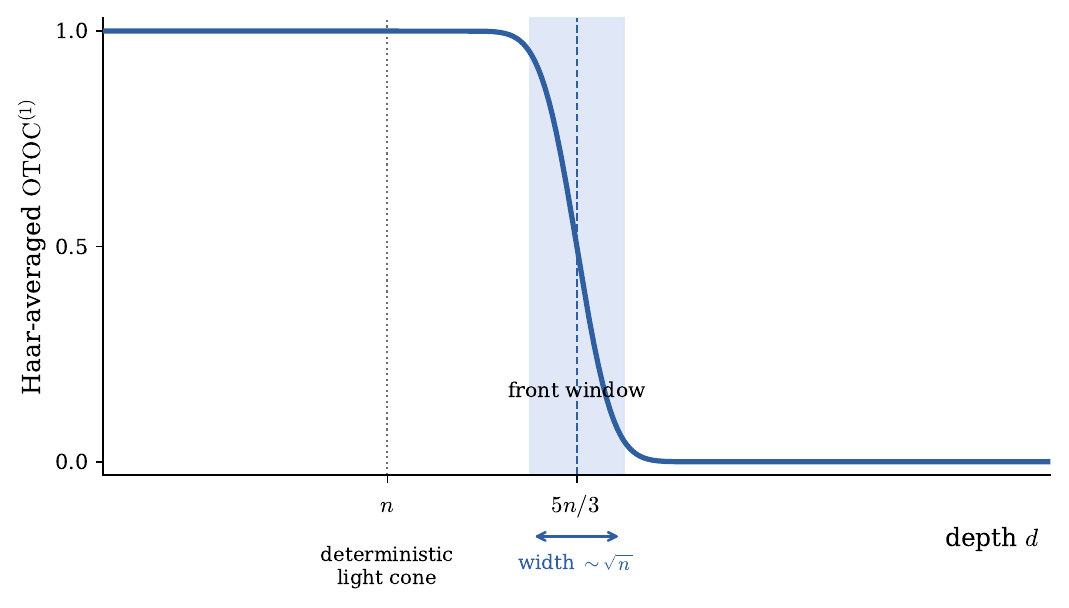}%
            }
        \end{minipage}%
    }

    \caption{
    First-order OTOC in the one-dimensional endpoint geometry.
    \textbf{(a)} Circuit geometry and the eye-shaped region of influential
    gates, with longitudinal extent \(\Theta(n)\) and maximal width
    \(\Theta(\sqrt n)\).
    \textbf{(b)} Circuit-averaged OTOC as a function of depth. The solid curve
    shows the exact circuit average and the dashed curve the Gaussian
    approximation of Eq.~\eqref{eq:main-otoc1-gaussian-front}. The shaded
    region marks the \(O(\sqrt n)\) depth window around \(d_\star=5n/3\).
    }
    \label{fig:otoc1-gate-sensitivity-lens}
\end{figure*}

\subsection{Over which depth interval is the OTOC instance-dependent, and which gates contribute to its fluctuations?}

We first determine the depth interval over which the OTOC changes. In a
one-dimensional Haar-random brickwork circuit, the circuit-averaged
Heisenberg evolution of \(Z_1\) can be mapped to a classical biased random
walk describing the right endpoint of the Pauli strings appearing in the
evolved operator~\cite{nahum2018operator,keyserlingk2018operator}.

Starting from the first qubit, this endpoint propagates toward the last qubit
with average velocity \(3/5\), and therefore typically reaches the probe at
depth \(d_\star=5n/3\). Its position broadens diffusively, leading to
fluctuations of order \(\sqrt n\) in the arrival depth. Consequently, the
circuit-averaged OTOC changes from near \(1\) to near \(0\) over a depth
window of the same width. More precisely, uniformly over all input states
\(\rho\),
\begin{align}
\label{eq:main-otoc1-gaussian-front}
\left|
    \E_{U_d}\mathrm{OTOC}^{(1)}_\rho(U_d)
    -
    \phi\!\left(
        \frac{d-\frac53 n}{\frac43\sqrt d}
    \right)
\right|
&\leq
\frac{5}{\sqrt n},
\end{align}
where
\(\phi(x)\coloneqq(2\pi)^{-1/2}\int_x^\infty e^{-s^2/2}\,ds\) is the cumulative Gaussian distribution.
Thus the relevant front window is
\(\lvert d-5n/3\rvert=O(\sqrt n)\), as shown in
Fig.~\ref{fig:otoc1-gate-sensitivity-lens}(b). The complete random-walk
mapping, finite-size analysis, and proof of
Eq.~\eqref{eq:main-otoc1-gaussian-front} are given in
Sec.~\ref{subsec:endpoint-otoc1-exact-formula} of the SM.

We next ask whether the OTOC remains instance-dependent throughout this
window, and which gates are responsible for its fluctuations. In this
one-dimensional setting, both questions can be answered sharply: the
variance remains inverse-polynomial throughout the front window, and its
fluctuations arise collectively from a macroscopic number of gates.

\begin{theorem}[Macroscopic origin of the OTOC fluctuations (informal)]
\label{thm:main-otoc1-variance-front-window}
Consider a one-dimensional Haar-random brickwork circuit on \(n\) qubits at
a depth satisfying
\(\lvert d-5n/3\rvert\leq c\sqrt n\), for any fixed \(c>0\).
Then, uniformly over all input states \(\rho\),
\begin{align}
    \Var_{U_d}\!\left[
        \mathrm{OTOC}^{(1)}_\rho(U_d)
    \right]
    &=
    \Omega(n^{-1/2}).
\end{align}
Moreover, there are \(\Theta(n^{3/2})\) gates whose individual
contribution to the variance is at least inverse-polynomial.
\end{theorem}
The proof is given in the SM; see Theorem~\ref{thm:endpoint-otoc1-variance-front-window} for the variance bound and Proposition~\ref{lem:otoc1-conditional-mean-variance-light-cone} for the \(\Theta(n^{3/2})\) individual gate contributions.

These \(\Theta(n^{3/2})\) influential gates occupy the blue region shown in Fig.~\ref{fig:otoc1-gate-sensitivity-lens}(a), which has the shape of an eye; for this reason, we will refer to it as the \emph{eye-shaped} region. This geometry has a simple
interpretation in terms of operator spreading. For a gate to appreciably
affect the final OTOC, the Heisenberg-evolved butterfly operator must first
reach that gate, and the change produced there must subsequently have enough
time to reach the probe by the final circuit depth. Thus an influential gate
must lie in the overlap between the region reached by the butterfly and the
region from which the probe can still be reached.
In the biased-random-walk picture, both propagation fronts broaden
diffusively around their ballistic trajectories. Their overlap therefore
extends over \(\Theta(n)\) circuit layers, while its maximal spatial width is
\(\Theta(\sqrt n)\). This produces the characteristic eye-shaped region
containing
\(\Theta(n)\times\Theta(\sqrt n)=\Theta(n^{3/2})\) gates. A closely related eye-shaped region was observed numerically through tensor-network simulations in Ref.~\cite{bermejo2026tensornetworks}; here this geometry emerges from our
rigorous gate-influence analysis.

More quantitatively, we identify \(\Theta(n^{3/2})\) gates whose individual
contributions to the variance are each \(\Omega(n^{-2})\). These
contributions are orthogonal and therefore add, yielding the total lower
bound
\(\Omega(n^{-1/2})\). The precise definition of gate influence, the
construction of the eye-shaped region, and the complete proof are deferred
to the SM.

Thus, the OTOC fluctuations arise collectively from a macroscopic region of
the circuit, rather than being controlled by only a few exceptional gates.
We stress, however, that this stronger gate-level characterization is
specific to \(\mathrm{OTOC}^{(1)}\). Our general higher-order OTOC theorem
establishes non-concentration much more broadly, but does not by itself show
that a macroscopic number of gates contribute to the fluctuations. The
gate-level analysis for \(\mathrm{OTOC}^{(1)}\) relies crucially on the fact
that both the ensemble mean and the relevant single-gate conditional means
are analytically tractable. For \(\mathrm{OTOC}^{(2)}\) and higher orders,
comparable analytical control for the conditional means is currently unavailable, even for the
ensemble mean, so the same argument does not directly extend.

\subsection{Application to improved classical simulation}
\label{subsubsec:classical-simulation-main}

Although a macroscopic number of gates contributes to the variance, these
gates occupy a thin spacetime region rather than the entire circuit. This
suggests that the cost of classical simulation may be controlled by the
width of this region, rather than by its total spacetime volume.

For the infinite-temperature \(\mathrm{OTOC}^{(1)}\) at the critical depth, we make this intuition rigorous in Theorem~\ref{thm:endpoint-sim-otoc1-runtime} of the SM. We prove that the
influence of gates outside the eye decays with a Gaussian tail. To achieve
additive error \(\epsilon\) and failure probability \(\delta\), the algorithm
retains a slightly enlarged eye of width
\(O(\sqrt{n\log(n/(\epsilon\delta))})\), conditions on the gates inside this
region, and averages over the Haar randomness outside it. The resulting conditional expectation can be estimated classically at a cost exponential
only in the maximal width of the retained region.

For a circuit drawn from the Haar-random brickwork ensemble, the algorithm
estimates the infinite-temperature OTOC to additive error \(\epsilon\) with
probability at least \(1-\delta\), where the probability is over both the
random circuit and the internal randomness of the algorithm. For
inverse-polynomial accuracy and failure probability, the runtime is
\begin{align}
    \mathrm{time}
    &=
    \operatorname{poly}(n)\,
    2^{O(\sqrt{n\log n})}.
\end{align}
This is subexponential in \(n\). To the best of our knowledge, it is the
first provable asymptotic improvement over brute-force classical simulation
for this average-case OTOC-estimation problem. Previous work has explored
approximate Monte-Carlo and tensor-network approaches
~\cite{
abanin2025constructiveinterferenceedgequantum,
bermejo2026tensornetworks,XuSwingle2020},
but these methods do not provide provable runtime guarantees. Our result applies
specifically to the infinite-temperature first-order OTOC; extending such
guarantees to state-dependent or higher-order OTOCs remains open.


\section{Conclusions and open questions}
\label{sec:discussion}
Our results reveal an intermediate regime of random quantum dynamics in which
information has already propagated across the system, while the observable
still retains a measurable dependence on the particular circuit realization.
Before the butterfly reaches the probe, the OTOC is trivial; at sufficiently
large depths, concentration suppresses circuit-to-circuit variation, leaving
the interesting high-signal regime between these two limits. This is precisely
the regime highlighted experimentally by Google Quantum AI, where significant
instance-to-instance OTOC fluctuations were observed
~\cite{
abanin2025constructiveinterferenceedgequantum,
king2025simplifiedversionquantumotoc2}.
By proving that these fluctuations remain at least inverse-polynomial for
broad classes of local random circuits, our results provide a rigorous
foundation for a central ingredient of this verifiable quantum-advantage
proposal.

Our one-dimensional analysis of \(\mathrm{OTOC}^{(1)}\) gives a more detailed
picture of this regime. Throughout the high-fluctuation window around the
linear depth \(d_\star=5n/3\), macroscopically many gates---\(\Theta(n^{3/2})\)
in total---contribute to the fluctuations, ruling out the possibility that
the large variance is caused by only a few exceptional gates. At the same
time, these gates are confined to a spacetime region of width only
\(O(\sqrt n)\) in each layer. A closely related eye-shaped geometry appears
in the tensor-network analysis of Ref.~\cite{bermejo2026tensornetworks},
where numerical evidence indicates that the corresponding region is largely
incompressible. Our results complement these findings by rigorously
establishing, for Haar-random brickwork circuits, the eye-shaped localization
of the influential gates. This localization also has a direct computational
consequence: for the infinite-temperature \(\mathrm{OTOC}^{(1)}\), the same
structure that supports sizeable circuit-to-circuit fluctuations enables a
subexponential classical simulation.

Several natural questions remain open. First, it is not clear whether our
subexponential algorithm for the infinite-temperature
\(\mathrm{OTOC}^{(1)}\) can be substantially improved. The cached Monte
Carlo method introduced in the Google experiment performs remarkably well
for finite-size first-order OTOCs~\cite{abanin2025constructiveinterferenceedgequantum}, but a rigorous
asymptotic analysis of its accuracy and runtime is still lacking. It would
be particularly interesting to determine whether this or related ideas can
lead to a quasipolynomial, or even polynomial, classical algorithm.

A related basic question is to determine the actual scale of the OTOC
fluctuations. Our results provide lower bounds on the variance, but no
corresponding nontrivial upper bounds. In particular, it remains open whether
the \(\Omega(n^{-1/2})\) scaling obtained for
\(\mathrm{OTOC}^{(1)}\) in one dimension is tight. Matching, or even
nontrivial, upper bounds would give a much sharper description of the
fluctuation regime and could reveal how the contributions of different gates
combine and cancel, potentially exposing further structure useful for
classical simulation.

A second direction is to go beyond single-gate influence. Our analysis
isolates the contribution of individual gates through conditional means, but
one can similarly study irreducible correlations between pairs or larger
sets of gates. Understanding this hierarchy could reveal how much of the
OTOC is captured by low-order gate dependencies and whether truncating it can
lead to improved classical simulation algorithms.

The corresponding questions for \(\mathrm{OTOC}^{(2)}\) and higher orders
appear substantially more challenging. Even the ensemble mean does not
currently admit an analytic treatment comparable to
\(\mathrm{OTOC}^{(1)}\) in the regimes of interest, and no efficient
classical procedure for evaluating it is known. Moreover, the classical
heuristics tested in the Google experiment perform substantially worse at
reproducing instance-to-instance fluctuations of higher-order OTOCs
~\cite{abanin2025constructiveinterferenceedgequantum}. Developing analytic
control of the ensemble mean, single-gate conditional means, and higher-order
gate correlations would shed light both on the origin of these fluctuations
and on whether the localization-based simulation picture extends beyond
\(\mathrm{OTOC}^{(1)}\).

More broadly, the persistence of an instance-dependent signal in deep random
circuits suggests that OTOC-based observables may provide useful objectives
in regimes where more conventional variational cost functions suffer from
concentration~\cite{McCleanEtAl2018}. Whether this can be exploited in
variational algorithms, quantum learning, or other applications remains an
interesting direction for future work.

\begin{acknowledgments}
This work was done while A.A.M. and F.A.M. were Student Researchers at Google Quantum AI. We thank Amira Abbas, Ryan Babbush, David Gosset, Jeongwan Haah,
Kostyantyn Kechedzhi, Robin Kothari, Laura Lewis, Tony Metger,
Thomas Schuster, Vadim Smelyanskiy, Rolando D. Somma, Aaron Szasz,
Guifré Vidal, Benjamin Villalonga, Brayden Ware, and Ronald de Wolf
for helpful discussions and feedback on this work.
\end{acknowledgments}

\emph{AI disclosure.}
GPT-5.5 was used extensively during the research process to explore and refine proof ideas and to perform numerical experiments that guided intuition about what to prove next. The authors wrote the manuscript and take full responsibility for its content.
GPT-6 Astra (Ultra) assisted in developing Lean formalizations of our main theorems. The Lean code is available in the accompanying \href{https://gist.github.com/AntMele/b4b8dc772444c70b73069d44e0047605}{GitHub repository}.

\bibliography{ref.bib}

\clearpage
\onecolumngrid
\begin{center}
\vspace*{\baselineskip}
{\Large\textbf{Supplemental Material}}
\end{center}

\setcounter{tocdepth}{2} 
\tableofcontents
\vspace*{2\baselineskip}

This Supplemental Material is organized as follows. Section~\ref{sec:appendixPRELIM} introduces the local circuit architectures, light cones, and basic OTOC properties. Section~\ref{app:fluct} proves the fluctuation bound for fixed-order OTOCs, combining design convergence with local reverse-variance bounds. Section~\ref{subsec:endpoint-otoc1-setup} studies the first-order OTOC in the one-dimensional endpoint geometry: we derive its ensemble mean and Gaussian approximation, prove the variance lower bound, and characterize the individual gate contributions. Finally, Section~\ref{sec:eye-reduction-classical-simulation} proves the subexponential classical simulation algorithm for the infinite-temperature endpoint \(\mathrm{OTOC}^{(1)}\) in one dimension.

\section{Preliminaries and technical tools}
\label{sec:appendixPRELIM}



We write \([n]\coloneqq \{1,\ldots,n\}\).  The Hilbert space of \(n\) qubits is
\(\mathcal H_n\coloneqq (\mathbb C^2)^{\otimes n}\), and its dimension is \(d_n\coloneqq 2^n\).  For
\(S\subseteq[n]\), we write \(\mathcal H_S\coloneqq \bigotimes_{i\in S}\mathbb C^2\) and
\(S^c\coloneqq [n]\setminus S\).  Operators supported on \(S\) are embedded into the full system by the
convention \(O_S\equiv O_S\otimes I_{S^c}\).  We write \(\mathcal B(\mathcal H)\) for the algebra
of linear operators on \(\mathcal H\), and \(I_S\) for the identity on \(\mathcal H_S\).

For \(O\in\mathcal B(\mathcal H_n)\), \(\operatorname{supp}(O)\) denotes the smallest subset
\(S\subseteq[n]\) such that \(O=O_S\otimes I_{S^c}\) for some
\(O_S\in\mathcal B(\mathcal H_S)\).  All local observables appearing below have support size
independent of \(n\).  We use the operator norm \(\|O\|_\infty\), the trace norm \(\|O\|_1\), and
the Hilbert--Schmidt norm \(\|O\|_2\), defined by
\begin{align}
    \|O\|_\infty
    &\coloneqq 
    \sup_{\|\psi\|=1}\|O|\psi\rangle\|,
    &
    \|O\|_1
    &\coloneqq 
    \Tr\sqrt{O^\dagger O},
    &
    \|O\|_2
    &\coloneqq 
    \sqrt{\Tr[O^\dagger O]}.
    \label{eq:norm-conventions}
\end{align}
Expectations over random circuits and random gates are denoted by \(\E\).  When it is useful to
specify the underlying ensemble, we write, for example, \(\E_{U_d}\), \(\E_{g\sim\nu}\), or
\(\E_{U\sim\mathrm{Haar}}\).  For a complex-valued random variable \(X\), we use
\begin{align}
    \Var(X)
    &\coloneqq 
    \E[|X-\E X|^2]
    =
    \E[|X|^2]-|\E X|^2.
    \label{eq:complex-variance-convention}
\end{align}
If \(Y\) is a random object on which we condition, such as a partial circuit or a local block of
gates, we write
\begin{align}
    \Var(X\mid Y)
    &\coloneqq 
    \E\!\left[
        |X-\E[X\mid Y]|^2
        \,\middle|\,
        Y
    \right].
    \label{eq:conditional-variance-definition}
\end{align}
With this convention, the law of total variance gives
\begin{align}
    \Var(X)
    &=
    \E\!\left[
        \Var(X\mid Y)
    \right]
    +
    \Var\!\left(
        \E[X\mid Y]
    \right).
    \label{eq:law-total-variance-complex}
\end{align}
In particular, \(\Var(X)\ge\E[\Var(X\mid Y)]\).

We denote by \(\mathcal P_n\coloneqq \{I,X,Y,Z\}^{\otimes n}\) the set of \(n\)-qubit Pauli
strings.  Thus every Pauli observable \(P\in\mathcal P_n\) satisfies
\(P^\dagger=P\) and \(P^2=I\).  Throughout this appendix, \(B\) denotes the butterfly Pauli
observable and \(M\) denotes the measured Pauli observable.  Unless stated otherwise,
\(B,M\in\mathcal P_n\) are nonidentity Pauli observables with
\begin{align}
    |\operatorname{supp}(B)|
    &=
    O(1),
    &
    |\operatorname{supp}(M)|
    &=
    O(1).
    \label{eq:BM-local-support}
\end{align}

For a graph \(G=([n],E)\), \(\dist_G(i,j)\) denotes the minimum number of edges in a path from
\(i\) to \(j\).  For \(S,T\subseteq[n]\), we use
\(\dist_G(i,S)\coloneqq \min_{j\in S}\dist_G(i,j)\) and
\(\dist_G(S,T)\coloneqq \min_{i\in S,j\in T}\dist_G(i,j)\).  This notation will be applied to the interaction
graph of the local circuit architecture.

Asymptotic notation with a subscript, such as \(O_k(\cdot)\), \(\Omega_k(\cdot)\), or
\(\Theta_k(\cdot)\), means that the implicit constants may depend on the fixed OTOC order \(k\),
the fixed local gate ensemble, and fixed architectural parameters, but not on \(n\).

\subsection{Local circuit architectures and deterministic light cones}
\label{app:local-architectures}

Let \(G_n=([n],E_n)\) be an interaction graph of bounded degree, meaning that every vertex has at
most a constant number of neighbours, independently of \(n\).  Equivalently, each qubit can interact
directly with only \(O(1)\) other qubits.  The graph distance between two vertices is the minimum
number of edges in a path connecting them.

A local architecture is a sequence of layers \(\mathcal L_1,\mathcal L_2,\ldots\), where each layer
\(\mathcal L_j\) is a collection of mutually disjoint patches \(S\subseteq[n]\).  Mutually disjoint means
that no qubit belongs to two different patches in the same layer.  Each patch has size \(O(1)\), and
all qubits in a patch are within graph distance \(O(1)\) of each other in \(G_n\).  On each
\(S\in\mathcal L_j\), we apply an independent random unitary
\(g_{j,S}\in\mathrm U(\mathcal H_S)\), drawn from a local distribution \(\nu_{j,S}\).  The layer unitary is
\begin{align}
    L_j
    &\coloneqq 
    \bigotimes_{S\in\mathcal L_j}g_{j,S},
    \label{eq:layer-unitary}
\end{align}
with the identity on all qubits not touched by the layer.  We allow the distributions
\(\nu_{j,S}\) to depend on \(j\) and \(S\), but all local gates are independent unless explicitly stated
otherwise.

We use the Heisenberg ordering
\begin{align}
    U_d
    &\coloneqq 
    L_dL_{d-1}\cdots L_1.
    \label{eq:Ud-definition}
\end{align}
Thus
\begin{align}
    U_d^\dagger B U_d
    &=
    U_{d-1}^\dagger L_d^\dagger B L_d U_{d-1}.
    \label{eq:new-layer-adjacent}
\end{align}
With this convention, the newly added layer \(L_d\) is the layer adjacent to \(B\) in the Heisenberg
evolution.

For \(S\subseteq[n]\), define the deterministic Heisenberg light cone \(\mathsf{LC}_d(S)\) to be the
smallest subset of \([n]\), determined only by the architecture, such that for every operator \(O\)
with \(\operatorname{supp}(O)\subseteq S\), and for every choice of local gates compatible with the
architecture,
\begin{align}
    \operatorname{supp}(U_d^\dagger O U_d)
    &\subseteq
    \mathsf{LC}_d(S).
    \label{eq:light-cone-support}
\end{align}
This light cone is deterministic: it depends only on which gates are allowed by the architecture, not
on the sampled gate values.  It is therefore a worst-case light cone and does not use cancellations or
special algebraic properties of particular circuit instances.  In particular, $\operatorname{supp}(U_d^\dagger B U_d) \subseteq
    \mathsf{LC}_d(\operatorname{supp}B).$

The bounded-degree and locality assumptions imply finite propagation speed.  More precisely, there
exists a constant \(v_{\mathcal A}<\infty\), depending only on the architecture, such that
\(\mathsf{LC}_d(S)\) is contained in the set of vertices whose graph distance from \(S\) is at most
\(v_{\mathcal A}d\).  Hence, if \(|S|=O(1)\) and \(d=O(1)\), then
\begin{align}
    |\mathsf{LC}_d(S)|
    &=
    O_d(1),
    \label{eq:fixed-depth-light-cone}
\end{align}
with a constant independent of \(n\).

In all applications below, \(\operatorname{supp}B\) and
\(\operatorname{supp}M\) are initially disjoint and the deterministic
light cone of \(B\) eventually reaches \(\operatorname{supp}M\). We define
the light-cone depth by
\begin{align}
    d_{\rm lc}
    &\coloneqq
    \max\left\{
        d\ge0:
        \mathsf{LC}_d(\operatorname{supp}B)
        \cap
        \operatorname{supp}M
        =
        \varnothing
    \right\}.
\end{align}
Thus \(d_{\rm lc}\) is the last depth at which the deterministic
Heisenberg light cone of \(B\) is disjoint from \(M\). Equivalently, for every \(d\le d_{\rm lc}\), the Heisenberg light
cone of \(B\) has not reached \(M\).

For nearest-neighbour brickwork circuits in one dimension, if
\(\dist(\operatorname{supp}B,\operatorname{supp}M)=\Theta(n)\), then
\(d_{\rm lc}=\Theta(n)\).  For a standard nearest-neighbour circuit schedule on a \(D\)-dimensional grid of side length
\(\ell\), macroscopically separated \(B\) and \(M\) give
\(d_{\rm lc}=\Theta(\ell)\). In particular, if \(n=\ell^D\), then
\(d_{\rm lc}=\Theta(n^{1/D})\).

\subsection{Higher-order OTOCs}
\label{app:higher-order-otoc-definitions}

For a unitary \(U\in\mathrm U(\mathcal H_n)\), set
\(A_U\coloneqq U^\dagger B U\) and \(C_U\coloneqq A_UM=U^\dagger BUM\).
For a density matrix \(\rho\) and an integer \(k\ge1\), define
\begin{align}
    \mathrm{OTOC}^{(k)}_\rho(U)
    &\coloneqq
    \Tr\!\left[\rho\,C_U^{2k}\right]
    =
    \Tr\!\left[\rho\,(U^\dagger BUM)^{2k}\right].
\end{align}

Since \(B\) and \(M\) are Pauli observables, \(A_U\) and \(M\) are Hermitian unitaries. Hence \(C_U=A_UM\) is unitary and, for every density matrix
\(\rho\),
\begin{align}
    \left|\mathrm{OTOC}^{(k)}_\rho(U)\right|
    &\le
    \|C_U^{2k}\|_\infty\Tr\rho
    =
    1.
\end{align}

In general, \(C_U^{2k}\) need not be Hermitian, so
\(\mathrm{OTOC}^{(k)}_\rho(U)\) can be complex. Whenever
\(M\rho M=\rho\), however, it is real. Indeed,
\(C_U^\dagger=MC_UM\), and therefore
\((C_U^{2k})^\dagger=MC_U^{2k}M\). Thus
\begin{align}
    \overline{\mathrm{OTOC}^{(k)}_\rho(U)}
    &=
    \Tr\!\left[(C_U^{2k})^\dagger\rho\right]
    =
    \Tr\!\left[C_U^{2k}M\rho M\right]
    =
    \mathrm{OTOC}^{(k)}_\rho(U).
\end{align}
This includes the infinite-temperature state \(\rho=2^{-n}I\) and, more
generally, any state commuting with \(M\). For arbitrary \(\rho\), all
variance bounds use the complex variance convention of
Eq.~\eqref{eq:complex-variance-convention}.

As a function of \(U\), \(\mathrm{OTOC}^{(k)}_\rho(U)\) is a
degree-\((2k,2k)\) polynomial in the entries of \(U\) and \(\overline U\):
each of the \(2k\) factors \(U^\dagger B U\) contributes one copy of \(U\)
and one copy of \(U^\dagger\). Thus the relevant moment order for the
design statements below is \(2k\)~\cite{RobertsYoshida2017}.

\begin{lemma}[Pre-light-cone identity]
\label{lem:pre-light-cone-identity}
For every \(d\le d_{\rm lc}\), every circuit realization \(U_d\), every
\(k\ge1\), and every density matrix \(\rho\),
\(\mathrm{OTOC}^{(k)}_\rho(U_d)=1\).
\end{lemma}

\begin{proof}
Let \(A_d\coloneqq U_d^\dagger B U_d\). For \(d\le d_{\rm lc}\),
\(\operatorname{supp}(A_d)\cap\operatorname{supp}(M)=\varnothing\), hence
\([A_d,M]=0\). Since \(A_d^2=M^2=I\), we have
\((A_dM)^2=I\), and therefore
\(\mathrm{OTOC}^{(k)}_\rho(U_d)
=\Tr[\rho(A_dM)^{2k}]
=\Tr\rho=1\).
\end{proof}

Thus, before the deterministic Heisenberg light cone of \(B\) reaches
\(M\), the higher-order OTOC is exactly equal to \(1\) for every circuit
instance. Circuit-to-circuit fluctuations can therefore arise only after
this pre-light-cone regime.

We now record two simple locality properties that clarify the spacetime region on
which an OTOC can depend. Let \(G=(V,E)\) be the interaction graph and
\(U=U_T\cdots U_1\) a depth-\(T\) local circuit. For Pauli observables \(B,M\),
define
\begin{align}
    F^{(k)}_{\mathrm{tr},T}(U)
    &\coloneqq
    2^{-|V|}\Tr[(U^\dagger BUM)^{2k}],
    &
    F^{(k)}_{0,T}(U)
    &\coloneqq
    \langle 0^V|(U^\dagger BUM)^{2k}|0^V\rangle .
\end{align}

Let \(\mathsf{LC}^{-}_T(B)\) denote the backward light cone of \(B\), namely
the spacetime gates from which there is a directed path to
\(\operatorname{supp}(B)\) at the final time, and let
\(\mathsf{LC}^{+}_T(M)\) denote the forward light cone of \(M\). Their
intersection
\begin{align}
    \mathsf D_T(B,M)
    &\coloneqq
    \mathsf{LC}^{-}_T(B)\cap\mathsf{LC}^{+}_T(M)
\end{align}
is the causal diamond between \(M\) and \(B\).
The distinction between the two relevant regions is illustrated in
Fig.~\ref{fig:otoc-causal-reductions}.

\begin{figure}[t]
    \centering
    \includegraphics[width=\textwidth]{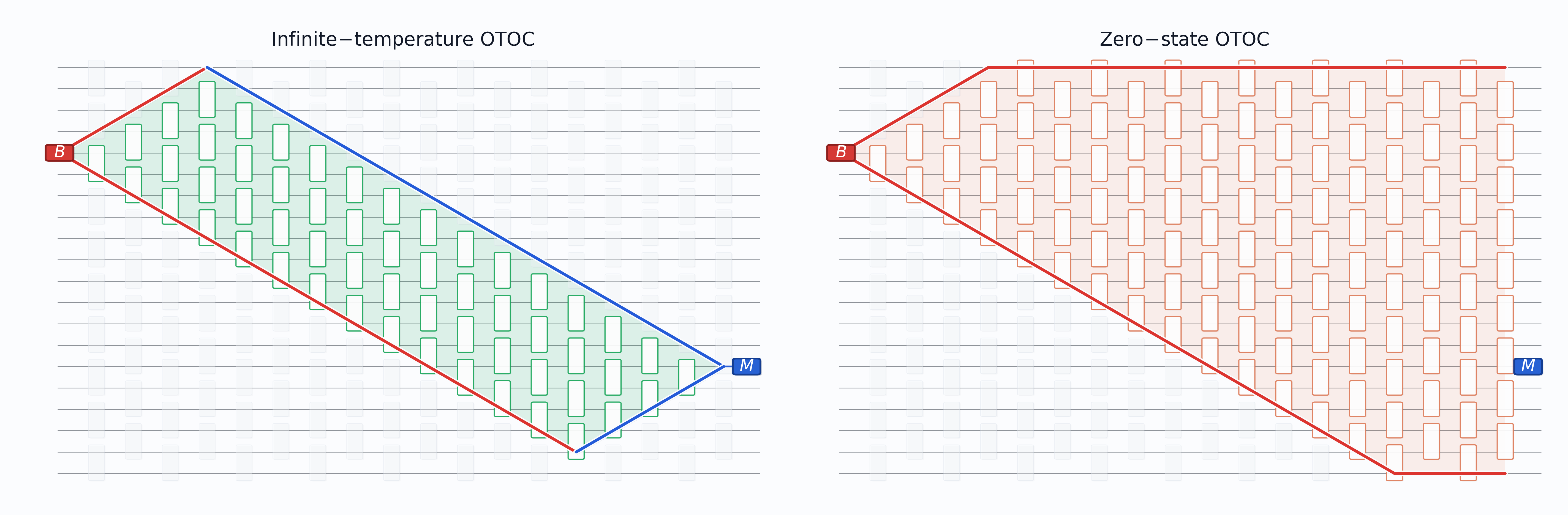}
    \caption{
    Causal light cones for OTOCs in one-dimensional brickwork circuits.
    For the infinite-temperature trace OTOC (left), only gates in the causal
    diamond \(\mathsf{LC}^{-}_T(B)\cap\mathsf{LC}^{+}_T(M)\) can affect the
    observable. For a fixed-state OTOC (right), cyclic trace cancellation is
    unavailable, and the exact deterministic reduction retains the full
    backward light cone \(\mathsf{LC}^{-}_T(B)\) of the butterfly operator.
    }
    \label{fig:otoc-causal-reductions}
\end{figure}

\begin{lemma}[Causal reductions for OTOCs]
\label{lem:otoc-causal-reductions}
The infinite-temperature OTOC
\(F^{(k)}_{\mathrm{tr},T}(U)
\coloneqq 2^{-|V|}\Tr[(U^\dagger BUM)^{2k}]\)
depends only on the gates in the causal diamond
\(\mathsf D_T(B,M)
=\mathsf{LC}^{-}_T(B)\cap\mathsf{LC}^{+}_T(M)\).
More generally, for any fixed input state \(\rho\),
\(\mathrm{OTOC}^{(k)}_\rho(U)\) depends only on the gates in the backward
light cone \(\mathsf{LC}^{-}_T(B)\).
\end{lemma}

\begin{proof}
Fix a gate \(W\) and write the circuit as \(U=U_{>}WU_{<}\), where
\(U_{<}\) and \(U_{>}\) contain the gates before and after \(W\),
respectively. Set \(B_W\coloneqq U_{>}^\dagger B U_{>}\) and
\(M_W\coloneqq U_{<} M U_{<}^\dagger\).

For the trace OTOC, unitary invariance of the trace gives
\(F^{(k)}_{\mathrm{tr},T}(U)
=2^{-|V|}\Tr[(W^\dagger B_W W M_W)^{2k}]\).
If \(W\notin\mathsf{LC}^{+}_T(M)\), then \([W,M_W]=0\), so
\((W^\dagger B_W W M_W)^{2k}
=W^\dagger(B_WM_W)^{2k}W\), and the trace is independent of \(W\).
If instead \(W\notin\mathsf{LC}^{-}_T(B)\), then
\([W,B_W]=0\), so \(W^\dagger B_WW=B_W\). Hence a gate can affect the
trace OTOC only if it belongs to both light cones, i.e. to
\(\mathsf D_T(B,M)\).

For a fixed input state \(\rho\), all circuit dependence enters through
\(U^\dagger B U\). If \(W\notin\mathsf{LC}^{-}_T(B)\), then
\([W,B_W]=0\), and therefore
\(U^\dagger B U
=U_{<}^\dagger W^\dagger B_WWU_{<}
=U_{<}^\dagger B_WU_{<}\),
which is independent of \(W\). Thus
\(\mathrm{OTOC}^{(k)}_\rho(U)\) depends only on the backward light cone of
\(B\).
\end{proof}

\subsection{\texorpdfstring{\(\mathrm{OTOC}^{(k)}\)}{OTOC(k)} properties in circuit ensembles}

We say that a unitary ensemble is right-Pauli invariant if \(U\) and \(US\)
have the same distribution for every \(S\in\mathcal P_n\).

\begin{lemma}[State-independence under right-Pauli invariance]
\label{lem:state-independence-higher-otoc}
Let the circuit ensemble be right-Pauli invariant and let
\(B,M\in\mathcal P_n\). Then, for every \(k\ge1\), there exists
\(\alpha_k\in\mathbb C\) such that
\begin{align}
    \E_U (U^\dagger BUM)^{2k}
    &=
    \alpha_k I .
\end{align}
Consequently,
\(\E_U\,\mathrm{OTOC}^{(k)}_\rho(U)=\alpha_k\)
for every density matrix \(\rho\); in particular, the ensemble-averaged OTOC
is independent of the input state.
\end{lemma}

\begin{proof}
Set \(A_U\coloneqq U^\dagger B U\) and
\(X\coloneqq\E_U(A_UM)^{2k}\). For any \(S\in\mathcal P_n\),
right-Pauli invariance gives
\(X=\E_U(S^\dagger A_USM)^{2k}\).
Since \(S\) and \(M\) either commute or anticommute and the power \(2k\)
is even,
\((S^\dagger A_USM)^{2k}=S^\dagger(A_UM)^{2k}S\).
Hence \(X=S^\dagger XS\) for every Pauli string \(S\).
The only operators commuting with all Pauli strings are scalar multiples
of the identity, so \(X=\alpha_k I\) for some \(\alpha_k\in\mathbb C\).
Therefore
\(\E_U\,\mathrm{OTOC}^{(k)}_\rho(U)
=\Tr(\rho X)=\alpha_k\)
for every density matrix \(\rho\).
\end{proof}

A sufficient condition for right-Pauli invariance is that the input-side random block of the circuit can
absorb arbitrary Pauli strings.  For example, this holds if that block contains independent local
random gates on patches covering all qubits, and each local gate distribution is invariant under right
multiplication by Pauli operators supported on the corresponding patch.  Haar-random local gates
satisfy this condition.  It also holds for ensembles in which fixed
entangling gates are interleaved with independent Haar-random single-qubit layers covering all
qubits.  By contrast, the lemma should not be applied unchanged to symmetry-preserving ensembles,
such as \(U(1)\)-symmetric circuits, because right multiplication by a generic Pauli string does not
preserve the allowed symmetry sector.

We now record that the global-Haar OTOC mean is exponentially
small.
\begin{proposition}[Haar mean of fixed-order OTOCs]
\label{prop:haar-otock-small}
Let \(D=2^n\), let \(U\sim\mathrm{Haar}(U(D))\), and let \(B,M\) be
nonidentity Hermitian Pauli strings. For every fixed \(k\ge1\), there exists
a constant \(C_k<\infty\), depending only on \(k\), such that, uniformly over
all density matrices \(\rho\),
\[
    \left|
    \mathbb E_U\,\mathrm{OTOC}^{(k)}_\rho(U)
    \right|
    \le C_kD^{-2}
    =C_k2^{-2n}.
\]
Moreover, the Haar mean is independent of \(\rho\). Consequently,
\(\left|1-\mathbb E_U\mathrm{OTOC}^{(k)}_\rho(U)\right|\ge1/2\)
for all sufficiently large \(n\).
\end{proposition}

\begin{proof}
Set \(r=2k\) and \(A_U=U^\dagger B U\). By
Lemma~\ref{lem:state-independence-higher-otoc},
\(\mathbb E_U(A_UM)^r=\alpha_k I\) for some scalar \(\alpha_k\). Hence
\(\mathbb E_U\mathrm{OTOC}^{(k)}_\rho(U)=\alpha_k\), independently of
\(\rho\), and
\(\alpha_k=D^{-1}\mathbb E_U\operatorname{Tr}[(A_UM)^r]\).
It therefore remains to show that \(|\alpha_k|=O_k(D^{-2})\).

Let \(P_\pi\) denote the permutation operator associated with
\(\pi\in S_r\), and let \(\gamma=(1\,2\,\ldots\,r)\) be chosen so that
\(\operatorname{Tr}(X_1\cdots X_r)
=\operatorname{Tr}[(X_1\otimes\cdots\otimes X_r)P_\gamma]\). Thus
\(\operatorname{Tr}[(A_UM)^r]
=\operatorname{Tr}[A_U^{\otimes r}M^{\otimes r}P_\gamma]\).

We now average \(A_U^{\otimes r}=(U^\dagger B U)^{\otimes r}\).
The standard unitary Haar-integration formula
\cite[Cor.~2.4, Eq.~(11)]{collins2006integration} gives
\(\mathbb E_U A_U^{\otimes r}
=\sum_{\sigma,\tau\in S_r}
\operatorname{Wg}_D(\tau^{-1}\sigma)
\operatorname{Tr}[B^{\otimes r}P_\sigma]P_\tau\),
where \(\operatorname{Wg}_D\) is the unitary Weingarten function.
Substituting this into the expression for \(\alpha_k\) gives
\[
    \alpha_k
    =
    \frac1D
    \sum_{\sigma,\tau\in S_r}
    \operatorname{Wg}_D(\tau^{-1}\sigma)\,
    \operatorname{Tr}[B^{\otimes r}P_\sigma]\,
    \operatorname{Tr}[P_\tau M^{\otimes r}P_\gamma].
\]

The two Pauli traces are simple. For any \(\sigma\in S_r\),
\(\operatorname{Tr}[B^{\otimes r}P_\sigma]
=\prod_{c\in\mathrm{cyc}(\sigma)}\operatorname{Tr}(B^{|c|})\).
Since \(B^2=I\) and \(\operatorname{Tr}B=0\), one has
\(\operatorname{Tr}(B^s)=0\) for odd \(s\) and
\(\operatorname{Tr}(B^s)=D\) for even \(s\). Hence the trace vanishes
unless every cycle of \(\sigma\) has even length, in which case it equals
\(D^{\#\sigma}\), where \(\#\sigma\) denotes the number of cycles.
Exactly the same statement holds for \(M\).

Let \(\mathcal E_r\subset S_r\) be the set of permutations with only
even-length cycles. Since
\(\operatorname{Tr}[P_\tau M^{\otimes r}P_\gamma]
=\operatorname{Tr}[M^{\otimes r}P_{\gamma\tau}]\), only
\(\sigma\in\mathcal E_r\) and \(\gamma\tau\in\mathcal E_r\) contribute.
Relabelling \(\eta=\gamma\tau\), so that
\(\tau^{-1}\sigma=\eta^{-1}\gamma\sigma\), we obtain
\begin{equation}
    \alpha_k
    =
    \frac1D
    \sum_{\sigma,\eta\in\mathcal E_r}
    D^{\#\sigma+\#\eta}\,
    \operatorname{Wg}_D(\eta^{-1}\gamma\sigma).
    \label{eq:haar-otock-weingarten}
\end{equation}

It remains to bound the terms in this sum. For fixed \(r\),
Ref.~\cite[Prop.~2.6]{collins2006integration} gives
\(\operatorname{Wg}_D(\pi)=O_r(D^{-r-|\pi|})\), where
\(|\pi|=r-\#\pi\) is the transposition length of \(\pi\). In particular,
\(\operatorname{Wg}_D(\pi)=O_r(D^{-r})\) for every \(\pi\), and
\(\operatorname{Wg}_D(\pi)=O_r(D^{-r-1})\) whenever \(\pi\neq e\).

Since every cycle of \(\sigma,\eta\in\mathcal E_r\) has length at least two,
\(\#\sigma,\#\eta\le r/2\). If
\(\#\sigma+\#\eta\le r-1\), the corresponding contribution in
Eq.~\eqref{eq:haar-otock-weingarten} is
\(O_k(D^{-1}D^{r-1}D^{-r})=O_k(D^{-2})\).

The only remaining case is
\(\#\sigma=\#\eta=r/2=k\). Then every cycle of \(\sigma\) and \(\eta\)
has length exactly two, so both are products of \(k\) disjoint
transpositions and have sign \((-1)^k\). Since \(r=2k\) is even,
the \(r\)-cycle \(\gamma\) is odd. Therefore
\(\operatorname{sgn}(\eta^{-1}\gamma\sigma)
=(-1)^k(-1)(-1)^k=-1\), so
\(\eta^{-1}\gamma\sigma\neq e\). The stronger Weingarten estimate then gives
\(\operatorname{Wg}_D(\eta^{-1}\gamma\sigma)=O_k(D^{-r-1})\), and the
corresponding contribution is again
\(O_k(D^{-1}D^rD^{-r-1})=O_k(D^{-2})\).

Thus every nonzero term in Eq.~\eqref{eq:haar-otock-weingarten} is
\(O_k(D^{-2})\). Since \(r=2k\) is fixed, there are only \(O_k(1)\) terms,
and therefore \(|\alpha_k|=O_k(D^{-2})\). The Weingarten estimate holds for
all sufficiently large \(D\) at fixed \(r\); the finitely many remaining
dimensions can be absorbed into \(C_k\), using \(|\alpha_k|\le1\).

Finally,
\(\left|1-\alpha_k\right|\ge1-|\alpha_k|
\ge1-C_kD^{-2}\), which is at least \(1/2\) for all sufficiently large \(n\).
\end{proof}

\section{Fluctuations of \texorpdfstring{\(\mathrm{OTOC}^{(k)}\)}{OTOC(k)}}
\label{app:fluct}
In this section we show that $\mathrm{OTOC}^{(k)}$ in local random circuits have
inverse-polynomial circuit-to-circuit fluctuations in the transition window.  We first explain the
general proof idea.  Before the light cone reaches the measured operator, the OTOC is pinned to
\(1\).  After the circuit has reached sufficiently high depth, its mean is nearly zero -- which can be certified for example using convergence to higher-order design.  Hence the averaged OTOC must change by a constant amount across this
transition window. Since this change is produced by local random gates, at least one
constant-depth block in the transition window must have a non-negligible effect on the OTOC, and
this forces circuit-to-circuit fluctuations.

We begin by proving this in the concrete case of \(\mathrm{OTOC}^{(2)}\) for
one-dimensional brickwork circuits with Haar random gates. The subsequent subsections abstract the two inputs used in
that proof---mean control from Haar/design convergence and local ensemble reverse-variance bound ---and combine them to obtain the more general $\mathrm{OTOC}^{(k)}$ theorem for more general circuit architectures and local gate ensembles, such as fixed-two-qubits gates interleaved with random single-qubit rotations which are more aligned with the Google experiment~\cite{abanin2025constructiveinterferenceedgequantum}.

\subsubsection{Warm-up: One-dimensional Haar brickwork \(\mathrm{OTOC}^{(2)}\)}

\begin{theorem}[One-dimensional Haar brickwork \(\mathrm{OTOC}^{(2)}\)]
\label{thm:one-dimensional-haar-brickwork-otoc2}
There exists a constant \(C\ge1\) such that, for every fixed integer \(R\ge0\), there exists
\(c_R>0\) for which the following holds.  Let \(n\ge\max\{4,R\}\) be even, and consider a
one-dimensional Haar brickwork circuit on \(n\) qubits.  Set \(B\coloneqq Z_1\) and
\(M\coloneqq Z_{1+n/2}\), and let \(\rho\) be any fixed \(n\)-qubit initial state.

Then there exists an interval of \(R+1\) consecutive depths
\begin{align}
    I
    &\subseteq
    \left\{
        \frac n2,\ldots,\lceil Cn\rceil
    \right\}
\end{align}
such that, for every \(d\in I\),
\begin{align}
    \Var_{U_d}
    \!\left(
        \mathrm{OTOC}^{(2)}_\rho(U_d)
    \right)
    &\ge
    \frac{c_R}{n^2}.
    \label{eq:otoc2-brickwork-variance-lower-bound}
\end{align}
In particular, throughout the whole interval \(I\), the circuit-to-circuit fluctuations are
inverse-polynomial:
\begin{align}
    \Var_{U_d}
    \!\left(
        \mathrm{OTOC}^{(2)}_\rho(U_d)
    \right)
    &=
    \Omega\!\left(
        \frac{1}{\operatorname{poly}(n)}
    \right)
    \qquad
    \text{for every } d\in I .
\end{align}
\end{theorem}

\begin{proof}
We prove a slightly stronger quantitative statement.  There are constants \(C,c>0\) and
\(0<\kappa<1\), independent of \(R,n\), such that for every \(0\le R\le n\) there is a depth
\(d_\star\) for which
\begin{align}
    I_R
    &\coloneqq 
    \{d_\star,d_\star+1,\ldots,d_\star+R\}
    \subseteq
    \left\{
        \frac n2,\ldots,\lceil Cn\rceil
    \right\},
    \label{eq:otoc2-proof-interval-definition}
\end{align}
and, for every \(r=0,\ldots,R\),
\begin{align}
    \Var_{U_{d_\star+r}}
    \!\left(
        \mathrm{OTOC}^{(2)}_\rho(U_{d_\star+r})
    \right)
    &\ge
    \frac{c\,\kappa^r}{n^2}.
    \label{eq:otoc2-proof-pointwise-quantitative-bound}
\end{align}

Write
\begin{align}
    F_d
    &\coloneqq 
    \mathrm{OTOC}^{(2)}_\rho(U_d),
    &
    \mu_d
    &\coloneqq 
    \E F_d,
    &
    d_0
    &\coloneqq 
    \frac n2-1 .
\end{align}
For every \(d\le d_0\), the Heisenberg light cone of \(B=Z_1\) has not reached the support of
\(M=Z_{1+n/2}\).  Hence Lemma~\ref{lem:pre-light-cone-identity} gives
\begin{align}
    F_d
    &=
    1
    \qquad
    \text{for every circuit instance and every } d\le d_0,
\end{align}
and therefore
\begin{align}
    \mu_{d_0}
    &=
    1 .
    \label{eq:otoc2-proof-pre-light-cone-mean}
\end{align}

Choose \(C_0\ge1\) so that, with \(T\coloneqq \lceil C_0n\rceil\), linear-depth fourth-moment convergence
for one-dimensional Haar brickwork circuits gives~\cite{brandao2016local,haferkamp2022random}
\begin{align}
    \left|
        \mu_T
        -
        \E_{U\sim{\rm Haar}}
        \mathrm{OTOC}^{(2)}_\rho(U)
    \right|
    &\le
    \frac14 .
    \label{eq:otoc2-proof-design-mean-control}
\end{align}
By the state-independence lemma for averaged OTOCs,
Proposition~\ref{prop:haar-otock-small}, the global-Haar mean is independent of
\(\rho\).  The exact fourth-moment computation, as done in
Proposition~\ref{prop:haar-otock-small}, gives
\begin{align}
    h_2
    &\coloneqq 
    \E_{U\sim{\rm Haar}}
    \mathrm{OTOC}^{(2)}_\rho(U)
    =
    -\frac{2D^2-9}{(D^2-1)(D^2-9)},
    &
    D
    &\coloneqq 
    2^n .
    \label{eq:otoc2-proof-exact-haar-value}
\end{align}
For \(n\ge4\), \(h_2\le0\).  Combining this with
Eqs.~\eqref{eq:otoc2-proof-pre-light-cone-mean} and
\eqref{eq:otoc2-proof-design-mean-control}, we get
\begin{align}
    |\mu_T-\mu_{d_0}|
    &=
    |\mu_T-1|
    \ge
    |h_2-1|-|\mu_T-h_2|
    \ge
    1-\frac14
    =
    \frac34
    \ge
    \frac12 .
    \label{eq:otoc2-proof-mean-change}
\end{align}

By telescoping, the triangle inequality, and Eq.~\eqref{eq:otoc2-proof-mean-change},
\begin{align}
    \frac12
    &\le
    |\mu_T-\mu_{d_0}|
    =
    \left|
        \sum_{d=d_0+1}^{T}
        (\mu_d-\mu_{d-1})
    \right|
    \le
    \sum_{d=d_0+1}^{T}
    |\mu_d-\mu_{d-1}|
    \notag\\
    &\le
    (T-d_0)
    \max_{d\in\{d_0+1,\ldots,T\}}
    |\mu_d-\mu_{d-1}| .
\end{align}
Hence there exists \(d_\star\in\{d_0+1,\ldots,T\}\) such that
\begin{align}
    |\mu_{d_\star}-\mu_{d_\star-1}|
    &\ge
    \frac{1}{2(T-d_0)} .
\end{align}

Since \(T-d_0=\lceil C_0n\rceil-n/2+1\le2C_0n\) for \(n\ge4\), this gives
\begin{align}
    |\mu_{d_\star}-\mu_{d_\star-1}|
    &\ge
    \frac{1}{4C_0n}.
    \label{eq:otoc2-proof-large-increment}
\end{align}

We now convert this mean increment into variance at depth \(d_\star\).  Condition on
\(U_{d_\star-1}=V\), and expose the layer \(L_{d_\star}\).  Let \(X\) denote the active part of this
layer, namely the local gate or gates that can affect \(L_{d_\star}^\dagger B L_{d_\star}\).  All
other gates in the layer commute through \(B\) and cancel, so they are suppressed from the notation.
In one-dimensional Haar brickwork, \(X\) is a constant-size local patch made by Haar random $2$-qubits gates.  If the active
patch were empty, then \(F_{d_\star}=F_{d_\star-1}\) after conditioning on \(V\), and hence
\(\mu_{d_\star}=\mu_{d_\star-1}\), contradicting Eq.~\eqref{eq:otoc2-proof-large-increment}.
Thus the active patch is nonempty.

For fixed \(V\), define
\begin{align}
    f_V(X)
    &\coloneqq 
    \mathrm{OTOC}^{(2)}_\rho(L_{d_\star}(X)V).
\end{align}
Replacing the active patch by the identity gives
\begin{align}
    f_V(I)
    &=
    F_{d_\star-1}(V),
\end{align}
while averaging over the active patch gives
\begin{align}
    \E_X f_V(X)
    &=
    \E[
        F_{d_\star}
        \mid
        U_{d_\star-1}=V
    ].
\end{align}
For Haar random gates (as well as for more general gates ensembles) the following reverse-variance bound holds, proved below in
Corollary~\ref{cor:explicit-product-haar-local-blocks}. Specifically, there exists an absolute constant
\(\eta_{\rm H}>0\), independent of \(n\), such that
\begin{align}
    \Var_X f_V(X)
    &\ge
    \eta_{\rm H}
    \left|
        \E_X f_V(X)-f_V(I)
    \right|^2 .
    \label{eq:otoc2-proof-local-rv}
\end{align}
Using the law of total variance, Eq.~\eqref{eq:otoc2-proof-local-rv}, and Jensen's inequality,
\begin{align}
    \Var(F_{d_\star})
    &=
    \E_V\Var_X f_V(X)
    +
    \Var_V\!\left(
        \E_X f_V(X)
    \right)
    \notag\\
    &\ge
    \eta_{\rm H}
    \E_V
    \left|
        \E_X f_V(X)-f_V(I)
    \right|^2
    \notag\\
    &\ge
    \eta_{\rm H}
    \left|
        \E_V\E_X f_V(X)-\E_V f_V(I)
    \right|^2
    \notag\\
    &=
    \eta_{\rm H}
    |\mu_{d_\star}-\mu_{d_\star-1}|^2 .
    \label{eq:otoc2-proof-slope-to-variance-at-dstar}
\end{align}
Together with Eq.~\eqref{eq:otoc2-proof-large-increment}, this yields
\begin{align}
    \Var(F_{d_\star})
    &\ge
    \frac{\eta_{\rm H}}{16C_0^2n^2}.
    \label{eq:otoc2-proof-first-variance-depth}
\end{align}

It remains to propagate this variance forward. The intuition is that adding one more local Haar layer
cannot erase an existing variance spike. Formally, by Lemma~\ref{lem:forward-persistence-one-layer}, there exists
\begin{align}
    \kappa
    &\coloneqq 
     \frac{\eta_{\rm H}}{1+\eta_{\rm H}}
    >
    0
\end{align}
such that \(\Var(F_{d+1})\ge\kappa\,\Var(F_d)\) for every depth \(d\).  Iterating from the
variance-producing depth \(d_\star\), we obtain, for every \(r=0,\ldots,R\),
\begin{align}
    \Var(F_{d_\star+r})
    &\ge
    \kappa^r\Var(F_{d_\star})
    \notag\\
    &\ge
    \frac{\eta_{\rm H}\,\kappa^r}{16C_0^2n^2}.
    \label{eq:otoc2-proof-persistence-pointwise}
\end{align}

Setting \(c\coloneqq \eta_{\rm H}/(16C_0^2)\) gives
\begin{align}
    \Var(F_{d_\star+r})
    &\ge
    \frac{c\,\kappa^r}{n^2}
    \qquad
    \text{for every } r=0,\ldots,R.
    \label{eq:otoc2-proof-pointwise-bound-r}
\end{align}

It remains only to check the location of the interval.  Since
\(d_\star\in\{d_0+1,\ldots,T\}\), we have \(d_\star\ge n/2\).  Also, using \(R\le n\),
\begin{align}
    d_\star+R
    &\le
    T+R
    \le
    \lceil C_0n\rceil+n
    \le
    \lceil Cn\rceil
\end{align}
for any constant \(C\ge C_0+2\).  Thus \(I_R\) satisfies
Eq.~\eqref{eq:otoc2-proof-interval-definition}.  Finally, since \(0<\kappa<1\),
Eq.~\eqref{eq:otoc2-proof-pointwise-bound-r} implies
\begin{align}
    \Var(F_{d_\star+r})
    &\ge
    \frac{c\,\kappa^R}{n^2}
    \qquad
    \text{for every } r=0,\ldots,R.
\end{align}
For fixed \(R\), this proves the theorem with \(I\coloneqq I_R\) and \(c_R\coloneqq c\,\kappa^R>0\).
\end{proof}

\subsection{Mean change from convergence to the Haar OTOC value}
\label{app:moment-control-design-depth}
The slope-to-variance argument uses approximate designs convergence: specifically, it uses that,
at some depth, the circuit ensemble must reproduce the Haar average of the OTOC observable under
consideration.  We formulate this condition directly.

\begin{definition}[OTOC-test moment control]
\label{def:otoc-test-moment-control}
Fix \(k\ge1\), a density matrix \(\rho\), and Pauli observables \(B,M\).  A unitary ensemble
\(\nu\) on \(\mathrm U(\mathcal H_n)\) satisfies \((k,\varepsilon)\)-OTOC-test moment control for
\((\rho,B,M)\) if
\begin{align}
    \left|
        \E_{U\sim\nu}
        \mathrm{OTOC}^{(k)}_\rho(U)
        -
        \E_{U\sim\mathrm{Haar}}
        \mathrm{OTOC}^{(k)}_\rho(U)
    \right|
    &\le
    \varepsilon,
    \label{eq:otoc-test-moment-control}
\end{align}
where
\begin{align}
    \mathrm{OTOC}^{(k)}_\rho(U)
    &\coloneqq 
    \Tr\!\left[
        \rho\,(U^\dagger BUM)^{2k}
    \right].
    \label{eq:otoc-test-observable}
\end{align}
For a depth-\(d\) circuit ensemble \(\nu_d\), we denote by
\(d_{\rm OTOC}^{(2k)}(\varepsilon;\rho,B,M)\) the smallest depth at which
\(\nu_d\) satisfies \((k,\varepsilon)\)-OTOC-test moment control for \((\rho,B,M)\).
\end{definition}

Definition~\ref{def:otoc-test-moment-control} is the exact mean-control property required in the
proof. We next recall approximations-notions of $k$-design that would imply this in particular and that we can import from the literature.  One of the cleanest approach is through the operational definition of approximate designs: if no experiment with forward and inverse access to the unitary can distinguish the ensemble from Haar, then in particular the OTOC experiment cannot distinguish them.

\subsubsection{Operational unitary designs with forward and inverse access}
\label{app:operational-design-notions}

As for any approximation notion, a particularly useful metric for unitary designs is dictated by its operational
meaning. For an approximate-design, we may want to ask whether the ensemble can be distinguished from Haar
randomness by experiments with query access to the sampled unitary.  
Since OTOC protocols use
time reversal, the relevant experiments are allowed to query both \(U\) and \(U^\dagger\).

A \((p,q)\)-query forward/inverse experiment is a quantum protocol that makes at most \(p\) calls
to \(U\) and at most \(q\) calls to \(U^\dagger\), with arbitrary auxilirary qubits, quantum memory,
intermediate quantum channels, and a final measurement.  If \(\mathcal A\) is such an experiment,
let \(\sigma_\nu^{\mathcal A}\) denote the final output state averaged over \(U\sim\nu\), and let
\(\sigma_{\rm Haar}^{\mathcal A}\) denote the corresponding output state averaged over Haar-random
\(U\).

\begin{definition}[Operational forward/inverse strong-design approximation]
\label{def:operational-forward-inverse-design}
A unitary ensemble \(\nu\) is an operational forward/inverse \((p,q)\)-design approximation with
error \(\varepsilon_{\rm op}\) if
\begin{align}
    \sup_{\mathcal A}
    \frac12
    \left\|
        \sigma_\nu^{\mathcal A}
        -
        \sigma_{\rm Haar}^{\mathcal A}
    \right\|_1
    &\le
    \varepsilon_{\rm op},
    \label{eq:operational-forward-inverse-error}
\end{align}
where the supremum is over all \((p,q)\)-query forward/inverse experiments.
\end{definition}

Thus \(\varepsilon_{\rm op}\) bounds the distinguishability of the circuit
ensemble from Haar by any allowed experiment. Closely related notions are
referred to as strong, measurable, or adaptive unitary designs in the
literature
~\cite{bittel2025adaptivehomeopathy,cui2025unitarydesigns,
schuster2025strongrandom,parelladilmé2026strongunitarydesignsoptimal}.
Some works consider the stronger mixed-query setting in which each oracle
call may be chosen from \(U\), \(U^\dagger\), \(\overline U\), or \(U^T\).
Any bound in this stronger model immediately implies the corresponding
forward/inverse bound
~\cite{schuster2025strongrandom,folkertsma2026artscrafts,
parelladilmé2026strongunitarydesignsoptimal}.

We now connect this notion to OTOCs. Set
\(A_U\coloneqq U^\dagger B U\) and
\(W_U\coloneqq(A_UM)^{2k}\), so that
\(\mathrm{OTOC}^{(k)}_\rho(U)=\Tr[\rho W_U]\).
A Hadamard test for \(W_U\) does not require controlled access to \(U\):
indeed,
\(\operatorname{ctrl}(A_U)
=(I_{\rm c}\otimes U^\dagger)\operatorname{ctrl}(B)
(I_{\rm c}\otimes U)\),
while \(\operatorname{ctrl}(M)\) is a controlled Pauli. Hence
\(\operatorname{ctrl}(W_U)\) can be implemented using \(2k\) calls to
\(U\), \(2k\) calls to \(U^\dagger\), and controlled Pauli gates.
Measuring \(X\) on the control qubit gives
\(\operatorname{Re}\Tr[\rho W_U]\), while measuring \(Y\) gives
\(\operatorname{Im}\Tr[\rho W_U]\).

\begin{lemma}[Operational designs control OTOC tests]
\label{lem:operational-design-implies-otoc}
Let \(B,M\) be Pauli observables and let \(\rho\) be a density matrix.
If \(\nu\) is an operational forward/inverse
\((2k,2k)\)-design approximation with error \(\varepsilon_{\rm op}\), then
\begin{align}
    \left|
        \E_{U\sim\nu}
        \mathrm{OTOC}^{(k)}_\rho(U)
        -
        \E_{U\sim{\rm Haar}}
        \mathrm{OTOC}^{(k)}_\rho(U)
    \right|
    &\le
    2\varepsilon_{\rm op}.
    \label{eq:operational-design-implies-otoc}
\end{align}
\end{lemma}

\begin{proof}
Let \(z_U\coloneqq\Tr[\rho W_U]\). Preparing the control qubit in
\(\lvert+\rangle\), applying \(\operatorname{ctrl}(W_U)\), and tracing out
the system leaves the control qubit in the state
\begin{align}
    \sigma_U
    &=
    \frac12
    \begin{pmatrix}
        1 & \overline{z_U}\\
        z_U & 1
    \end{pmatrix}.
\end{align}
Writing
\(\delta\coloneqq
\E_{U\sim\nu}z_U-\E_{U\sim{\rm Haar}}z_U\),
the difference between the corresponding averaged control-qubit states is
\(\frac12
\bigl(\begin{smallmatrix}
0 & \overline{\delta}\\
\delta & 0
\end{smallmatrix}\bigr)\),
whose trace norm is \(|\delta|\). Hence
Definition~\ref{def:operational-forward-inverse-design} gives
\(|\delta|/2\le\varepsilon_{\rm op}\), or equivalently
\(|\delta|\le2\varepsilon_{\rm op}\), proving
Eq.~\eqref{eq:operational-design-implies-otoc}.
\end{proof}

Consequently, any convergence theorem controlling the forward/inverse
operational error above, or a stronger mixed-query error, can be used
directly to establish the OTOC-test moment control of
Definition~\ref{def:otoc-test-moment-control}.
  
\subsubsection{OTOC control via the moment-operator spectral gap}
\label{app:tpe-otoc-control}
Although the operational forward/inverse notion above provides a natural
way to formulate convergence of OTOC means to their Haar values, many
local-random-circuit convergence results in the literature are instead
stated in terms of moment operators and tensor-product-expander (TPE)
bounds. We therefore recall the TPE formulation and show that a sufficiently
small \(2k\)-TPE error implies the OTOC-test moment control of
Definition~\ref{def:otoc-test-moment-control}.


Let \(r\ge1\) and \(d_n\coloneqq2^n\). For a probability measure \(\nu\) on
\(\mathrm U(\mathcal H_n)\), define the \(r\)-fold moment channel
\begin{align}
    \Phi_\nu^{(r)}(X)
    &\coloneqq
    \E_{U\sim\nu}
    \left[
        U^{\otimes r}X(U^\dagger)^{\otimes r}
    \right],
    \qquad
    X\in\mathcal B(\mathcal H_n^{\otimes r}).
    \label{eq:tpe-moment-channel}
\end{align}
Its Liouville representation is, up to the conventional ordering of tensor
factors, the moment operator
\begin{align}
    \mathcal M_\nu^{(r)}
    &\coloneqq
    \E_{U\sim\nu}
    \left[
        U^{\otimes r}\otimes\overline U^{\otimes r}
    \right].
    \label{eq:tpe-moment-operator}
\end{align}

\begin{definition}[\(r\)-TPE error]
\label{def:tpe-error}
The \(r\)-tensor-product-expander error of \(\nu\) is
\begin{align}
    g(\nu,r)
    &\coloneqq
    \left\|
        \mathcal M_\nu^{(r)}
        -
        \mathcal M_{\rm Haar}^{(r)}
    \right\|_\infty
    =
    \left\|
        \Phi_\nu^{(r)}
        -
        \Phi_{\rm Haar}^{(r)}
    \right\|_{2\to2}.
    \label{eq:tpe-error-definition}
\end{align}
\end{definition}

Thus \(g(\nu,r)\) is the operator-norm error of the \(r\)-th moment
operator, equivalently the Hilbert--Schmidt \(2\to2\) error of the
\(r\)-fold moment channel.
  
The usefulness of the TPE norm comes from its amplification under independent composition.  Let
\(\mu\) be the distribution of one step, such as a single circuit layer or a fixed block of layers, and let
\(\mu^{*L}\) be the distribution obtained by composing \(L\) independent copies of this step.  Assume
that the fixed-point space of \(\mathcal M_\mu^{(r)}\) is the Haar-invariant subspace, so that
\(\mathcal M_{\rm Haar}^{(r)}\) is the corresponding fixed-point projection.  Define the global
\(r\)-moment spectral gap of this step by
\begin{align}
    \Delta_r(n)
    &\coloneqq 
    1-g(\mu,r).
    \label{eq:def-global-tpe-gap}
\end{align}

If
\(\Delta_r(n)>0\), then it holds
\begin{align}
    g(\mu^{*L},r)
    &\le
    \exp(-\Delta_r(n)L).
    \label{eq:tpe-spectral-gap-amplification}
\end{align}
Indeed, independent composition gives
\(\mathcal M_{\mu^{*L}}^{(r)}=(\mathcal M_\mu^{(r)})^L\), while the invariance of the Haar measure
implies
\(\mathcal M_\mu^{(r)}\mathcal M_{\rm Haar}^{(r)}
=\mathcal M_{\rm Haar}^{(r)}\mathcal M_\mu^{(r)}
=\mathcal M_{\rm Haar}^{(r)}\).  Hence
\(\mathcal M_{\mu^{*L}}^{(r)}-\mathcal M_{\rm Haar}^{(r)}
=(\mathcal M_\mu^{(r)}-\mathcal M_{\rm Haar}^{(r)})^L\), and taking operator norms yields
\(g(\mu^{*L},r)\le g(\mu,r)^L=(1-\Delta_r(n))^L\le e^{-\Delta_r(n)L}\).  This is the standard TPE-gap amplification mechanism behind moment convergence
~\cite{brandao2016local,haferkamp2022random}.

\begin{lemma}[Direct TPE coefficient bound and OTOC consequence]
\label{lem:tpe-coefficient-bound}
Let \(r\ge1\), let \(A,Y\in\mathcal B(\mathcal H_n^{\otimes r})\), and set
\(f_{A,Y}(U)\coloneqq \Tr[A\,U^{\otimes r}Y(U^\dagger)^{\otimes r}]\). Then
\begin{align}
    \left|
        \E_{U\sim\nu}f_{A,Y}(U)
        -
        \E_{U\sim{\rm Haar}}f_{A,Y}(U)
    \right|
    &\le
    \|A\|_2\,\|Y\|_2\,g(\nu,r).
    \label{eq:tpe-direct-coefficient-bound}
\end{align}
Consequently, for \(r=2k\) and \(d_n\coloneqq \dim\mathcal H_n=2^n\),
\begin{align}
    \left|
        \E_{U\sim\nu}\mathrm{OTOC}^{(k)}_\rho(U)
        -
        \E_{U\sim{\rm Haar}}\mathrm{OTOC}^{(k)}_\rho(U)
    \right|
    &\le
    d_n^{\,r-\frac12}\|\rho\|_2\,g(\nu,r),
    \qquad r=2k.
    \label{eq:tpe-otoc-control-rho}
\end{align}
Moreover, if \(\nu\) is right-Pauli invariant, then the averaged OTOC is state-independent, and the
maximally mixed bound applies to every density matrix \(\rho\):
\begin{align}
    \left|
        \E_{\nu}\mathrm{OTOC}^{(k)}_\rho
        -
        \E_{\rm Haar}\mathrm{OTOC}^{(k)}_\rho
    \right|
    &\le
    d_n^{2k-1}g(\nu,2k).
    \label{eq:tpe-otoc-control-right-pauli}
\end{align}
\end{lemma}

\begin{proof}
By definition of the moment channel and by Cauchy--Schwarz,
\begin{align}
    \left|
        \E_\nu f_{A,Y}
        -
        \E_{\rm Haar}f_{A,Y}
    \right|
    &=
    \left|
        \Tr\!\left[
            A
            \left(
                \Phi_\nu^{(r)}
                -
                \Phi_{\rm Haar}^{(r)}
            \right)(Y)
        \right]
    \right|
    \le
    \|A\|_2
    \left\|
        \left(
            \Phi_\nu^{(r)}
            -
            \Phi_{\rm Haar}^{(r)}
        \right)(Y)
    \right\|_2
    \le
    \|A\|_2\,\|Y\|_2\,g(\nu,r).
    \label{eq:tpe-coefficient-proof}
\end{align}
This proves Eq.~\eqref{eq:tpe-direct-coefficient-bound}.

For the OTOC application, let \(r=2k\). Let \(\Omega_r\) be the cyclic permutation operator chosen
so that \(\Tr[X_1\cdots X_r]=\Tr[(X_1\otimes\cdots\otimes X_r)\Omega_r]\), and define
\(Y_{\rho,M}^{(r)}\coloneqq (M^{\otimes(r-1)}\otimes M\rho)\Omega_r\). Writing \(A_U\coloneqq U^\dagger B U\), cyclicity of the trace together
with the defining property of \(\Omega_r\) gives
\begin{align}
    \mathrm{OTOC}^{(k)}_\rho(U)
    &=
    \Tr[\rho(A_U M)^r]
    =
    \Tr\!\left[
        B^{\otimes r}
        U^{\otimes r}
        Y_{\rho,M}^{(r)}
        (U^\dagger)^{\otimes r}
    \right].
    \label{eq:otoc-as-tpe-coefficient}
\end{align}
Thus the first part applies with \(A=B^{\otimes r}\) and \(Y=Y_{\rho,M}^{(r)}\). Since \(B\) and
\(M\) are Pauli unitaries and \(\Omega_r\) is unitary, \(\|B^{\otimes r}\|_2=d_n^{r/2}\) and
\(\|Y_{\rho,M}^{(r)}\|_2=d_n^{(r-1)/2}\|\rho\|_2\). Substituting these estimates into
Eq.~\eqref{eq:tpe-direct-coefficient-bound} gives Eq.~\eqref{eq:tpe-otoc-control-rho}.

The final claim follows from Lemma~\ref{lem:state-independence-higher-otoc}, since both \(\nu\)
and Haar measure are right-Pauli invariant, and
\(\|I/d_n\|_2=d_n^{-1/2}\).
\end{proof}

\begin{corollary}[TPE accuracy sufficient for OTOC-test moment control]
\label{cor:tpe-sufficient-for-otoc-control}
Let \(d_n\coloneqq 2^n\).  For arbitrary input states, it is sufficient to have
\begin{align}
    g(\nu,2k)
    &\le
    \varepsilon\,d_n^{-(2k-\frac12)}
    \label{eq:tpe-error-needed-for-otoc}
\end{align}
in order for \(\nu\) to satisfy \((k,\varepsilon)\)-OTOC-test moment control for
\((\rho,B,M)\).

In particular, suppose that \(\nu=\mu^{*L}\), where \(\mu\) is one step, or one fixed block of layers,
and set
\begin{align}
    \Delta_{2k}(n)
    &\coloneqq 
    1-g(\mu,2k).
    \label{eq:def-global-gap-2k}
\end{align}
Then the condition
\begin{align}
    L
    &\ge
    \Delta_{2k}(n)^{-1}
    \left[
        2k \log d_n
        +
        \log\frac1{\varepsilon}
    \right]
    \label{eq:tpe-depth-needed-for-otoc}
\end{align}
is sufficient for \((k,\varepsilon)\)-OTOC-test moment control for arbitrary input states.
\end{corollary}

\begin{proof}
By Lemma~\ref{lem:tpe-coefficient-bound}, for arbitrary density matrices \(\rho\),
\begin{align}
    \left|
        \E_{\nu}\mathrm{OTOC}^{(k)}_\rho
        -
        \E_{\rm Haar}\mathrm{OTOC}^{(k)}_\rho
    \right|
    &\le
    d_n^{2k-\frac12}g(\nu,2k).
    \label{eq:tpe-corollary-first-step}
\end{align}
Thus the TPE accuracy condition~\eqref{eq:tpe-error-needed-for-otoc} implies the desired
\((k,\varepsilon)\)-OTOC-test moment control.

If \(\nu=\mu^{*L}\), the spectral-gap amplification gives
\begin{align}
    g(\mu^{*L},2k)
    &\le
    \exp[-\Delta_{2k}(n)L].
    \label{eq:tpe-gap-amplification-proof-use}
\end{align}
The lower bound~\eqref{eq:tpe-depth-needed-for-otoc} implies
\(g(\mu^{*L},2k)\le\varepsilon d_n^{-(2k-\frac12)}\), so the first part applies.
\end{proof}

\subsubsection{Design and OTOC-convergence from the literature}
There are several ways to obtain the mean-control condition in
Definition~\ref{def:otoc-test-moment-control}. Direct OTOC results are the
most tailored, but they are comparatively limited and mostly concern
\(\mathrm{OTOC}^{(1)}\)
~\cite{harrow2019separation,nahum2018operator}. For example,
Ref.~\cite{harrow2019separation} establishes rigorous OTOC-saturation
bounds for local random circuits on graphs. For fixed-dimensional Euclidean
lattices, the saturation scale is \(\Theta(n^{1/D_{\rm lat}})\).

A second route is through the operational design notion introduced above.
This is a natural metric for OTOC-like protocols, since it compares
experiments that may interleave calls to \(U\) and \(U^\dagger\). Any
convergence theorem controlling this operational error can be used directly
in Lemma~\ref{lem:operational-design-implies-otoc}. Recent strong-design
results provide such guarantees for one-dimensional and all-to-all
architectures, in a stronger query model allowing access to
\(U,U^\dagger,\overline U\), and \(U^T\)
~\cite{schuster2025strongrandom,parelladilmé2026strongunitarydesignsoptimal}. For constant-dimensional grids,
Ref.~\cite{folkertsma2026artscrafts} gives strong
\(\varepsilon\)-approximate \(q\)-design constructions in depth
\(O_q(n^{1/D_{\rm lat}})\).
Since the OTOC experiment above uses \(2k\) calls to \(U\) and \(2k\)
calls to \(U^\dagger\), taking \(q=4k\) gives OTOC-test moment control
at this depth scale.

As earlier mentioned, the broadest route available for local random circuits is through
moment-operator TPE bounds. Refs.~\cite{brandao2016local,haferkamp2022random}
give standard one-dimensional TPE/design inputs,
Ref.~\cite{harrowmehraban2023approximate} treats higher-dimensional
lattices, and Refs.~\cite{belkin2024generic,mittal2023local} treat more
general architectures and graph families. To state the resulting depth
bound, let \(d_{\rm TPE}^{(2k)}(\delta)\) denote the smallest depth \(d\)
for which the depth-\(d\) circuit ensemble \(\nu_d\) satisfies
\(g(\nu_d,2k)\le\delta\). Corollary~\ref{cor:tpe-sufficient-for-otoc-control}
then gives, for arbitrary input states,
\begin{align}
    d_{\rm OTOC}^{(2k)}(\varepsilon;\rho,B,M)
    &\le
    d_{\rm TPE}^{(2k)}
    \!\left(
        \varepsilon d_n^{-(2k-\frac12)}
    \right).
    \label{eq:otoc-depth-from-tpe-depth}
\end{align}

As a concrete example, Ref.~\cite{chen2024incompressibility} proves a
TPE-gap bound for the one-dimensional periodic brickwork random-circuit
ensemble. Let \(\nu_{\rm BRQC,n}\) denote one brickwork period, consisting
of the two alternating nearest-neighbor layers. In the notation of that
work,
\begin{align}
    g(\nu_{\rm BRQC,n},t)
    &\le
    1-\Omega\!\left((\log t)^{-7}\right),
    \qquad
    t\le \Theta(2^{2n/5}).
    \label{eq:chen-brickwork-gap}
\end{align}
Hence, for \(L\) independent brickwork periods,
\(g(\nu_{\rm BRQC,n}^{*L},t)
\le
\exp[-\Omega(L/(\log t)^7)]\).

We now apply Corollary~\ref{cor:tpe-sufficient-for-otoc-control} with
\(t=2k\) and fixed \(k\). The corollary requires
\(g(\nu_{\rm BRQC,n}^{*L},2k)
\le
\varepsilon d_n^{-(2k-\frac12)}\).
Since \(d_n=2^n\) and \(k\) is fixed, it is sufficient to take
\(L=O_k(n+\log(1/\varepsilon))\). Therefore
Ref.~\cite{chen2024incompressibility} implies
\begin{align}
    d_{\rm OTOC}^{(2k)}(\varepsilon;\rho,B,M)
    &=
    O_k\!\left(
        n+\log\frac1{\varepsilon}
    \right)
    \label{eq:chen-otoc-depth-final}
\end{align}
for arbitrary input states. In particular, inverse-polynomial OTOC-test
error is obtained in depth \(O_k(n)\).

For higher-dimensional lattices and more general architectures, one can similarly insert the available
\(2k\)-moment convergence depth into Eq.~\eqref{eq:otoc-depth-from-tpe-depth}.  This gives a
conservative \(O_k(n)\)-depth bound whenever the cited convergence theorem has logarithmic
dependence on the target TPE accuracy. This is sufficient for our purposes, although it is not meant
to be tight.

\subsubsection{Local gate ensembles for design convergence}

Several of the design-convergence results cited above are stated for
Haar-random local two-qubit gates, i.e., for the local ensemble
\({\rm Haar}(\mathrm U(4))\). This assumption is not essential for our
application. Ref.~\cite{yada2026nonhaar} gives comparison results that transfer
Haar-local-circuit design bounds to non-Haar local ensembles with a
positive local \(t\)-moment gap. For architectures in which the relevant
comparison involves only a constant number of consecutive layers, the
resulting overhead depends on the local ensemble and on \(t\), but not on
the total system size \(n\). Since \(t=2k\) is fixed in our application,
this overhead is absorbed into the constants hidden in \(O_k(\cdot)\).

Concretely, for a local two-qubit ensemble \(\nu_{\rm loc}\), write
\(M_t(\nu_{\rm loc})
\coloneqq
\E_{U\sim\nu_{\rm loc}}
[U^{\otimes t}\otimes\overline U^{\otimes t}]\),
and let
\(P_t\coloneqq M_t({\rm Haar}(\mathrm U(4)))\).
A positive local \(t\)-moment gap means that
\(\|M_t(\nu_{\rm loc})-P_t\|_{\rm op}
\le 1-\gamma_t(\nu_{\rm loc})\)
for some \(\gamma_t(\nu_{\rm loc})>0\).
Repeated independent applications of the same local ensemble then converge
exponentially to the Haar \(t\)-th moment operator, with a rate depending
only on the local ensemble and on \(t\).

This framework also covers hardware-motivated ensembles built from a fixed
entangling two-qubit gate interleaved with random one-qubit rotations. For
example, consider
\begin{align}
    U_{\rm loc}
    &\coloneqq
    (u_1\otimes u_2)\,G\,(v_1\otimes v_2),
    \label{eq:fixed-entangler-local-ensemble}
\end{align}
where \(G\) is a fixed entangling two-qubit gate and \(u_i,v_i\) are
independent Haar-random one-qubit unitaries. For this fully Haar-dressed ensemble, the required local \(t\)-moment gap
is positive for every fixed \(t\). Hence, taking \(t=2k\), the comparison
results of Ref.~\cite{yada2026nonhaar} transfer the corresponding
Haar-local-gate design bounds with an \(n\)-independent overhead for the
architectures considered here. For more general local ensembles, the
positive moment gap may hold only after grouping a constant number of
local steps; such constant-size blocking changes only constant factors and
does not affect the asymptotic dependence on \(n\).


\subsection{Reverse-variance bounds for local gate ensembles}
\label{app:local-reverse-variance}
In this subsection we discuss the local gate-ensemble assumptions under which our fluctuation lower
bound applies. The fluctuation argument uses two distinct properties of the local gate
ensemble. The first is control of the ensemble-averaged OTOC, discussed in
the previous subsection above. The second is a local reverse-variance bound,
which we establish here.

Fix two depths \(s<t\) with \(t-s=O(1)\), and write \(U_t=WU_s\), where
\(W\) is the random circuit block between depths \(s\) and \(t\). We
condition on the circuit up to depth \(s\), \(U_s=V\), and average over
the gates in \(W\). In the Heisenberg evolution \(W^\dagger B W\), only
gates inside the backward light cone of \(B\) can affect the OTOC. Since
\(W\) has constant depth and the architecture has bounded degree, these
gates act on a subsystem of \(O(1)\) qubits.

We denote by \(X\) the effective unitary generated by the gates in this
active part of \(W\), and by \(\nu_{\rm loc}\) its induced distribution.
If the active block contains several local gates, expectations over \(X\)
can equivalently be evaluated over their joint distribution. Throughout
this subsection, \(B\) denotes the corresponding Pauli operator on the
active subsystem, with identities on the remaining qubits left implicit.

For fixed \(V\), define
\begin{align}
    f_V(X)
    &\coloneqq
    \Tr\!\left[
        \rho
        \left(
            V^\dagger X^\dagger B X V M
        \right)^{2k}
    \right].
    \label{eq:conditional-local-otoc-function}
\end{align}
Setting the active block to the identity gives
\(f_V(I)=\mathrm{OTOC}^{(k)}_\rho(U_s)\), while averaging over the local
randomness gives
\(\E_X f_V(X)
=\E[\mathrm{OTOC}^{(k)}_\rho(U_t)\mid U_s=V]\).

The reverse-variance bound expresses the following simple principle:
if averaging over the local block changes the OTOC relative to the
identity-block value, then the same local randomness must generate
fluctuations. More precisely, we seek conditions on \(\nu_{\rm loc}\)
under which
\begin{align}
    \Var_{X\sim\nu_{\rm loc}} f_V(X)
    &\ge
    \eta
    \left|
        \E_{X\sim\nu_{\rm loc}}f_V(X)-f_V(I)
    \right|^2 ,
    \label{eq:desired-local-reverse-variance}
\end{align}
where \(\eta>0\) depends only on the local ensemble and on \(k\), but not on the total number of
qubits \(n\).

\subsubsection{Local criteria for reverse variance}
\label{app:local-reverse-variance-criteria}
We now record simple criteria which imply the reverse-variance bound. The basic point is that it is enough to rule out the following degeneracy: a OTOC function should not be constant on the support of the local ensemble while taking a different value at the identity. More precisely, the finite-dimensional lemma below shows that it is enough to prove the implication
\begin{align*} 
\Var_{X\sim\nu_{\rm loc}} f(X)=0 \quad\Longrightarrow\quad \E_{X\sim\nu_{\rm loc}}f(X)=f(I) 
\end{align*}
for the relevant finite-dimensional space of local OTOC functions. Here \(B\) denotes the local Pauli seen by the active block \(X\). If several microscopic gates are grouped into one effective local block, then \(\nu_{\rm loc}\) denotes the product distribution of this grouped block.

\begin{definition}[Local reverse-variance criteria]
\label{def:local-gate-set-condition}
Let \(\nu_{\rm loc}\) be the distribution of the local block on the relevant local subsystem, and let
\(B\) denote the butterfly Pauli acting on this subsystem.  We say that \(\nu_{\rm loc}\) satisfies a
local reverse-variance criterion for \(B\) if at least one of the following conditions holds:
\begin{enumerate}
    \item \emph{Stabilizing element.}  There exists
    \(X_\star\in\operatorname{supp}\nu_{\rm loc}\) such that
    \begin{align}
        X_\star^\dagger B X_\star
        &=
        \pm B .
        \label{eq:local-gate-set-condition}
    \end{align}
    We refer to this condition as the local gate-set condition.

    \item \emph{Open support.}  Up to irrelevant global phases, the support of \(\nu_{\rm loc}\)
    contains a neighbourhood of some local unitary.
\end{enumerate}
\end{definition}

The first condition is the most elementary one.  If \(I\in\operatorname{supp}\nu_{\rm loc}\), it
holds with \(X_\star=I\).  More generally, the identity block need not be in the support: it is
enough that the support contain one element which stabilizes the relevant Pauli up to sign.  Indeed,
if \(X_\star^\dagger B X_\star=\pm B\), then for
\(f_V(X)=\Tr[\rho(V^\dagger X^\dagger B X V M)^{2k}]\) we have
\begin{align}
    f_V(X_\star)
    &=
    f_V(I),
    \label{eq:local-gate-set-identity-value}
\end{align}
because the sign appears \(2k\) times and therefore cancels.  Hence, if
\(\Var_X f_V(X)=0\), then \(f_V\) is constant on \(\operatorname{supp}\nu_{\rm loc}\), and
\[
    \E_X f_V(X)=f_V(X_\star)=f_V(I).
\]
Thus the stabilizing-element condition gives the required zero-variance implication.  The next lemma
turns this exact implication into a quantitative lower bound.  We will then apply it to the
finite-dimensional space generated by the local OTOC functions.

\begin{lemma}[Finite-dimensional reverse variance]
\label{lem:finite-dimensional-reverse-variance}
Let \(\nu_{\rm loc}\) be a probability distribution on the local block of gates, and let
\(\mathcal F\) be a finite-dimensional vector space of continuous functions of \(X\).  Suppose that,
for every \(f\in\mathcal F\),
\begin{align}
    \Var_{X\sim\nu_{\rm loc}} f(X)=0
    \quad\Longrightarrow\quad
    \E_{X\sim\nu_{\rm loc}}f(X)=f(I).
    \label{eq:finite-dimensional-zero-var-implication}
\end{align}
Then there exists \(\eta>0\), depending only on \(\nu_{\rm loc}\) and \(\mathcal F\), such that
\begin{align}
    \Var_{X\sim\nu_{\rm loc}} f(X)
    &\ge
    \eta
    \left|
        \E_{X\sim\nu_{\rm loc}}f(X)-f(I)
    \right|^2
    \label{eq:finite-dimensional-reverse-variance}
\end{align}
for every \(f\in\mathcal F\).
\end{lemma}

\begin{proof}
Let \(\mathcal N\coloneqq \{f\in\mathcal F:\Var_X f(X)=0\}\).  This is the subspace of functions that are
constant on \(\operatorname{supp}\nu_{\rm loc}\).  By the assumption of the lemma, the linear
functional \(L(f)\coloneqq \E_X f(X)-f(I)\) vanishes on \(\mathcal N\).

Choose a vector-space complement \(\mathcal G\) such that
\begin{align}
    \mathcal F
    &=
    \mathcal N\oplus\mathcal G .
\end{align}
Every \(f\in\mathcal F\) can be written uniquely as \(f=f_0+g\), with
\(f_0\in\mathcal N\) and \(g\in\mathcal G\).  Since \(f_0\) is constant on the support of
\(\nu_{\rm loc}\), adding \(f_0\) does not change the variance.  Since \(L\) vanishes on
\(\mathcal N\), adding \(f_0\) also does not change \(L\).  Therefore
\begin{align}
    \Var_X f(X)
    &=
    \Var_X g(X),
    &
    L(f)
    &=
    L(g).
    \label{eq:reduce-to-complement-local}
\end{align}

On \(\mathcal G\), the map \(g\mapsto\sqrt{\Var_X g(X)}\) is a norm.  Indeed,
\(\sqrt{\Var_X g(X)}=\|g-\E_Xg\|_{L^2(\nu_{\rm loc})}\), so homogeneity and the triangle
inequality follow from the \(L^2\)-norm.  Moreover, if this norm vanishes for some
\(g\in\mathcal G\), then \(g\in\mathcal N\cap\mathcal G\), and hence \(g=0\).

Since \(\mathcal G\) is finite-dimensional, the unit sphere
\begin{align}
    S\coloneqq \{g\in\mathcal G:\sqrt{\Var_X g(X)}=1\}
\end{align}
is compact whenever \(\mathcal G\neq\{0\}\).  If \(\mathcal G=\{0\}\), then
\(L(f)=0\) for all \(f\in\mathcal F\), and the claim is immediate.  Otherwise, the function
\(g\mapsto |L(g)|\) is continuous on \(S\), and hence
\begin{align}
    C\coloneqq \max_{g\in S}|L(g)|<\infty .
\end{align}
If \(C=0\), then again \(L(f)=0\) for all \(f\in\mathcal F\), and the claim is immediate.  We may
therefore assume \(C>0\).  For \(g\neq0\), set
\begin{align}
    \widetilde g\coloneqq \frac{g}{\sqrt{\Var_X g(X)}} .
\end{align}
Then \(\widetilde g\in S\), and by linearity of \(L\),
\begin{align}
    |L(g)|
    =
    \sqrt{\Var_X g(X)}\,|L(\widetilde g)|
    \le
    C\sqrt{\Var_X g(X)}.
\end{align}
The same bound is trivial for \(g=0\).  Using Eq.~\eqref{eq:reduce-to-complement-local}, we get,
for every \(f\in\mathcal F\),
\begin{align}
    \left|
        \E_X f(X)-f(I)
    \right|^2
    &=
    |L(f)|^2
    =
    |L(g)|^2
    \le
    C^2\Var_X g(X)
    =
    C^2\Var_X f(X).
\end{align}
Taking \(\eta\coloneqq C^{-2}\) proves the lemma.
\end{proof}

We now apply the finite-dimensional lemma to the local OTOC functions.  Set \(r\coloneqq 2k\), let
\(\mathcal H_{\rm loc}\) be the Hilbert space of the local subsystem on which \(X\) acts, and write
\(d_{\rm loc}\coloneqq \dim\mathcal H_{\rm loc}\).  Define
\begin{align}
    \Phi_{B,k}(X)
    &\coloneqq 
    (X^\dagger B X)^{\otimes r}.
    \label{eq:Phi-B-k-local}
\end{align}
Fix an orthonormal basis
\(\mathcal B_{\rm loc}^{(r)}=\{|\alpha\rangle\}_{\alpha=1}^{d_{\rm loc}^r}\) of
\(\mathcal H_{\rm loc}^{\otimes r}\), and define \(\mathcal F_{B,k}\) as the finite-dimensional vector space of functions generated by the matrix entries of
\((X^\dagger B X)^{\otimes r}\), i.e.
\begin{align}
    \mathcal F_{B,k}
    &\coloneqq 
    \operatorname{span}
    \left\{
        X\longmapsto
        \langle \alpha|
            \Phi_{B,k}(X)
        |\beta\rangle
        \;:\;
        |\alpha\rangle,|\beta\rangle\in\mathcal B_{\rm loc}^{(r)}
    \right\}.
    \label{eq:FBk-definition}
\end{align}
Since
\(|\mathcal B_{\rm loc}^{(r)}|=d_{\rm loc}^r\), we have $ \dim\mathcal F_{B,k}
    \le
    |\mathcal B_{\rm loc}^{(r)}|^2
     =
    d_{\rm loc}^{2r}
     =
    d_{\rm loc}^{4k}.$
Since the local subsystem has constant size, this dimension is independent of the total number of
qubits.  Although the coefficients with which a particular OTOC function is expanded in this spanning family may depend on the conditioned circuit \(V\), on \(\rho\), on \(M\), and on the total system size \(n\), the function space \(\mathcal F_{B,k}\) itself depends only on the local subsystem, on \(B\), and on \(k\).

Every conditional OTOC function in Eq.~\eqref{eq:conditional-local-otoc-function} belongs to
\(\mathcal F_{B,k}\).  Indeed, \(f_V(X)\) is obtained by contracting the \(r\) copies of
\(X^\dagger B X\) with an operator depending on \(V\), \(M\), and \(\rho\).  Equivalently, there is
an operator \(K_{V,\rho,M}\in\mathsf L(\mathcal H_{\rm loc}^{\otimes r})\), independent of \(X\),
such that
\begin{align}
    f_V(X)
    &=
    \Tr\!\left[
        K_{V,\rho,M}\,
        \Phi_{B,k}(X)
    \right]=
    \sum_{\alpha,\beta}
    c_{\alpha,\beta}^{(V,\rho,M)}
    \langle \alpha|
        \Phi_{B,k}(X)
    |\beta\rangle ,
    \label{eq:fV-in-FBk}
\end{align}
for suitable coefficients \(c_{\alpha,\beta}^{(V,\rho,M)}\).  Hence
\(f_V\in\mathcal F_{B,k}\).

\begin{lemma}[Local criteria imply reverse variance]
\label{lem:two-local-criteria-reverse-variance}
Assume that \(\nu_{\rm loc}\) satisfies one of the local reverse-variance criteria in
Definition~\ref{def:local-gate-set-condition}.  Then there exists a constant \(\eta_{B,k}>0\),
independent of \(n\), such that every conditional OTOC function \(f_V\) satisfies
\begin{align}
    \Var_{X\sim\nu_{\rm loc}} f_V(X)
    &\ge
    \eta_{B,k}
    \left|
        \E_{X\sim\nu_{\rm loc}}f_V(X)-f_V(I)
    \right|^2 .
    \label{eq:two-local-criteria-reverse-variance}
\end{align}
\end{lemma}

\begin{proof}
It is enough to verify the zero-variance implication
\eqref{eq:finite-dimensional-zero-var-implication} for all \(f\in\mathcal F_{B,k}\).

First assume the stabilizing-element condition.  Since \(X_\star^\dagger B X_\star=\pm B\) and
\(2k\) is even, \((X_\star^\dagger B X_\star)^{\otimes 2k}=B^{\otimes 2k}\).  Therefore every
\(f\in\mathcal F_{B,k}\) satisfies \(f(X_\star)=f(I)\).  If \(\Var_X f(X)=0\), then \(f\) is
constant on \(\operatorname{supp}\nu_{\rm loc}\).  Since
\(X_\star\in\operatorname{supp}\nu_{\rm loc}\), we get $\E_X f(X)=f(X_\star)=f(I).$

Now assume the open-support condition.  If \(\Var_X f(X)=0\), then \(f\) is constant on
\(\operatorname{supp}\nu_{\rm loc}\), and hence on some nonempty open set of local unitaries.  Let
\(c\) be this constant value.  For \(f\in\mathcal F_{B,k}\), the difference \(f-c\) is a polynomial
function of the matrix entries of \(X\) and \(\overline X\), restricted to the local unitary group.
Since the local unitary group is connected, such a function that vanishes on a nonempty open set
vanishes everywhere on the local unitary group.  Therefore \(f(X)=c\) for all local unitaries \(X\),
and in particular \(\E_X f(X)=c=f(I)\).

In either case the hypothesis of Lemma~\ref{lem:finite-dimensional-reverse-variance} holds with
\(\mathcal F=\mathcal F_{B,k}\), and the claimed bound follows.
\end{proof}

\subsubsection{Examples of local ensembles satisfying the reverse-variance criteria} 
\label{app:examples-local-reverse-variance-criteria} 

We now list local ensembles to which Lemma~\ref{lem:two-local-criteria-reverse-variance} applies. In the next subsection, we give explicit constants. 

\begin{itemize}
    \item \emph{Identity in the support, including Haar-random gates.}
    If \(I\in\operatorname{supp}\nu_{\rm loc}\), then the stabilizing-element criterion holds with
    \(X_\star=I\).  This includes Haar-random two-qubit gates and any constant-size block whose
    support contains the identity, for example blocks of Haar-random two-qubit gates.

    \item \emph{Controlled-phase, CZ, and CNOT gates interleaved with one-qubit rotations.}
    Consider local ensembles of the form
    \begin{align}
        X
        &=
        (u_1\otimes u_2)\,G\,(v_1\otimes v_2),
        \label{eq:fixed-entangler-local-condition}
    \end{align}
    where \(G\) is fixed and the one-qubit ensembles have support containing the local basis changes
    used below; this is automatic for Haar-random one-qubit rotations.  If \(G\) is a
    controlled-phase gate, then \(G\) commutes with all \(Z\)-type Paulis
    \(Z\otimes I\), \(I\otimes Z\), and \(Z\otimes Z\).  Given any local Pauli \(B\), choose
    \(R\in U(2)\otimes U(2)\) such that \(RBR^\dagger\) is \(Z\)-type, and set
    \(X_\star=R^\dagger G R\).  Then \(X_\star\in\operatorname{supp}\nu_{\rm loc}\), and
    \begin{align}
        X_\star^\dagger B X_\star
        &=
        R^\dagger G^\dagger(RBR^\dagger)G R
         =
        B .
        \label{eq:controlled-phase-local-gate-set-condition}
    \end{align}
    Hence the stabilizing-element criterion holds.  The same argument applies to CNOT, using the
    Pauli directions stabilized by CNOT, for instance \(Z\otimes I\), \(I\otimes X\), and
    \(Z\otimes X\).  Generic controlled-phase gates need not be Clifford; the criterion only uses
    the existence of a Pauli direction stabilized by the fixed entangler, after local basis changes.
    \item \emph{iSWAP/fSim-type gates interleaved with one-qubit rotations.}
These gates are motivated by hardware implementations, including the Google
OTOC experiment, which uses fixed iSWAP-like two-qubit gates interleaved
with random one-qubit gates
~\cite{abanin2025constructiveinterferenceedgequantum,
king2025simplifiedversionquantumotoc2}. Consider the two-qubit block \(X=K_2GK_1GK_0\), where
\(G={\rm iSWAP}\) or \(G={\rm FSim}(\pi/2,\varphi)\), and
\(K_j=u_{j,1}\otimes u_{j,2}\). Assume that the six one-qubit gates
\(u_{j,a}\) are independent and their distributions have full support on
\(\mathrm U(2)\), as for Haar-random one-qubit rotations.
For either choice of \(G\), its square is diagonal: explicitly,
\({\rm iSWAP}^2=Z\otimes Z\) and
\({\rm FSim}(\pi/2,\varphi)^2
=\operatorname{diag}(1,-1,-1,e^{-2i\varphi})\).
Given a two-qubit Pauli \(B\), choose a local unitary \(R\) such that
\(RBR^\dagger\) is \(Z\)-type. The choices
\(K_2=R^\dagger\), \(K_1=I\), and \(K_0=R\) belong to the joint support
and give \(X_\star=R^\dagger G^2R\). Since \(G^2\) commutes with every
\(Z\)-type Pauli,
\begin{align}
    X_\star^\dagger B X_\star
    &=
    R^\dagger(G^2)^\dagger(RBR^\dagger)G^2R
    =
    B .
\end{align}
Thus the block satisfies the stabilizing-element criterion.
For one-qubit distributions without full support, the same argument
applies whenever the required choices of \(K_0,K_1,K_2\) belong to their
joint support.
This establishes the local reverse-variance condition for the two-entangler blocked ensemble. We do not claim here that our main theorem applies directly to a single iSWAP/fSim layer.

 
    \item \emph{Hamiltonian local ensembles.}
Consider \(X(h)=\exp(-i\sum_j h_jP_j)\), where the \(P_j\) are fixed
Hermitian operators and \(h\in\mathbb R^m\) is sampled from a probability
measure with support \(S\). If \(0\in S\), continuity implies
\(I=X(0)\in\operatorname{supp}\nu_{\rm loc}\), so the
identity-in-the-support criterion applies. More generally, the
stabilizing-element criterion holds whenever \(S\) contains
\(h_\star\) such that \(X(h_\star)^\dagger B X(h_\star)=\pm B\).

    \item \emph{General fixed entangling-gate interleaved with one-qubit rotations.} Assume that the one-qubit rotations are sampled independently from fixed distributions with full support on \(\mathrm U(2)\).
    A single block of the form
    \(X=(u_1\otimes u_2)G(v_1\otimes v_2)\), with \(G\) entangling, need not contain an element
    that stabilizes the particular Pauli \(B\).  However, after grouping a constant number of such
    blocks, the open-support criterion applies.  Indeed, exact universality of entangling two-qubit
    gates with arbitrary one-qubit gates implies that, for some finite blocking length depending only
    on \(G\), the support of the blocked ensemble contains a nonempty open set of two-qubit
    unitaries, up to global phases~\cite{brylinski2002universal,bremner2002practical,zhang2003exact}.
    Hence Lemma~\ref{lem:two-local-criteria-reverse-variance} gives a reverse-variance constant
    depending only on the blocked local ensemble and on \(k\), but not on \(n\).

\end{itemize}

Finally, the same argument applies when, as in our proof, a few consecutive individual local gates in
the active backward light cone of \(B\) are grouped into a single effective local ensemble \(X\).  In
this case \(\nu_{\rm loc}\) is the product distribution of the grouped local gates.  Since both the
number of grouped gates and the size of the active patch are bounded independently of \(n\), it is
enough to verify either the stabilizing-element criterion or the open-support criterion for this
grouped ensemble.  For example, for constant-size active patches built from CZ, CNOT, or
controlled-phase local ensembles interleaved with Haar one-qubit rotations, one can choose supported
local gates whose product stabilizes the local Pauli \(B\) up to sign.  Hence the finite-dimensional
reverse-variance argument gives a constant depending only on the grouped local ensemble and on \(k\),
but not on the total number of qubits.

\subsubsection{Explicit constants for product-Haar local ensembles}
\label{app:explicit-product-haar-local-constants}

For the Haar examples in
Section~\ref{app:examples-local-reverse-variance-criteria}, the reverse-variance constant can be
made explicit.  This includes Haar two-qubit gates, as well as fixed two-qubit gates interleaved
with Haar-random one-qubit rotations.  We use the following standard consequence of Schur
orthogonality.

\begin{lemma}[Haar matrix-coefficient reverse variance]
\label{lem:haar-matrix-coefficient-reverse-variance}
Let \(\mathsf G\) be a compact group with normalized Haar measure, let
\(\pi:\mathsf G\to\mathrm U(\mathcal V)\) be a finite-dimensional unitary representation, and let
\(\mathcal M(\pi)\) be the span of the matrix coefficients of \(\pi\).  Then, for every
\(F\in\mathcal M(\pi)\),
\begin{align}
    \Var_{g\sim{\rm Haar}(\mathsf G)}F(g)
    &\ge
    (\dim\mathcal V)^{-2}
    \left|
        \E_g F(g)-F(e)
    \right|^2.
    \label{eq:haar-matrix-coefficient-reverse-variance}
\end{align}
\end{lemma}

\begin{proof}
Set \(p(g)\coloneqq F(g)-\E_gF(g)\).  Then
\begin{align}
    \Var_g F(g)
    &=
    \E_g|p(g)|^2,
    &
    \E_gF(g)-F(e)
    &=
    -p(e).
    \label{eq:p-variance-identity}
\end{align}
Decompose the representation as
\begin{align}
    \mathcal V
    &\simeq
    \bigoplus_{\lambda\in\Lambda}
    \mathcal V_\lambda\otimes\mathbb C^{m_\lambda},
    &
    \pi(g)
    &\simeq
    \bigoplus_{\lambda\in\Lambda}
    \pi_\lambda(g)\otimes I_{m_\lambda},
\end{align}
where the \(\pi_\lambda\)'s are inequivalent irreducible representations, and write
\(d_\lambda\coloneqq \dim\mathcal V_\lambda\).  The space \(\mathcal M(\pi)\) is spanned by the matrix
coefficients \(u_{ij}^{(\lambda)}(g)\) of the irreducible representations \(\pi_\lambda\) which
appear in \(\pi\).  Since \(p\) has zero Haar average, its trivial component vanishes.  Hence we may
write
\begin{align}
    p(g)
    &=
    \sum_{\lambda\neq{\rm triv}}
    \sum_{i,j=1}^{d_\lambda}
    c_{ij}^{(\lambda)}
    \sqrt{d_\lambda}\,
    u_{ij}^{(\lambda)}(g).
    \label{eq:p-schur-expansion}
\end{align}
By Schur orthogonality, the functions
\(\sqrt{d_\lambda}\,u_{ij}^{(\lambda)}\) are orthonormal in \(L^2(\mathsf G)\).  Therefore
\begin{align}
    \E_g|p(g)|^2
    &=
    \sum_{\lambda\neq{\rm triv}}
    \sum_{i,j=1}^{d_\lambda}
    \left|
        c_{ij}^{(\lambda)}
    \right|^2 .
    \label{eq:p-L2-haar}
\end{align}
On the other hand, \(u_{ij}^{(\lambda)}(e)=\delta_{ij}\), so evaluating
Eq.~\eqref{eq:p-schur-expansion} at the identity gives
\begin{align}
    p(e)
    &=
    \sum_{\lambda\neq{\rm triv}}
    \sum_{i=1}^{d_\lambda}
    c_{ii}^{(\lambda)}
    \sqrt{d_\lambda}.
    \label{eq:p-identity-value}
\end{align}
Applying Cauchy--Schwarz to Eq.~\eqref{eq:p-identity-value},
\begin{align}
    |p(e)|^2
    &\le
    \left(
        \sum_{\lambda\neq{\rm triv}}
        \sum_{i=1}^{d_\lambda}
        \left|
            c_{ii}^{(\lambda)}
        \right|^2
    \right)
    \left(
        \sum_{\lambda\neq{\rm triv}}
        d_\lambda^2
    \right)
    \notag\\
    &\le
    \E_g|p(g)|^2
    \left(
        \sum_{\lambda\in\Lambda}
        d_\lambda
    \right)^2
    \notag\\
    &\le
    \E_g|p(g)|^2
    \left(
        \sum_{\lambda\in\Lambda}
        m_\lambda d_\lambda
    \right)^2
    =
    (\dim\mathcal V)^2\E_g|p(g)|^2.
    \label{eq:pe-bound}
\end{align}
Combining Eq.~\eqref{eq:pe-bound} with Eq.~\eqref{eq:p-variance-identity} gives
\[
    \left|
        \E_gF(g)-F(e)
    \right|^2
    =
    |p(e)|^2
    \le
    (\dim\mathcal V)^2 \Var_gF(g),
\]
which is equivalent to the claimed bound.
\end{proof}

\begin{corollary}[Explicit constants for product-Haar local blocks]
\label{cor:explicit-product-haar-local-blocks}
Let the active local block be parameterized by independent Haar variables
\(h=(h_1,\ldots,h_s)\in H\coloneqq \prod_{a=1}^s \mathrm U(d_a)\), and write the corresponding local
unitary as \(X(h)\).  Suppose that, for every conditional OTOC function \(f_V\), the function
\(F(h)\coloneqq f_V(X(h))\) is a matrix coefficient of
\begin{align}
    \pi(h_1,\ldots,h_s)
    &\coloneqq 
    \bigotimes_{a=1}^s
    \left(
        h_a^{\otimes 2k}
        \otimes
        \overline h_a^{\otimes 2k}
    \right).
    \label{eq:product-haar-representation}
\end{align}
Assume moreover that there exists \(h_\star\in H\) such that
\(X(h_\star)^\dagger B X(h_\star)=\pm B\).  Then
\begin{align}
    \Var_h f_V(X(h))
    &\ge
    \left(
        \prod_{a=1}^s d_a^{-8k}
    \right)
    \left|
        \E_h f_V(X(h))-f_V(I)
    \right|^2 .
    \label{eq:product-haar-local-reverse-variance}
\end{align}
\end{corollary}

\begin{proof}
The representation space in Eq.~\eqref{eq:product-haar-representation} has dimension
\[
    \dim\mathcal V
    =
    \prod_{a=1}^s d_a^{4k}.
\]
Apply Lemma~\ref{lem:haar-matrix-coefficient-reverse-variance} to the shifted matrix coefficient
\(\widetilde F(h)\coloneqq F(h_\star h)\).  By Haar invariance,
\(\Var_h\widetilde F(h)=\Var_hF(h)\) and \(\E_h\widetilde F(h)=\E_hF(h)\), while
\(\widetilde F(e)=F(h_\star)\).  Therefore
\begin{align}
    \Var_h F(h)
    &\ge
    \left(
        \prod_{a=1}^s d_a^{-8k}
    \right)
    \left|
        \E_hF(h)-F(h_\star)
    \right|^2 .
    \label{eq:shifted-product-haar-reverse-variance}
\end{align}
Finally, \(X(h_\star)^\dagger B X(h_\star)=\pm B\), and the sign cancels in the \(2k\)-fold OTOC.
Thus \(F(h_\star)=f_V(I)\), giving Eq.~\eqref{eq:product-haar-local-reverse-variance}.
\end{proof}
We now spell out the constants for the examples above.

\begin{itemize}
    \item \emph{Haar two-qubit gates.}
    If the active block contains \(m\) independent Haar-random two-qubit gates, then
    \(s=m\) and \(d_a=4\) for every \(a\).  Since the identity is in the support, we may take
    \(h_\star=e\).  Hence
    \begin{align}
        \Var f_V
        &\ge
        4^{-8km}
        \left|
            \E f_V-f_V(I)
        \right|^2 .
        \label{eq:two-qubit-haar-explicit-constant}
    \end{align}
    In particular, for one Haar two-qubit gate the constant is \(4^{-8k}=2^{-16k}\).

    \item \emph{Controlled-phase, CZ, and CNOT with Haar one-qubit rotations.}
    Consider
    \[
        X=(u_1\otimes u_2)\,G\,(v_1\otimes v_2),
        \qquad
        u_i,v_i\sim{\rm Haar}(\mathrm U(2)),
    \]
    with \(G\) a controlled-phase gate, CZ, or CNOT.  As in
    Section~\ref{app:examples-local-reverse-variance-criteria}, there is a choice
    \(h_\star=(u_{1,\star},u_{2,\star},v_{1,\star},v_{2,\star})\) such that
    \(X(h_\star)^\dagger B X(h_\star)=\pm B\).  Here there are \(s=4\) independent one-qubit Haar
    variables, all with \(d_a=2\).  Corollary~\ref{cor:explicit-product-haar-local-blocks} gives
    \begin{align}
        \Var f_V
        &\ge
        2^{-32k}
        \left|
            \E f_V-f_V(I)
        \right|^2 .
        \label{eq:controlled-entangler-explicit-constant}
    \end{align}

    \item \emph{iSWAP/fSim blocks with Haar one-qubit layers.}
    Consider a two-entangler block
    \[
        X=K_2G K_1G K_0,
    \]
    where \(G={\rm iSWAP}\) or \(G={\rm FSim}(\pi/2,\varphi)\), and each \(K_j\) is an independent
    Haar-random one-qubit layer on the two sites.  Each layer contributes two independent
    \(\mathrm U(2)\) variables, so \(s=6\) and \(d_a=2\).  Since the blocked ensemble contains a
    stabilizing element for \(B\), as explained in
    Section~\ref{app:examples-local-reverse-variance-criteria}, we obtain
    \begin{align}
        \Var f_V
        &\ge
        2^{-48k}
        \left|
            \E f_V-f_V(I)
        \right|^2 .
        \label{eq:iswap-fsim-explicit-constant}
    \end{align}

    \item \emph{General fixed-gate blocks with Haar one-qubit layers and a stabilizing element.}
    More generally, suppose a constant-size two-qubit block contains \(L\) fixed two-qubit gates and
    \(L+1\) independent Haar-random one-qubit layers, and suppose the block satisfies the
    stabilizing-element criterion for the relevant local Pauli \(B\).  Then \(s=2(L+1)\), all
    \(d_a=2\), and
    \begin{align}
        \Var f_V
        &\ge
        2^{-16k(L+1)}
        \left|
            \E f_V-f_V(I)
        \right|^2 .
        \label{eq:general-fixed-gate-block-explicit-constant}
    \end{align}
\end{itemize}

The explicit constants above apply to product-Haar examples verified by the stabilizing-element
criterion.  For the generic fixed-entangler examples covered by the open-support criterion, the
finite-dimensional argument still gives a constant depending only on the blocked local ensemble and
on \(k\), but we do not use the product-Haar formula above to make it explicit.

\subsection{Forward persistence of variance}
\label{subsec:forward-persistence}

We now show that the local reverse-variance bound has a useful consequence:
adding one more random layer cannot erase existing circuit-to-circuit
fluctuations by more than a constant factor.

Write
\(F_d:=\mathrm{OTOC}^{(k)}_\rho(U_d)\) and
\(\Var Z:=\E|Z-\E Z|^2\).
After conditioning on the depth-\(d\) circuit \(U_d=V\), the only remaining
randomness in \(F_{d+1}\) comes from the last layer. If \(X\) denotes the effective local unitary generated by the gates in the
active part of the last layer, then the notation of the previous subsection gives \(F_{d+1}=f_V(X)\),
\(\E[F_{d+1}\mid U_d=V]=\E_X f_V(X)\), and
\(F_d(V)=f_V(I)\). Thus the local reverse-variance bound is equivalently
\begin{align}
    \Var(F_{d+1}\mid U_d)
    &\ge
    \eta\,
    \left|
        \E[F_{d+1}\mid U_d]-F_d
    \right|^2 ,
    \label{eq:one-step-local-reverse-variance}
\end{align}
for some constant \(\eta>0\) independent of the system size.

This condition has a simple interpretation. If averaging over the new layer
changes the OTOC substantially relative to \(F_d\), then the new layer must
itself generate fluctuations. Consequently, a pre-existing variance cannot
disappear in a single step.

\begin{lemma}[Forward persistence of variance]
\label{lem:forward-persistence-one-layer}
Suppose Eq.~\eqref{eq:one-step-local-reverse-variance} holds uniformly at
every depth with some \(\eta>0\). Then, for every \(d\),
\begin{align}
    \Var(F_{d+1})
    &\ge
    \kappa\,\Var(F_d),
    \qquad
    \kappa:=\frac{\eta}{1+\eta}>0.
    \label{eq:one-step-forward-persistence}
\end{align}
\end{lemma}

\begin{proof}
Let \(m:=\E[F_{d+1}\mid U_d]\) be the conditional mean after averaging
over the last layer, and let \(c:=\E F_{d+1}\). By the law of total variance,
\begin{align}
    \Var(F_{d+1})
    &=
    \E\Var(F_{d+1}\mid U_d)
    +
    \E|m-c|^2.
    \label{eq:forward-persistence-total-variance}
\end{align}
The reverse-variance assumption gives, for every conditioned circuit \(U_d\),
\(\Var(F_{d+1}\mid U_d)\ge\eta|m-F_d|^2\). Averaging this inequality and
substituting it above yields
\begin{align}
    \Var(F_{d+1})
    &\ge
    \E\!\left[
        \eta|m-F_d|^2+|m-c|^2
    \right].
    \label{eq:forward-persistence-two-errors}
\end{align}

We now eliminate the intermediate quantity \(m\). For any complex numbers
\(m,x,c\), completing the square gives
\begin{align}
    \eta|m-x|^2+|m-c|^2
    &=
    (1+\eta)
    \left|
        m-\frac{\eta x+c}{1+\eta}
    \right|^2
    +
    \frac{\eta}{1+\eta}|x-c|^2
    \nonumber\\
    &\ge
    \frac{\eta}{1+\eta}|x-c|^2.
    \label{eq:forward-persistence-completing-square}
\end{align}
Applying this pointwise with \(x=F_d\) in
Eq.~\eqref{eq:forward-persistence-two-errors} gives
\(\Var(F_{d+1})\ge\kappa\,\E|F_d-c|^2\), where
\(\kappa=\eta/(1+\eta)\).

Finally, the mean \(\E F_d\) minimizes the mean squared distance from
\(F_d\) to a constant. Explicitly,
\(\E|F_d-c|^2
=\Var(F_d)+|\E F_d-c|^2\ge\Var(F_d)\).
Therefore
\(\Var(F_{d+1})\ge\kappa\,\Var(F_d)\), as claimed.
\end{proof}

\subsection{Putting the pieces together: large OTOC fluctuations in the transition depth window}
\label{app:putting-pieces-together-transition-window}

With the ingredients established in the previous subsections, the same mechanism as the one-dimensional Haar-brickwork
\(\mathrm{OTOC}^{(2)}\) result of Theorem~\ref{thm:one-dimensional-haar-brickwork-otoc2} can be established for \(\mathrm{OTOC}^{(k)}\) as well in more generality architecture and ensembles.  In
particular, any bounded-degree local architecture whose random-circuit ensemble reaches the relevant
fixed-order OTOC-test moment control within polynomial depth, and whose local gate ensemble satisfies
the reverse-variance condition, exhibits inverse-polynomial \(\mathrm{OTOC}^{(k)}\) fluctuations on
an interval of depths inside the transition window.  The required OTOC-test moment control can be
certified using the design-convergence criteria discussed in
Section~\ref{app:moment-control-design-depth}, while the reverse-variance input is verified for the
local gate ensembles described in Section~\ref{app:local-reverse-variance}.

For completeness, we now state the result in this general form and briefly repeat the proof,
emphasizing how the ingredients are combined.
 
\begin{theorem}[Fixed-order OTOC fluctuations for bounded-degree random circuits]
\label{thm:fixed-k-transition-window-fluctuation-bound}
Fix \(k\ge1\).  Let \(\mathcal A\) be a bounded-degree local random-circuit architecture whose
ensembles reach \((k,1/4)\)-OTOC-test moment control within \(\operatorname{poly}(n)\) depth
(see Section~\ref{app:moment-control-design-depth} for examples), and assume that every one-layer
active patch satisfies the local reverse-variance condition with an \(n\)-independent constant
(see Section~\ref{app:local-reverse-variance} for examples).

Then, for every fixed integer \(R\ge0\), the following holds.  For all sufficiently large \(n\), for
every fixed \(n\)-qubit state \(\rho\) and local Pauli observables \(B,M\), there exists an interval
\(I\) of \(R+1\) consecutive depths such that
\begin{align}
    I
    &\subseteq
    \{d_{\rm lc}+1,\ldots,d_{\rm mc}+R\},
\end{align}
where \(d_{\rm lc}\) is the last pre-light-cone depth and \(d_{\rm mc}\) is a polynomial-depth
OTOC-test moment-control depth.  Moreover, for every \(d\in I\),
\begin{align}
    \Var_{U_d}
    \!\left(
        \mathrm{OTOC}^{(k)}_\rho(U_d)
    \right)
    &=
    \Omega
    \!\left(
        \frac{1}{\operatorname{poly}(n)}
    \right).
\end{align}
\end{theorem}

\begin{proof}
Let \(d_{\rm lc}\) be the last pre-light-cone depth, and let \(d_{\rm mc}\) be a depth at which
\((k,1/4)\)-OTOC-test moment control holds.  Write
\(F_d\coloneqq \mathrm{OTOC}^{(k)}_\rho(U_d)\), \(\mu_d\coloneqq \E F_d\), and
\(W\coloneqq d_{\rm mc}-d_{\rm lc}\).  By assumption, \(W\le p(n)\) for some polynomial \(p\).
Lemma~\ref{lem:pre-light-cone-identity} gives \(\mu_{d_{\rm lc}}=1\).  Let
\(h_{\rho,k}\coloneqq \E_{U\sim{\rm Haar}}\mathrm{OTOC}^{(k)}_\rho(U)\).  By \((k,1/4)\)-OTOC-test
moment control, \(|\mu_{d_{\rm mc}}-h_{\rho,k}|\le1/4\), while
Proposition~\ref{prop:haar-otock-small} gives \(|h_{\rho,k}|\le1/4\) for fixed \(k\) and all
sufficiently large \(n\).  Hence
\begin{align}
    |\mu_{d_{\rm mc}}-\mu_{d_{\rm lc}}|
    &=
    |\mu_{d_{\rm mc}}-1|
    \ge
    |1-h_{\rho,k}|-|\mu_{d_{\rm mc}}-h_{\rho,k}|
    \ge
    1-\frac14-\frac14
    =
    \frac12 .
    \label{eq:fixed-k-proof-order-one-mean-change}
\end{align}

By telescoping and the triangle inequality,
\begin{align}
    \frac12
    &\le
    |\mu_{d_{\rm mc}}-\mu_{d_{\rm lc}}|
    =
    \left|
        \sum_{d=d_{\rm lc}+1}^{d_{\rm mc}}
        (\mu_d-\mu_{d-1})
    \right|
    \le
    W
    \max_{d\in\{d_{\rm lc}+1,\ldots,d_{\rm mc}\}}
    |\mu_d-\mu_{d-1}| .
\end{align}
Therefore there exists \(d_\star\in\{d_{\rm lc}+1,\ldots,d_{\rm mc}\}\) such that
\begin{align}
    |\mu_{d_\star}-\mu_{d_\star-1}|
    &\ge
    \frac{1}{2W}.
    \label{eq:fixed-k-proof-large-increment}
\end{align}

Condition on \(U_{d_\star-1}=V\), expose the active one-layer patch \(X\), and set
\(f_V(X)\coloneqq \mathrm{OTOC}^{(k)}_\rho(L_{d_\star}(X)V)\).  As in the proof of
Theorem~\ref{thm:one-dimensional-haar-brickwork-otoc2}, replacing the active patch by the identity
gives \(f_V(I)=F_{d_\star-1}(V)\), while averaging over \(X\) gives
\(\E_X f_V(X)=\E[F_{d_\star}\mid U_{d_\star-1}=V]\).  Let \(\eta>0\) be the
\(n\)-independent local reverse-variance constant.  The local reverse-variance estimate, the law of
total variance, and Jensen's inequality give
\begin{align}
    \Var(F_{d_\star})
    &\ge
    \eta
    |\mu_{d_\star}-\mu_{d_\star-1}|^2
    \ge
    \frac{\eta}{4W^2}.
    \label{eq:fixed-k-proof-first-variance-depth}
\end{align}

It remains to propagate this variance forward. By Lemma~\ref{lem:forward-persistence-one-layer}, with
\(\kappa:=\eta/(1+\eta)\), we have
\(\Var(F_{d+1})\ge\kappa\Var(F_d)\) for every depth \(d\).  Iterating from \(d_\star\), for every
\(r=0,\ldots,R\),
\begin{align}
    \Var(F_{d_\star+r})
    &\ge
    \kappa^r\Var(F_{d_\star})
    \ge
    \frac{\eta\,\kappa^r}{4W^2}.
    \label{eq:fixed-k-proof-pointwise-interval}
\end{align}
Set \(I\coloneqq \{d_\star,d_\star+1,\ldots,d_\star+R\}\).  Then
\(I\subseteq\{d_{\rm lc}+1,\ldots,d_{\rm mc}+R\}\).  Since \(0<\kappa<1\), \(r\le R\), and
\(W\le p(n)\), Eq.~\eqref{eq:fixed-k-proof-pointwise-interval} gives, for every \(d\in I\),
\begin{align}
    \Var(F_d)
    &\ge
    \frac{\eta\,\kappa^R}{4W^2}
    \ge
    \frac{\eta\,\kappa^R}{4p(n)^2}.
\end{align}
For fixed \(k,R\) and fixed architecture, the prefactor is independent of \(n\).  Since \(p\) is a
polynomial, this is an inverse-polynomial lower bound.
\end{proof}

\begin{remark}[From the averaged OTOC profile to fluctuations]
\label{rem:variance-bound-mean-slope}
The argument above reveals a more general principle: a sufficiently rapid
change of the ensemble-averaged OTOC necessarily produces
circuit-to-circuit fluctuations.

Indeed, write
\(F_d\coloneqq\mathrm{OTOC}^{(k)}_\rho(U_d)\) and
\(\mu_d\coloneqq\E F_d\), and suppose that, for some depths \(a<b\),
\begin{align}
    |\mu_b-\mu_a|
    &\ge
    \Delta,
    \qquad
    W\coloneqq b-a .
\end{align}
By telescoping, there must exist
\(d_\star\in\{a+1,\ldots,b\}\) for which
\(|\mu_{d_\star}-\mu_{d_\star-1}|\ge\Delta/W\).
The local reverse-variance bound then converts this change of the mean
into fluctuations,
\begin{align}
    \Var(F_{d_\star})
    &\ge
    \eta\,\frac{\Delta^2}{W^2}.
\end{align}
Moreover, forward persistence implies that, for every fixed \(R\ge0\),
the same inverse-polynomial scaling, up to an \(n\)-independent constant
factor, holds throughout
\(\{d_\star,\ldots,d_\star+R\}\).

Thus, the use of Haar convergence in
Theorem~\ref{thm:fixed-k-transition-window-fluctuation-bound} is not
essential to the mechanism. It provides a convenient way of certifying
an order-one change of the averaged OTOC between the pre-light-cone
regime and a later depth. If this change occurs over \(W\) layers, the
argument certifies fluctuations of order at least \(W^{-2}\).

More detailed information about the transition profile can therefore
lead directly to stronger fluctuation bounds. In particular, if an
order-one change of the averaged OTOC can be localized to a narrower
depth window, then the same argument yields a correspondingly larger
variance lower bound. We will see precisely this improvement in the next
section for \(\mathrm{OTOC}^{(1)}\) in one-dimensional Haar brickwork circuits.
\end{remark}

\newpage

\section{\texorpdfstring{\(\mathrm{OTOC}^{(1)}\)}{OTOC(1)} in one-dimensional Haar brickwork circuit}
\label{subsec:endpoint-otoc1-setup}

We now specialize to \texorpdfstring{\(\mathrm{OTOC}^{(1)}\)}{OTOC(1)} in a one-dimensional chain of \(n\) qubits, with \(n\) even. The
depth-\(d\) circuit is \(U_d=L_1L_2\cdots L_d\), with \(d\) even. Odd layers act on
\((1,2),(3,4),\ldots,(n-1,n)\), and even layers act on
\((2,3),(4,5),\ldots,(n-2,n-1)\). All two-qubit gates are independent
Haar-random gates in \(U(4)\).

We take the butterfly operator to be \(Z_1\) and the measurement operator to be \(Z_n\). Since qubits \(1\) and \(n\) are the endpoints of the one-dimensional chain of \(n\) qubits, we call the resulting OTOC the \emph{endpoint \texorpdfstring{\(\mathrm{OTOC}^{(1)}\)}{OTOC(1)}}. Thus the endpoint \texorpdfstring{\(\mathrm{OTOC}^{(1)}\)}{OTOC(1)} takes the form
\begin{equation}
\label{eq:endpoint-otoc1-def}
    \mathrm{OTOC}^{(1)}_\rho(U_d)
    \coloneqq
    \Tr\!\left[
        \rho\,
        \bigl(U_d^\dagger Z_1U_d\,Z_n\bigr)^2
    \right].
\end{equation}
Below, we find an exact expression for the averaged value of this OTOC, namely~\(\E_{U_d}\mathrm{OTOC}^{(1)}_\rho(U_d)\).

\subsection{Exact formula for the averaged endpoint \texorpdfstring{\(\mathrm{OTOC}^{(1)}\)}{OTOC(1)}}
\label{subsec:endpoint-otoc1-exact-formula}
In Lemma~\ref{thm:endpoint-otoc1-exact-formula} below, we record an exact formula for the averaged endpoint \(\mathrm{OTOC}^{(1)}\). The starting point of the proof is the operator-spreading picture for
Haar-random local circuits developed in the seminal works
\cite{nahum2018operator,keyserlingk2018operator}. In this picture, after
averaging over the Haar gates, the squared Pauli weights of a Heisenberg-evolved
local operator evolve according to a classical Markov process~\cite{BensaZnidaric2022}. For the endpoint
OTOC considered here, the averaged value is determined by the distribution of
the right endpoint of the evolved Pauli string, namely the largest site \(i\)
on which the Pauli string is nontrivial. Specifically, we will show that this
endpoint evolves as a biased random walk with open-boundary effects. In the
language of Refs.~\cite{nahum2018operator,keyserlingk2018operator}, the drift
of this walk gives the butterfly velocity \(v_B=3/5\).

We exploit this structure to solve the finite-chain Markov chain analytically.
This gives an exact formula for the averaged endpoint
\(\mathrm{OTOC}^{(1)}\) for every finite \(n\), going beyond the asymptotic
front picture of Refs.~\cite{nahum2018operator,keyserlingk2018operator}. This
exact expression will be important for the sensitivity and variance analysis
below.

\begin{lemma}[Exact averaged endpoint \(\mathrm{OTOC}^{(1)}\)]
\label{thm:endpoint-otoc1-exact-formula}
Assume that \(n\ge2\) and \(d\ge0\) are even. Then, for every density matrix
\(\rho\),
\begin{equation}
\label{eq:endpoint-otoc1-exact-formula}
\E_{U_d}\mathrm{OTOC}^{(1)}_\rho(U_d)
=
1-
\sum_{j=0}^{\left\lfloor\frac{d-n}{2n}\right\rfloor}
4^{-nj}\,
\Psi_d\!\left(
\frac{d-n}{2}-nj
\right),
\end{equation}
where  
\begin{equation}
\label{eq:Psi-d-def}
\Psi_d(r)
\coloneqq
\Pr(X_d\le r)
+
\frac23\,\Pr(X_d=r)
-
\sum_{s=1}^{r}16^{-s}\Pr(X_d=r-s)
\end{equation}
with \(X_d\sim\operatorname{Bin}(d,1/5)\)---that is,
\(\Pr(X_d=m)=\binom dm(1/5)^m(4/5)^{d-m}\) and
\(\Pr(X_d\le r)=\sum_{m=0}^{r}\binom dm(1/5)^m(4/5)^{d-m}\).
In particular, if \(d<n\), then
\(\E_{U_d}\mathrm{OTOC}^{(1)}_\rho(U_d)=1\).
\end{lemma}
\begin{proof}
By Lemma~\ref{lem:state-independence-higher-otoc}, the average is independent
of \(\rho\), and therefore we are free to compute it using \(\rho=2^{-n}I\). Thus,  $\E_{U_d}\mathrm{OTOC}^{(1)}_\rho(U_d)
=
\E_{U_d}\left[
2^{-n}\Tr\!\left(
(U_d^\dagger Z_1U_dZ_n)^2
\right)
\right]$. The proof has three steps.

\medskip
\noindent
\textbf{Step 1: reduction to the endpoint weight.}

Expand the evolved operator in the Pauli basis: $U_d^\dagger Z_1U_d=\sum_P c_P(U_d^\dagger Z_1 U_d)P$ with $c_P(U_d^\dagger Z_1 U_d)\coloneqq 2^{-n}\Tr\!\left(PU_d^\dagger Z_1U_d\right)$, where \(P=P_1\otimes\cdots\otimes P_n\) runs over all \(n\)-qubit Pauli
strings. Since \(U_d^\dagger Z_1U_d\) is Hermitian, the coefficients
\(c_P(U_d^\dagger Z_1 U_d)\) are real. Moreover, Pauli orthogonality gives $\sum_P\left(c_P(U_d^\dagger Z_1 U_d)\right)^2
= 
1$.

For a Pauli string \(P\), the sign of \(2^{-n}\Tr(PZ_nPZ_n)\) depends only on
the last single-site Pauli \(P_n\). If \(P_n=I\) or \(P_n=Z\), then \(P\)
commutes with \(Z_n\), and this trace equals \(1\). If \(P_n=X\) or \(P_n=Y\),
then \(P\) anticommutes with \(Z_n\), and the trace equals \(-1\). Using Pauli
orthogonality once more, we get $2^{-n}\Tr\!\left[(U_d^\dagger Z_1U_dZ_n)^2\right]
=
1-
2\sum_{P:\,P_n=X\text{ or }Y}
\left(c_P(U_d^\dagger Z_1 U_d)\right)^2$. Therefore
\begin{equation}
\label{eq:otoc1-Cn-pauli-last-site}
1-\E_{U_d}\mathrm{OTOC}^{(1)}_\rho(U_d)
=
2\,\E_{U_d}
\sum_{P:\,P_n=X\text{ or }Y}
\left(c_P(U_d^\dagger Z_1 U_d)\right)^2 .
\end{equation}

We now rewrite the right-hand side in terms of the total averaged weight of
Pauli strings whose right endpoint is the last site. For a Pauli string
\(P=P_1\otimes\cdots\otimes P_n\), the condition \(P_n\neq I\) means precisely
that the Pauli string has reached site \(n\). Since \(n\) is the right endpoint
of the chain, this is the same as saying that the right endpoint of the string
is at site \(n\).

We claim that, after Haar averaging,
\begin{equation}
\label{eq:otoc1-XY-weight-two-thirds-endpoint-weight}
\E_{U_d}
\sum_{P:\,P_n=X\text{ or }Y}
\left(c_P(U_d^\dagger Z_1 U_d)\right)^2
=
\frac23\,
\E_{U_d}
\sum_{P:\,P_n\neq I}
\left(c_P(U_d^\dagger Z_1 U_d)\right)^2 .
\end{equation}
Indeed, condition on all gates except the last Haar gate which touches site
\(n\). Since \(n\) and \(d\) are even, this is a gate on the boundary bond
\((n-1,n)\) in the last odd layer before time \(d\); the final even layer does
not touch site \(n\). Conditional on the event that the Pauli string is
nontrivial on site \(n\) after this boundary gate, Haar invariance makes the
local Pauli direction at site \(n\) uniform among \(X,Y,Z\). Therefore, among
the total weight with \(P_n\neq I\), a fraction \(2/3\) has \(P_n=X\) or
\(P_n=Y\). This proves
\eqref{eq:otoc1-XY-weight-two-thirds-endpoint-weight}.

Substituting \eqref{eq:otoc1-XY-weight-two-thirds-endpoint-weight} into
\eqref{eq:otoc1-Cn-pauli-last-site}, we obtain
\begin{equation}
\label{eq:otoc1-Cn-endpoint-weight}
1-\E_{U_d}\mathrm{OTOC}^{(1)}_\rho(U_d)
=
\frac43\,
\E_{U_d}
\sum_{P:\,P_n\neq I}
\left(c_P(U_d^\dagger Z_1 U_d)\right)^2 .
\end{equation}
Thus the averaged OTOC can be expressed in terms of the averaged squared Pauli weight of
strings whose right endpoint is the last site.

\medskip
\noindent
\textbf{Step 2: the endpoint Markov matrix.}

We now compute the endpoint weight appearing in
\eqref{eq:otoc1-Cn-endpoint-weight}. Let us recall the Haar rule for a
single two-qubit gate~\cite{nahum2018operator,keyserlingk2018operator}. The identity \(I\otimes I\) remains fixed. If the input
on a two-qubit bond is a nonidentity Pauli, then, after Haar averaging, the
squared Pauli coefficients are uniformly distributed over the \(15\)
nonidentity two-qubit Paulis.

Among these \(15\) Paulis, \(12\) are nontrivial on the right qubit and \(3\)
are nontrivial only on the left qubit. Therefore, if the right endpoint is on
the left qubit of an active bond, it moves one site to the right with
probability \(12/15=4/5\), and it stays on the left qubit with probability
\(1/5\). Similarly, if the endpoint is on the right qubit of an active bond, it
moves one site to the left with probability \(1/5\), and it stays on the right
qubit with probability \(4/5\).

Since \(n\) is even, the odd layers act inside the two-site cells
\[
\mathcal D_a\coloneqq\{2a-1,2a\},
\qquad
a=1,\ldots,\frac n2 .
\]
We observe the endpoint immediately after odd layers. This is the convenient
observation time because an odd layer acts inside every cell. After such a
layer, conditional on the endpoint being in \(\mathcal D_a\), the endpoint is
at the right site \(2a\) with probability \(4/5\), and at the left site
\(2a-1\) with probability \(1/5\).

Let \(Y_m\) be the cell containing the endpoint immediately after the \(m\)-th
odd layer. After the first odd layer, the endpoint is in the first cell, so
\(Y_1=1\). One effective step of \(Y_m\) consists of an even layer followed by
the next odd layer. The even layer can move the endpoint from one cell to a
neighbouring cell, while the next odd layer only resets the internal left/right
distribution inside the new cell.

We now compute the transition probabilities. Consider an interior cell
\(\mathcal D_a\), with \(2\le a\le n/2-1\). Immediately after an odd layer, the
endpoint is at \(2a\) with probability \(4/5\). The following even layer acts on
\((2a,2a+1)\), and moves the endpoint to the next cell with probability \(4/5\).
Hence $\mathbb P(Y_{m+1}=a+1\mid Y_m=a)
=
\frac45\cdot\frac45
=
\frac{16}{25}$. Similarly, the endpoint is at \(2a-1\) with probability \(1/5\). The following
even layer acts on \((2a-2,2a-1)\), and moves the endpoint to the previous cell
with probability \(1/5\). Hence $\mathbb P(Y_{m+1}=a-1\mid Y_m=a)
=
\frac15\cdot\frac15
=
\frac1{25}$. The remaining probability stays in the same cell: $\mathbb P(Y_{m+1}=a\mid Y_m=a)
=
1-\frac{16}{25}-\frac1{25}
=
\frac8{25}$. At the left boundary, the move to the left is impossible, so its probability is
added to the probability of staying in the first cell: $\mathbb P(Y_{m+1}=1\mid Y_m=1)=\frac9{25}$, $\mathbb P(Y_{m+1}=2\mid Y_m=1)=\frac{16}{25}$. At the right boundary, the move to the right is impossible, so $\mathbb P\!\left(Y_{m+1}=\frac n2-1\,\middle|\,Y_m=\frac n2\right)
=
\frac1{25}$, $
\mathbb P\!\left(Y_{m+1}=\frac n2\,\middle|\,Y_m=\frac n2\right)
=
\frac{24}{25}$.

Let \(Q\) be the transition matrix of this cell chain, defined as: \(Q_{a,b}\coloneqq \mathbb P(Y_{m+1}=a\mid Y_m=b)\). Thus
the \(b\)-th column of \(Q\) is the distribution of \(Y_{m+1}\) when the chain
starts from cell \(b\). Equivalently, \(Qe_b\) is the distribution after one
effective step, starting from cell \(b\).

For \(n=2\), there is only one cell and \(Q=(1)\). For \(n\ge4\), \(Q\) is the
\((n/2)\times(n/2)\) tridiagonal matrix
\begin{equation}
\label{eq:otoc1-Q}
Q
=
\frac1{25}
\begin{pmatrix}
9 & 1 & 0 & 0 & \cdots & 0 & 0\\
16 & 8 & 1 & 0 & \cdots & 0 & 0\\
0 & 16 & 8 & 1 & \cdots & 0 & 0\\
0 & 0 & 16 & 8 & \ddots & 0 & 0\\
\vdots & \vdots & \vdots & \ddots & \ddots & 1 & 0\\
0 & 0 & 0 & 0 & 16 & 8 & 1\\
0 & 0 & 0 & 0 & 0 & 16 & 24
\end{pmatrix}.
\end{equation}
Equivalently, \(Qe_1=\frac9{25}e_1+\frac{16}{25}e_2\),
\(Qe_a=\frac1{25}e_{a-1}+\frac8{25}e_a+\frac{16}{25}e_{a+1}\) for
\(2\le a\le n/2-1\), and
\(Qe_{n/2}=\frac1{25}e_{n/2-1}+\frac{24}{25}e_{n/2}\).

Because \(d\) is even, the first \(d\) layers contain exactly \(d/2\) odd
layers. The chain starts at \(Y_1=1\), and therefore after the last odd layer
the cell distribution is \(Q^{d/2-1}e_1\). The last odd layer is layer \(d-1\),
while the final even layer does not touch site \(n\). Thus, at time \(d\), the
Pauli string is nontrivial at site \(n\) exactly when, after the last odd layer,
the endpoint is in the last cell and is on its right site. Conditional on
\(Y_{d/2}=n/2\), this happens with probability \(4/5\). Therefore
\[
\E_{U_d}
\sum_{P:\,P_n\neq I}
\left(c_P(U_d^\dagger Z_1 U_d)\right)^2
=
\frac45\,
\mathbb P\!\left(Y_{d/2}=\frac n2\right)
=
\frac45\,e_{n/2}^{\mathsf T}Q^{d/2-1}e_1.
\]
Combining this with \eqref{eq:otoc1-Cn-endpoint-weight}, we obtain
\begin{equation}
\label{eq:otoc1-Cn-Q-formula}
1-\E_{U_d}\mathrm{OTOC}^{(1)}_\rho(U_d)
=
\frac{16}{15}\,
e_{n/2}^{\mathsf T}Q^{d/2-1}e_1.
\end{equation}

\medskip
\noindent
\textbf{Step 3: analytic evaluation of the matrix element.} 

We now evaluate the matrix element in \eqref{eq:otoc1-Cn-Q-formula}. Set $f_s\coloneqq e_{n/2}^{\mathsf T}Q^s e_1$ for all $s\ge0$, and introduce its generating function $F_n(z)\coloneqq\sum_{s\ge0} f_s z^s$. For \(|z|\) sufficiently small, \((I-zQ)^{-1}=\sum_{s\ge0}z^sQ^s\). Therefore $F_n(z)=e_{n/2}^{\mathsf T}(I-zQ)^{-1}e_1$.

If \(n=2\), then \(Q=(1)\), and hence \(F_2(z)=1/(1-z)\). The formula
\eqref{eq:otoc1-Fn-r-form} below is then checked directly. Thus, we may assume \(n\ge4\).

Let \(A\coloneqq I-zQ\). Using the explicit form of \(Q\), we have
\[
A=
\begin{pmatrix}
1-\frac{9z}{25} & -\frac z{25} & 0 & 0 & \cdots & 0 & 0\\[1mm]
-\frac{16z}{25} & 1-\frac{8z}{25} & -\frac z{25} & 0 & \cdots & 0 & 0\\[1mm]
0 & -\frac{16z}{25} & 1-\frac{8z}{25} & -\frac z{25} & \cdots & 0 & 0\\
0 & 0 & -\frac{16z}{25} & 1-\frac{8z}{25} & \ddots & 0 & 0\\
\vdots & \vdots & \vdots & \ddots & \ddots & -\frac z{25} & 0\\[1mm]
0 & 0 & 0 & 0 & -\frac{16z}{25} & 1-\frac{8z}{25} & -\frac z{25}\\[1mm]
0 & 0 & 0 & 0 & 0 & -\frac{16z}{25} & 1-\frac{24z}{25}
\end{pmatrix},
\]
with the evident truncation when \(n=4\).

We first compute the numerator of \(F_n(z)=(A^{-1})_{n/2,1}\). Recall the
cofactor formula: if \(A\) is invertible, then $(A^{-1})_{i,j}=\frac{C_{j,i}(A)}{\det A}$, where \(C_{i,j}(A)\coloneqq(-1)^{i+j}\det A^{(i,j)}\), and \(A^{(i,j)}\) is
obtained from \(A\) by deleting row \(i\) and column \(j\). Hence $F_n(z)=\frac{C_{1,n/2}(A)}{\det A}$. Since \(A\) is tridiagonal, the cofactor \(C_{1,n/2}(A)\) is especially simple.
After deleting row \(1\) and column \(n/2\), the only nonzero product in the
determinant is the product along the path \(1\to2\to\cdots\to n/2\). The
corresponding entries of \(A=I-zQ\) are \(A_{a+1,a}=-16z/25\), for
\(a=1,\ldots,n/2-1\). The sign in the cofactor cancels the sign of this product,
and we get $C_{1,n/2}(A)=\left(\frac{16z}{25}\right)^{n/2-1}$. Thus
\begin{equation}
\label{eq:otoc1-Fn-det-form}
F_n(z)
=
\frac{(16z/25)^{n/2-1}}{\det(I-zQ)}.
\end{equation}

It remains to compute the denominator \(\det(I-zQ)\). We use a
recursion method to calculate determinants of tridiagonal matrices. Let
\(D_m(z)\) be the determinant of the upper-left \(m\times m\) block before the
final right boundary is imposed. Thus \(D_0(z)=1\) and
\(D_1(z)=1-9z/25\). For \(2\le m\le n/2-1\), we are in the bulk. Expanding the
\(m\times m\) determinant along the last row gives
\begin{equation}
\label{eq:otoc1-bulk-det-recursion}
D_m(z)
=
\left(1-\frac{8z}{25}\right)D_{m-1}(z)
-
\frac{16z^2}{625}D_{m-2}(z).
\end{equation}
The second term is the product of the two off-diagonal entries, $\left(-\frac z{25}\right)
\left(-\frac{16z}{25}\right)
=
\frac{16z^2}{625}$. We now solve this recursion. Introduce $z\coloneqq \frac{25r}{4(1+r)^2}$. Then $1-\frac{8z}{25}=\frac{1+r^2}{(1+r)^2}$, $\frac{16z^2}{625}=\frac{r^2}{(1+r)^4}$. Therefore \eqref{eq:otoc1-bulk-det-recursion} becomes $D_m
=
\frac{1+r^2}{(1+r)^2}D_{m-1}
-
\frac{r^2}{(1+r)^4}D_{m-2}$. The characteristic roots are \(1/(1+r)^2\) and \(r^2/(1+r)^2\). Hence the
solution has the form $D_m\!\left(\frac{25r}{4(1+r)^2}\right)
=
\frac{\alpha+\beta r^{2m}}{(1+r)^{2m}}$. The initial conditions determine \(\alpha,\beta\). Since \(D_0=1\), and $D_1\!\left(\frac{25r}{4(1+r)^2}\right)
=
1-\frac{9r}{4(1+r)^2}
=
\frac{4-r+4r^2}{4(1+r)^2}$, we get $\alpha+\beta=1$, $\alpha+\beta r^2=\frac{4-r+4r^2}{4}$. Solving, $\alpha=\frac{4-r}{4(1-r^2)}$, $\beta=\frac{r(1-4r)}{4(1-r^2)}$. Thus, for \(0\le m\le n/2-1\),
\begin{equation}
\label{eq:otoc1-leading-det-solution}
D_m\!\left(\frac{25r}{4(1+r)^2}\right)
=
\frac{
(4-r)+(1-4r)r^{2m+1}
}{
4(1-r)(1+r)^{2m+1}
}.
\end{equation}

We still have to impose the right boundary. The last diagonal entry of
\(I-zQ\) is \(1-24z/25\), not the bulk value \(1-8z/25\). Therefore the full
determinant is
\begin{equation}
\label{eq:otoc1-right-boundary-det}
\det(I-zQ)
=
\left(1-\frac{24z}{25}\right)D_{n/2-1}(z)
-
\frac{16z^2}{625}D_{n/2-2}(z).
\end{equation}
Under the same change of variables, \(1-24z/25=(1-4r+r^2)/(1+r)^2\).
Substituting \eqref{eq:otoc1-leading-det-solution} into
\eqref{eq:otoc1-right-boundary-det}, and putting the two terms over the common
denominator \(4(1-r)(1+r)^{n+1}\), the numerator becomes
\[
\begin{aligned}
&
(1-4r+r^2)\bigl[(4-r)+(1-4r)r^{n-1}\bigr]
-
r^2\bigl[(4-r)+(1-4r)r^{n-3}\bigr] 
=
(4-r)(1-4r)(1-r^n).
\end{aligned}
\]
Therefore
\begin{equation}
\label{eq:otoc1-full-det-solution}
\det\left(I-\frac{25r}{4(1+r)^2}Q\right)
=
\frac{
(4-r)(1-4r)(1-r^n)
}{
4(1-r)(1+r)^{n+1}
}.
\end{equation}

Combining \eqref{eq:otoc1-Fn-det-form} with
\eqref{eq:otoc1-full-det-solution}, and using \(16z/25=4r/(1+r)^2\), we obtain
\begin{equation}
\label{eq:otoc1-Fn-r-form}
F_n\!\left(\frac{25r}{4(1+r)^2}\right)
=
\frac{
4^{n/2}r^{n/2-1}(1-r)(1+r)^3
}{
(4-r)(1-4r)(1-r^n)
}.
\end{equation}

We now extract the coefficient corresponding to the physical depth \(d\). We
use \([z^s]F(z)\) for the coefficient of \(z^s\) in the Taylor expansion of
\(F\) around \(z=0\). Since \(F_n(z)=\sum_{s\ge0}f_s z^s\), we have
\(f_{d/2-1}=[z^{d/2-1}]F_n(z)\). Therefore $1-\E_{U_d}\mathrm{OTOC}^{(1)}_\rho(U_d)
=
\frac{16}{15}f_{d/2-1}
=
\frac{16}{15}[z^{d/2-1}]F_n(z)$.

To compute this coefficient after the change of variables, we use Cauchy's
coefficient formula:
\[
[z^s]F(z)
=
\frac{1}{2\pi i}
\oint \frac{F(z)}{z^{s+1}}\,dz,
\]
where the contour is a small positively oriented circle around the origin.
Applying this with \(s=d/2-1\), we get $f_{d/2-1}
=
\frac{1}{2\pi i}
\oint
\frac{F_n(z)}{z^{d/2}}\,dz$. Now use \(z=25r/[4(1+r)^2]\). A direct differentiation gives $\frac{dz}{z}
=
\frac{1-r}{r(1+r)}\,dr$. Hence $\frac{dz}{z^{d/2}}
=
\left(\frac{4(1+r)^2}{25r}\right)^{d/2-1}
\frac{1-r}{r(1+r)}\,dr$. Substituting \eqref{eq:otoc1-Fn-r-form} into Cauchy's formula gives
\[
\begin{aligned}
f_{d/2-1}
&=
\frac{1}{2\pi i}
\oint
\frac{
4^{n/2}r^{n/2-1}(1-r)(1+r)^3
}{
(4-r)(1-4r)(1-r^n)
}
\left(\frac{4(1+r)^2}{25r}\right)^{d/2-1}
\frac{1-r}{r(1+r)}\,dr
\\
&=
4^{n/2}
\left(\frac4{25}\right)^{d/2-1}
\frac{1}{2\pi i}
\oint
\frac{
r^{(n-d)/2-1}
(1-r)^2(1+r)^d
}{
(4-r)(1-4r)(1-r^n)
}\,dr.
\end{aligned}
\]
By Cauchy's formula again, the last integral extracts the coefficient of
\(r^{(d-n)/2}\) from the remaining analytic function. Therefore
\[
f_{d/2-1}
=
4^{n/2}
\left(\frac4{25}\right)^{d/2-1}
[r^{(d-n)/2}]
\frac{
(1-r)^2(1+r)^d
}{
(4-r)(1-4r)(1-r^n)
}.
\]
Multiplying by \(16/15\), and moving one factor \(4\) inside the coefficient, we
obtain
\begin{equation}
\label{eq:otoc1-Cn-coefficient}
1-\E_{U_d}\mathrm{OTOC}^{(1)}_\rho(U_d)
=
\frac{16}{15}\,
4^{n/2-1}
\left(\frac4{25}\right)^{d/2-1}
[r^{(d-n)/2}]
\frac{
4(1-r)^2(1+r)^d
}{
(4-r)(1-4r)(1-r^n)
}.
\end{equation}

This formula already proves the light-cone constraint. If \(d<n\), then
\((d-n)/2<0\). But the rational function in
\eqref{eq:otoc1-Cn-coefficient} is regular at \(r=0\), hence it has only
nonnegative powers of \(r\). Therefore the coefficient of \(r^{(d-n)/2}\) is
zero, and \(1-\E_{U_d}\mathrm{OTOC}^{(1)}_\rho(U_d)=0\).

Assume now that \(d\ge n\). We must compute the coefficient of
\(r^{(d-n)/2}\) in $\frac{
4(1-r)^2(1+r)^d
}{
(4-r)(1-4r)(1-r^n)
}.$ We first expand $\frac1{1-r^n}=\sum_{j\ge0}r^{nj}$. Thus the term \(r^{nj}\) contributes only if the remaining factors provide the
power \(r^{(d-n)/2-nj}\). Equivalently, this contribution is present only when $\frac{d-n}{2}-nj\ge0$,
i.e.~$0\le j\le
\left\lfloor\frac{d-n}{2n}\right\rfloor$. For such a \(j\), set \(a\coloneqq(d-n)/2-nj\). We need $[r^a]
\left[
\frac{4(1-r)^2}{(4-r)(1-4r)}
(1+r)^d
\right]$. The rational factor has the elementary expansion $\frac{4(1-r)^2}{(4-r)(1-4r)}
=
1+\frac35
\sum_{\ell\ge1}
\left(4^\ell-4^{-\ell}\right)r^\ell$. Indeed, $\frac{4(1-r)^2}{(4-r)(1-4r)}
=
1+\frac{9r}{(4-r)(1-4r)}
=
1+\frac35
\left(
\frac1{1-4r}-\frac1{1-r/4}
\right)$, and expanding the two geometric series gives the claimed formula. Using also \([r^m](1+r)^d=\binom dm\), we obtain 
\begin{equation}
\label{eq:otoc1-coefficient-before-binomial}
\begin{aligned}
&[r^a]
\left[
\frac{4(1-r)^2}{(4-r)(1-4r)}
(1+r)^d
\right] =
\binom da
+
\frac35
\sum_{m=0}^{a-1}
\left(4^{a-m}-4^{-(a-m)}\right)
\binom dm.
\end{aligned}
\end{equation}We now multiply \eqref{eq:otoc1-coefficient-before-binomial} by the prefactor
in \eqref{eq:otoc1-Cn-coefficient}. By definition,
\(\Pr(X_d=m)=\binom dm(1/5)^m(4/5)^{d-m}\). Using
\(a=(d-n)/2-nj\), we get
\begin{align*}
\frac{16}{15}
4^{n/2-1}
\left(\frac4{25}\right)^{d/2-1}
\binom da
&=
\frac53\,4^{-nj}\,
\Pr(X_d=a)\,,\\
\frac{16}{15}
4^{n/2-1}
\left(\frac4{25}\right)^{d/2-1}
\frac35\,4^{a-m}\binom dm
&=
4^{-nj}\,
\Pr(X_d=m),\\
\frac{16}{15}
4^{n/2-1}
\left(\frac4{25}\right)^{d/2-1}
\frac35\,4^{-(a-m)}\binom dm
&=
4^{n(j+1)}\,
\Pr(X_d=d-m).
\end{align*}

Therefore the contribution of this value of \(j\) is
\(4^{-nj}
\left\{
\Pr(X_d\le a-1)
+\frac53\,\Pr(X_d=a)
\right\}
-
4^{n(j+1)}
\Pr(X_d\ge d-a+1)\). Substituting back \(a=(d-n)/2-nj\) and summing over all allowed \(j\), we obtain
\[
\begin{aligned}
1-\E_{U_d}\mathrm{OTOC}^{(1)}_\rho(U_d)
=
\sum_{j=0}^{\left\lfloor\frac{d-n}{2n}\right\rfloor}
\Bigg[
&4^{-nj}
\left\{
\Pr\!\left(
X_d\le \frac{d-n}{2}-nj-1
\right)
+\frac53\,
\Pr\!\left(
X_d= \frac{d-n}{2}-nj
\right)
\right\}
\\
&\hspace{14mm}
-
4^{n(j+1)}
\Pr\!\left(
X_d\ge d-\left(\frac{d-n}{2}-nj\right)+1
\right)
\Bigg].
\end{aligned}
\]

Finally, this expanded expression is exactly the compact formula stated in the
theorem. Indeed, for \(r\ge0\), the quantity \(\Psi_d(r)\) in
\eqref{eq:Psi-d-def} can be rewritten as $\Psi_d(r)
=
\Pr(X_d\le r-1)
+
\frac53\,\Pr(X_d=r)
-
4^{d-2r}\Pr(X_d\ge d-r+1)$. Taking \(r=(d-n)/2-nj\) and multiplying by \(4^{-nj}\), the last prefactor
becomes \(4^{-nj}4^{d-2r}=4^{n(j+1)}\). Thus the \(j\)-th term in the expanded
formula is precisely \(4^{-nj}
\Psi_d\!\left(
\frac{d-n}{2}-nj
\right)\). Together with the already proved light-cone case \(d<n\), this proves $1-\E_{U_d}\mathrm{OTOC}^{(1)}_\rho(U_d)
=
\sum_{j=0}^{\left\lfloor\frac{d-n}{2n}\right\rfloor}
4^{-nj}
\Psi_d\!\left(
\frac{d-n}{2}-nj
\right)$. This is
equivalent to \eqref{eq:endpoint-otoc1-exact-formula}.
\end{proof}

\subsection{Gaussian approximation of averaged \(\mathrm{OTOC}^{(1)}\)}
\label{subsec:endpoint-otoc1-gaussian-front}
We now show that, for large \(n\), the averaged endpoint
\(\mathrm{OTOC}^{(1)}\) is well approximated by a Gaussian tail, in agreement
with the operator-spreading picture of
Refs.~\cite{keyserlingk2018operator,nahum2018operator}. The reason is already
visible in the exact formula \eqref{eq:endpoint-otoc1-exact-formula}: up to an
\(O(n^{-1/2})\) error, the averaged OTOC is the tail probability of a binomial
random variable \(X_d\sim\operatorname{Bin}(d,1/5)\) above the threshold
\((d-n)/2\). Thus the transition occurs when this threshold is comparable with
the typical value of \(X_d\), namely \(d/5\). Equating \(d/5\) with
\((d-n)/2\) gives the transition depth \(d=5n/3\). In the language of
Refs.~\cite{keyserlingk2018operator,nahum2018operator}, this says that the
right endpoint reaches distance \(n\) in time \(5n/3\), or equivalently that
the butterfly velocity is \(v_B=3/5\). Around this depth, the remaining scale is
the standard deviation of \(X_d\), which is of order \(\sqrt d\). The averaged
endpoint \(\mathrm{OTOC}^{(1)}\) is therefore governed asymptotically by the
central-limit theorem applied to this binomial variable, and its profile is a
Gaussian tail. The following lemma makes quantitative this Gaussian-tail approximation to the
diffusive front picture of Refs.~\cite{keyserlingk2018operator,nahum2018operator}.
\begin{lemma}[Gaussian approximation of averaged OTOC]
\label{lem:endpoint-otoc1-gaussian-front-d-dependent}
Assume that \(n\ge4\) and \(d\ge2\) are even. Then, for every density matrix
\(\rho\),
\begin{equation}
\label{eq:endpoint-otoc1-gaussian-front-d-dependent}
\left|
\E_{U_d}\mathrm{OTOC}^{(1)}_\rho(U_d)
-
\frac1{\sqrt{2\pi}}
\int_{\frac{d-5n/3}{\frac43\sqrt d}}^\infty e^{-s^2/2}\,ds
\right|
\le
\frac{5}{\sqrt n}.
\end{equation}
In particular, the error in approximating the averaged endpoint
\(\mathrm{OTOC}^{(1)}\) by its Gaussian tail is of order \(n^{-1/2}\),
uniformly in the depth \(d\).
\end{lemma}

\begin{proof}
If \(d<n\), then
Lemma~\ref{thm:endpoint-otoc1-exact-formula} gives
\(\E_{U_d}\mathrm{OTOC}^{(1)}_\rho(U_d)=1\). Moreover
\((3d-5n)/(4\sqrt d)\le -\sqrt n/2\), since
\(d\mapsto(3d-5n)/(4\sqrt d)\) is increasing on \(d>0\). Hence $\left|
1-\frac1{\sqrt{2\pi}}
\int_{\frac{3d-5n}{4\sqrt d}}^\infty e^{-s^2/2}\,ds
\right|
=
\frac1{\sqrt{2\pi}}
\int_{-\infty}^{\frac{3d-5n}{4\sqrt d}}e^{-s^2/2}\,ds
\le e^{-n/8}
\le \frac5{\sqrt n}$. Thus the claim holds when \(d<n\).

Assume now that \(d\ge n\), and set \(a\coloneqq(d-n)/2\). By Theorem~\ref{thm:endpoint-otoc1-exact-formula}, $1-\E_{U_d}\mathrm{OTOC}^{(1)}_\rho(U_d)
=
\sum_{j=0}^{\lfloor a/n\rfloor}4^{-nj}\Psi_d(a-nj)$.
We compare this with
\(\Pr(X_d\le a)\). The \(j=0\) term satisfies
\(\Psi_d(a)-\Pr(X_d\le a)=\frac23\Pr(X_d=a)-\sum_{s=1}^{a}16^{-s}\Pr(X_d=a-s)\).
We use the standard local bound for sums of independent Bernoulli random
variables~\cite{BaillonCominettiVaisman2016}: if \(S\) has variance
\(\sigma^2\), then \(\sup_m\Pr(S=m)\le1/\sigma\). Since
\(\operatorname{Var}(X_d)=4d/25\), we have
\(\sup_m\Pr(X_d=m)\le5/(2\sqrt d)\). Therefore
\[
\left|\Psi_d(a)-\Pr(X_d\le a)\right|
\le
\left(\frac23+\sum_{s\ge1}16^{-s}\right)\sup_m\Pr(X_d=m)
=
\frac{11}{15}\sup_m\Pr(X_d=m)
\le
\frac{11}{6\sqrt d}
\le
\frac{11}{6\sqrt n}.
\]
The terms with $j\ge1$ are exponentially small: since
\(|\Psi_d(r)|\le1+\frac23+\sum_{s\ge1}16^{-s}<2\), we have $\left|
\sum_{j=1}^{\lfloor a/n\rfloor}4^{-nj}\Psi_d(a-nj)
\right|
\le
2\sum_{j\ge1}4^{-nj}
=
\frac{2}{4^n-1}
\le
\frac1{\sqrt n}$. Using the two bounds above, we obtain
\[
\left|
\E_{U_d}\mathrm{OTOC}^{(1)}_\rho(U_d)-\Pr(X_d>a)
\right|
=
\left|
\left[
1-\E_{U_d}\mathrm{OTOC}^{(1)}_\rho(U_d)
\right]
-
\Pr(X_d\le a)
\right|
\le
\frac{17}{6\sqrt n}.
\]

It remains to approximate \(\Pr(X_d>a)\). To do this, we use Berry--Esseen in its standard form. If
\(S_d=Y_1+\cdots+Y_d\), where the \(Y_i\)'s are i.i.d.~with mean \(\mu\),
variance \(\sigma^2\), and third moment \(\mathbb E|Y-\mu|^3\), then
\[
\sup_x
\left|
\Pr\!\left(
\frac{S_d-d\mu}{\sigma\sqrt d}\le x
\right)
-
\frac1{\sqrt{2\pi}}\int_{-\infty}^{x}e^{-s^2/2}\,ds
\right|
\le
\frac{C_{\rm BE}\,\mathbb E|Y-\mu|^3}{\sigma^3\sqrt d}.
\]
where \(C_{\rm BE}<1\)~\cite{Shevtsova2011}. In our case \(X_d=\sum_{i=1}^dY_i\), with
\(Y_i\sim\operatorname{Bernoulli}(1/5)\). Hence \(\mu=1/5\),
\(\sigma^2=4/25\), and
\(\mathbb E|Y-\mu|^3=\frac15(\frac45)^3+\frac45(\frac15)^3=\frac{68}{625}\), so
\(\mathbb E|Y-\mu|^3/\sigma^3=17/10\). Therefore, uniformly in \(x\),
\[
\left|
\Pr\!\left(
\frac{X_d-d/5}{(2/5)\sqrt d}\le x
\right)
-
\frac1{\sqrt{2\pi}}\int_{-\infty}^{x}e^{-s^2/2}\,ds
\right|
\le
\frac2{\sqrt d}.
\]
Since \(a=(d-n)/2\), the normalized threshold is
\((a-d/5)/((2/5)\sqrt d)=(3d-5n)/(4\sqrt d)\). Thus Berry--Esseen gives $\left|
\Pr(X_d\le a)
-
\frac1{\sqrt{2\pi}}
\int_{-\infty}^{\frac{3d-5n}{4\sqrt d}}e^{-s^2/2}\,ds
\right|
\le
\frac2{\sqrt d}$. Equivalently,  $\left|
\Pr(X_d>a)
-
\frac1{\sqrt{2\pi}}
\int_{\frac{3d-5n}{4\sqrt d}}^\infty e^{-s^2/2}\,ds
\right|
\le
\frac2{\sqrt d}
\le
\frac2{\sqrt n}$. Combining this with
\(\left|\E_{U_d}\mathrm{OTOC}^{(1)}_\rho(U_d)-\Pr(X_d>a)\right|
\le 17/(6\sqrt n)\), we get
\[
\left|
\E_{U_d}\mathrm{OTOC}^{(1)}_\rho(U_d)
-
\frac1{\sqrt{2\pi}}
\int_{\frac{3d-5n}{4\sqrt d}}^\infty e^{-s^2/2}\,ds
\right|
\le
\left(\frac{17}{6}+2\right)\frac1{\sqrt n}
\le
\frac5{\sqrt n}.
\]
This proves \eqref{eq:endpoint-otoc1-gaussian-front-d-dependent}.
\end{proof}

\subsection{Variance lower bound}

The goal of this section is to prove a lower bound on the circuit-to-circuit
fluctuations of \(\operatorname{OTOC}^{(1)}\) of the form
\[
\Var_{U_d}\!\left(\mathrm{OTOC}^{(1)}_\rho(U_d)\right)
\ge
\Omega(n^{-1/2})
\]
for every density matrix \(\rho\), uniformly throughout the depth window
\[
\left|d-\frac53 n\right|\le \sqrt n.
\]
The assumption $\left|d-\frac53 n\right|\le \sqrt n$ with constant \(1\) is only for simplicity. The same argument
works, with no conceptual change, under the more general assumption
\[
\left|d-\frac53 n\right|\le c\sqrt n
\]
for any fixed constant \(c>0\); in that case the conclusion becomes
\[
\Var_{U_d}\!\left(\mathrm{OTOC}^{(1)}_\rho(U_d)\right)
\ge
\Omega_c(n^{-1/2}),
\]
where the implicit constant may depend on \(c\). For simplicity, we take \(c=1\).

We now give an overview of the argument. To lower bound the variance, we look at
the fluctuation generated by one gate at a time. Fix an even-depth gate
\(z=(\ell,t)\), set this gate equal to a two-qubit unitary \(G_z\), and average
over all remaining gates. This gives the conditional mean
\[
\E_{\neq z}\mathrm{OTOC}^{(1)}_\rho(U_d(G_z)),
\]
where $U_d(G_z)$ denotes the circuit with where the gate at position $z$ is equal to $G_z$.

The first step is to compute this conditional mean exactly. The key point is
that, after averaging over all gates except the one at \(z\), the circuit before
and after \(z\) is again described by the endpoint Markov chain. Thus fixing
the gate \(G_z\) can affect the final OTOC only through three quantities, which
will be defined explicitly in Lemma~\ref{lem:otoc1-conditional-mean-even-gate}.

The first quantity is a past propagation factor, denoted
\(\mathsf P_{\ell,t}\). It depends only on the Markov chain and on the spacetime
position \(z=(\ell,t)\), and measures how much endpoint weight reaches the bond
\((\ell,\ell+1)\) at layer \(t\).

The second quantity is local and depends only on the fixed two-qubit gate
\(G_z\). We denote it by \(A(G_z)\). It represents the average squared Pauli weight
which, starting from an input supported on the left qubit of the gate, exits on
the right qubit. For a Haar-random gate, this local right-output weight has mean
\(4/5\). Thus the local signed excess created by replacing the Haar-random gate
with the fixed gate \(G_z\) is \(A(G_z)-4/5\).

The third quantity is a future propagation factor, denoted
\(\mathsf F_{\ell,t}\). Like \(\mathsf P_{\ell,t}\), it is defined in terms of
the Markov chain and depends only on the spacetime position of the gate. It
measures how much a small amount of endpoint weight shifted across the bond
\((\ell,\ell+1)\) can still affect the probability of reaching the final
endpoint \(n\) by depth \(d\).

With these definitions, the difference between the conditional mean and the
fully averaged OTOC has the form
\[
\text{past propagation}
\times
\text{local gate excess}
\times
\text{future propagation}.
\]
Lemma~\ref{lem:otoc1-conditional-mean-even-gate} makes this factorization exact:
\begin{equation}
\label{eq_cond_mean}
\E_{\neq z}\mathrm{OTOC}^{(1)}_\rho(U_d(G_z))
-
\E_{U_d}\mathrm{OTOC}^{(1)}_\rho(U_d)
=
-\frac{16}{15}\,
\mathsf P_{\ell,t}\mathsf F_{\ell,t}
\left(A(G_z)-\frac45\right).
\end{equation}
Taking the variance with respect the fixed gate $G_z$ at position $z$, one thus obtains
\begin{equation}
\label{eq_var_cicrc}
\Var_{G_z\sim\Haar(U(4))}
\!\left(
\E_{\neq z}\mathrm{OTOC}^{(1)}_\rho(U_d(G_z))
\right)
=
\frac{256}{225}\,
\sigma_A^2
\bigl(\mathsf P_{\ell,t}\mathsf F_{\ell,t}\bigr)^2,
\end{equation}
where $\sigma_A^2
\coloneqq
\Var_{G_z\sim\Haar(U(4))}(A(G_z))$ is a universal positive constant. Thus, a gate \(z=(\ell,t)\) gives a visible
contribution to the variance precisely when both propagation factors
\(\mathsf P_{\ell,t}\) and \(\mathsf F_{\ell,t}\) are not too small.

The third step is to choose the gates for which this happens. Set
\[
\Delta\coloneqq d-\frac53 n,
\]
and consider the spacetime set
\[
\mathcal S
\coloneqq
\left\{
(\ell,t):
\ell,t\ \mathrm{even},\quad
\frac n3\le \ell\le \frac{2n}{3},\quad
\left|t-\frac53\ell-\frac{\Delta}{2}\right|\le \sqrt n
\right\}.
\]
This is the region where the endpoint front has enough time to reach the gate,
and a perturbation created at the gate still has enough time to reach the final
endpoint. Lemma~\ref{lem:otoc1-PF-exact-front-window} proves that, uniformly
over \(z=(\ell,t)\in\mathcal S\),
\[
\mathsf P_{\ell,t}=\Omega(n^{-1/2}),
\qquad
\mathsf F_{\ell,t}=\Omega(n^{-1/2})\,.
\]
Together with \eqref{eq_var_cicrc}, this implies
\begin{equation}
\label{eq_otocotoc_var}
\Var_{G_z\sim\Haar(U(4))}
\!\left(
\E_{\neq z}\mathrm{OTOC}^{(1)}_\rho(U_d(G_z))
\right)
=
\Omega(n^{-2})
\end{equation}
for every gate \(z\in\mathcal S\).

The final step is to pass from the variance with respect to a single gate to the
variance with respect to the full circuit. We use the inequality
\[
\Var_{U_d}(\mathrm{OTOC}^{(1)}_\rho(U_d))
\ge
\sum_{z\in\mathcal S}
\Var_{G_z\sim\Haar(U(4))}
\!\left(
\E_{\neq z}\mathrm{OTOC}^{(1)}_\rho(U_d(G_z))
\right),
\]
which follows from the independence of the gates \(G_z\). Combining this
inequality with the single-gate variance contribution in
\eqref{eq_otocotoc_var}, each gate in \(\mathcal S\) contributes
\(\Omega(n^{-2})\), while \(|\mathcal S|=\Theta(n^{3/2})\). Hence the full
variance satisfies
\[
\Var_{U_d}(\mathrm{OTOC}^{(1)}_\rho(U_d))
\ge
\Theta(n^{3/2})\cdot\Omega(n^{-2})
=
\Omega(n^{-1/2}),
\]
which is the desired lower bound.

We now prove these steps in detail, starting with the exact conditional-mean
formula.

\begin{lemma}[Conditional mean after fixing one gate]
\label{lem:otoc1-conditional-mean-even-gate}
Let \(z=(\ell,t)\) be a physical gate acting on the bond
\((\ell,\ell+1)\) at layer \(t\), with \(\ell,t\) even,
\(2\le \ell\le n-2\), and \(2\le t\le d-2\). For \(G\in U(4)\), let
\(U_d(G)\) be the circuit in which the gate at \(z\) is fixed equal to \(G\),
while all other gates remain Haar-random. Then, for every density matrix
\(\rho\) and every \(G\in U(4)\),
\begin{equation}
\label{eq:otoc1-conditional-mean-even-gate}
\E_{\neq z}\mathrm{OTOC}^{(1)}_\rho(U_d(G))
=
\E_{U_d}\mathrm{OTOC}^{(1)}_\rho(U_d)
-
\frac{16}{15}\,
\mathsf P_{\ell,t}\mathsf F_{\ell,t}
\left(A(G)-\frac45\right),
\end{equation}
where \(\E_{\neq z}\) denotes expectation over all gates except the one at
\(z\), while \(\E_{U_d}\) denotes expectation over the full depth-\(d\)
Haar-random circuit, including the gate at \(z\). Here
\[
A(G)
\coloneqq
\frac13
\sum_{\alpha=1}^3
\sum_{\gamma=0}^3
\sum_{\delta=1}^3
\left|
\frac14
\Tr\!\left[
(\sigma_\gamma\otimes\sigma_\delta)
G^\dagger(\sigma_\alpha\otimes I)G
\right]
\right|^2 ,
\]
with \(\sigma_0=I,\sigma_1=X,\sigma_2=Y,\sigma_3=Z\). Thus \(A(G)\) is the
average, over inputs \(X\otimes I,Y\otimes I,Z\otimes I\), of the squared
Pauli weight which exits on the right qubit of the gate. Finally, with
\(a=\ell/2\),
\[
\mathsf P_{\ell,t}
\coloneqq
\frac45 e_a^{\mathsf T}Q^{\frac t2-1}e_1
-
\frac1{20}e_{a+1}^{\mathsf T}Q^{\frac t2-1}e_1,
\qquad
\mathsf F_{\ell,t}
\coloneqq
e_{n/2}^{\mathsf T}Q^{\frac d2-\frac t2-1}(e_{a+1}-e_a),
\]
where \(Q\) is the endpoint Markov matrix in \eqref{eq:otoc1-Q}, and \(e_a\)
denotes the \(a\)-th standard basis vector of \(\mathbb R^{n/2}\).
\end{lemma}
\begin{proof}
We compare the conditional mean with the fully Haar-averaged value. Define
\(W_n^G(d)\coloneqq
\sum_{P:\,P_n\neq I}\E_{\neq z}[(c_P(U_d^\dagger Z_1 U_d))^2]\). This is the conditional averaged endpoint weight at site \(n\), when the gate
at \(z=(\ell,t)\) is fixed equal to \(G\). Let \(W_n^{\rm Haar}(d)\) denote the
same endpoint weight when the gate at \(z\) is also Haar-random.

By the same state-independence and endpoint-weight reduction used to prove \eqref{eq:otoc1-Cn-endpoint-weight}, applied conditionally on the fixed gate \(G\), both OTOCs are obtained from the corresponding endpoint weights by
\[
\E_{\neq z}\mathrm{OTOC}^{(1)}_\rho(U_d(G))
=
1-\frac43 W_n^G(d),
\qquad
\E_{U_d}\mathrm{OTOC}^{(1)}_\rho(U_d)
=
1-\frac43 W_n^{\rm Haar}(d).
\]
Indeed, the last odd Haar layer touching site \(n\) is independent of the fixed
gate \(G\). Hence, conditional on being nontrivial on site \(n\), the local
Pauli direction at site \(n\) is still uniform among \(X,Y,Z\). For the fully
Haar-random circuit, the endpoint cell distribution after the last odd layer is
\(Q^{d/2-1}e_1\), and conditional on being in the last cell
\(\mathcal D_{n/2}\), the endpoint is on site \(n\) with probability \(4/5\).
Thus \(W_n^{\rm Haar}(d)=\frac45 e_{n/2}^{\mathsf T}Q^{d/2-1}e_1\).

We now isolate the effect of fixing the gate at \(z\):
\begin{equation}
\label{eq:otoc1-fixed-gate-otoc-gap}
\E_{\neq z}\mathrm{OTOC}^{(1)}_\rho(U_d(G))
=
\E_{U_d}\mathrm{OTOC}^{(1)}_\rho(U_d)
-\frac43\left(W_n^G(d)-W_n^{\rm Haar}(d)\right).
\end{equation}
The rest of the proof computes \(W_n^G(d)-W_n^{\rm Haar}(d)\).

Write \(a=\ell/2\). Since \(t\) is even, the fixed gate acts on
\((\ell,\ell+1)=(2a,2a+1)\), which separates the two cells
\(\mathcal D_a=\{2a-1,2a\}\) and
\(\mathcal D_{a+1}=\{2a+1,2a+2\}\). Immediately before layer \(t\), the endpoint
cell distribution is \(Q^{t/2-1}e_1\).

We compare two evolutions from depth \(t\) onward. In the first one, the gate
at \(z\) is fixed equal to \(G\). In the second one, the same gate is
Haar-random. All other gates are averaged in the same way in the two
evolutions. Hence the two evolutions can differ only through the local effect
produced at the gate \(z\).

Let \(\delta v_{\ell,t}(G)\in\mathbb R^{n/2}\) be the signed perturbation of the
cell distribution created by replacing the Haar-averaged gate at \(z\) with the
fixed gate \(G\):
\[
\delta v_{\ell,t}(G)
=
\left[
\text{cell distribution after the fixed gate \(G\)}
\right]
-
\left[
\text{cell distribution after a Haar-averaged gate at \(z\)}
\right].
\]
Once this perturbation is created, it evolves under the usual endpoint Markov
chain. From after the fixed even layer to the last odd layer, there are
\(d/2-t/2-1\) effective applications of \(Q\). Finally, conditional on being in
the last cell \(\mathcal D_{n/2}\), the endpoint is on site \(n\) with
probability \(4/5\). Therefore
\begin{equation}
\label{eq:otoc1-W-gap-from-cell-perturbation}
W_n^G(d)-W_n^{\rm Haar}(d)
=
\frac45\,
e_{n/2}^{\mathsf T}
Q^{\frac d2-\frac t2-1}
\delta v_{\ell,t}(G).
\end{equation}

It remains to compute \(\delta v_{\ell,t}(G)\). Since the gate
\((2a,2a+1)\) separates \(\mathcal D_a\) from \(\mathcal D_{a+1}\), changing
the probability of exiting on the right qubit transfers endpoint weight between
cell \(a\) and cell \(a+1\). More explicitly, let \(p_{\rm right}^{G}\) be the
averaged probability that the endpoint exits the gate \((2a,2a+1)\) on the
right qubit when the gate is fixed equal to \(G\), and let
\(p_{\rm right}^{\rm Haar}\) be the same probability when the gate at \(z\) is
Haar-averaged. Define $\eta_{\ell,t}(G)
\coloneqq
p_{\rm right}^{G}-p_{\rm right}^{\rm Haar}$. Thus \(\eta_{\ell,t}(G)\) is the excess right-output probability produced by
the fixed gate \(G\), compared with the Haar-averaged local rule.

If the endpoint exits on the right qubit of the bond \((2a,2a+1)\), it is
placed in the next cell \(\mathcal D_{a+1}\). If it does not exit on the right
qubit, it remains in the cell \(\mathcal D_a\). Therefore increasing the
right-output probability by \(\eta_{\ell,t}(G)\) adds mass
\(\eta_{\ell,t}(G)\) to cell \(a+1\) and removes the same mass from cell \(a\).
Hence $\delta v_{\ell,t}(G)
=
\eta_{\ell,t}(G)e_{a+1}
-
\eta_{\ell,t}(G)e_a
=
\eta_{\ell,t}(G)(e_{a+1}-e_a)$. Combining this with \eqref{eq:otoc1-W-gap-from-cell-perturbation}, we get
\begin{equation}
\label{eq:otoc1-W-gap-from-eta}
W_n^G(d)-W_n^{\rm Haar}(d)
=
\frac45\,
\eta_{\ell,t}(G)\,
e_{n/2}^{\mathsf T}
Q^{\frac d2-\frac t2-1}
(e_{a+1}-e_a).
\end{equation}

We now compute the local scalar \(\eta_{\ell,t}(G)\). There are two possible
ways in which the endpoint can enter the fixed gate.

First, the endpoint can enter from the left. Conditional on being in
\(\mathcal D_a\), it is on the right site \(2a\) with probability \(4/5\).
Therefore the left-input weight is \(\frac45 e_a^{\mathsf T}Q^{t/2-1}e_1\).
For this input, the local Pauli is uniform over
\(X\otimes I,Y\otimes I,Z\otimes I\). The right-output weight is \(A(G)\),
while its Haar average is \(4/5\). Hence the deviation from the Haar-averaged
local rule is \(A(G)-4/5\). Therefore the left-input contribution to
\(\eta_{\ell,t}(G)\) is \[\frac45 e_a^{\mathsf T}Q^{\frac t2-1}e_1
\left(A(G)-\frac45\right)\,.\]

Second, the endpoint can enter the gate from the right. This means that,
immediately before the fixed gate on \((2a,2a+1)\), the endpoint is on the site
\(2a+1\). Since \(2a+1\) is the left site of the cell
\(\mathcal D_{a+1}=\{2a+1,2a+2\}\), and since the endpoint is on the left site
of a cell with probability \(1/5\), the total weight entering the gate from the
right is \(p_{\rm R}\coloneqq
\frac15 e_{a+1}^{\mathsf T}Q^{\frac t2-1}e_1\). We now identify the
corresponding local Pauli inputs. Entering from the right means that the Pauli
on the right qubit of the bond \((2a,2a+1)\) is nontrivial. Hence the local
input belongs to the sector $\mathcal R
\coloneqq
\left\{
\sigma_\alpha\otimes\sigma_\beta:
\alpha\in\{0,1,2,3\},\ \beta\in\{1,2,3\}
\right\}$. This sector has \(12\) elements and consists exactly of the two-qubit Paulis
which are nontrivial on the right qubit. After averaging over the Haar gates
before depth \(t\), the right-input weight \(p_{\rm R}\) is uniformly averaged
over these \(12\) local Pauli inputs.

Recall that \(A(G)\) is the average right-output weight when the input is
uniformly distributed over \(X\otimes I,Y\otimes I,Z\otimes I\). For a right
input, define analogously
\[
B(G)
\coloneqq
\frac1{12}
\sum_{\alpha=0}^3
\sum_{\beta=1}^3
\sum_{\gamma=0}^3
\sum_{\delta=1}^3
\left|
\frac14
\Tr\!\left[
(\sigma_\gamma\otimes\sigma_\delta)
G^\dagger(\sigma_\alpha\otimes\sigma_\beta)G
\right]
\right|^2 .
\]
Thus \(B(G)\) is the average right-output weight when the input is uniformly
distributed over the \(12\) Paulis in \(\mathcal R\). We now express \(B(G)\) in terms of \(A(G)\). Set
\[
M_{\gamma,\delta;\alpha,\beta}(G)
\coloneqq
\left|
\frac14
\Tr\!\left[
(\sigma_\gamma\otimes\sigma_\delta)
G^\dagger(\sigma_\alpha\otimes\sigma_\beta)G
\right]
\right|^2 ,
\]
for \((\alpha,\beta)\neq(0,0)\) and \((\gamma,\delta)\neq(0,0)\). Since
conjugation by \(G\) is an orthogonal transformation of the \(15\)-dimensional
real vector space spanned by the nonidentity two-qubit Paulis, the matrix
\(M_{\gamma,\delta;\alpha,\beta}(G)\) is doubly stochastic. In particular, for
each fixed output Pauli \((\sigma_\gamma\otimes\sigma_\delta)\neq I\otimes I\),
\[
\sum_{\substack{\alpha,\beta=0\\(\alpha,\beta)\neq(0,0)}}^3
M_{\gamma,\delta;\alpha,\beta}(G)
=
1.
\]
Summing this identity over the output Paulis which are nontrivial on the right
qubit gives
\[
\sum_{\gamma=0}^3
\sum_{\delta=1}^3
\sum_{\substack{\alpha,\beta=0\\(\alpha,\beta)\neq(0,0)}}^3
\left|
\frac14
\Tr\!\left[
(\sigma_\gamma\otimes\sigma_\delta)
G^\dagger(\sigma_\alpha\otimes\sigma_\beta)G
\right]
\right|^2
=
\sum_{\gamma=0}^3\sum_{\delta=1}^3 1
=
12.
\]
The input sum splits into two disjoint parts:
\[
\{(\alpha,\beta):(\alpha,\beta)\neq(0,0)\}
=
\{(1,0),(2,0),(3,0)\}
\sqcup
\{(\alpha,\beta):\alpha=0,1,2,3,\ \beta=1,2,3\}.
\]
The first part contributes \(3A(G)\), and the second part contributes
\(12B(G)\). Therefore \(3A(G)+12B(G)=12\), and so
\(B(G)=1-\frac14A(G)\). Hence the deviation from the Haar value \(4/5\) for a
right input is \(B(G)-\frac45
=
1-\frac14A(G)-\frac45
=
-\frac14(A(G)-\frac45)\). Using the right-input weight \(p_{\rm R}\coloneqq
\frac15 e_{a+1}^{\mathsf T}Q^{\frac t2-1}e_1\), the
right-input contribution to \(\eta_{\ell,t}(G)\) is therefore
\[
-\frac1{20}e_{a+1}^{\mathsf T}Q^{\frac t2-1}e_1
\left(A(G)-\frac45\right).
\]

Combining the left- and right-input contributions, we find
\[
\eta_{\ell,t}(G)
=
\left[
\frac45 e_a^{\mathsf T}Q^{\frac t2-1}e_1
-
\frac1{20}e_{a+1}^{\mathsf T}Q^{\frac t2-1}e_1
\right]
\left(A(G)-\frac45\right).
\]
By the definition of \(\mathsf P_{\ell,t}\), this is $\eta_{\ell,t}(G)
=
\mathsf P_{\ell,t}
\left(A(G)-\frac45\right)$. Substituting this into
\eqref{eq:otoc1-W-gap-from-eta}, and using the definition of
\(\mathsf F_{\ell,t}\), gives $W_n^G(d)-W_n^{\rm Haar}(d)
=
\frac45\,
\mathsf P_{\ell,t}\mathsf F_{\ell,t}
\left(A(G)-\frac45\right)$. Finally, substituting this into \eqref{eq:otoc1-fixed-gate-otoc-gap}, it proves \eqref{eq:otoc1-conditional-mean-even-gate}.
\end{proof}
The conditional-mean formula immediately gives a formula for the variance of the conditional mean. After
all gates except the one at \(z\) have been averaged, the remaining dependence on \(G\) is an affine function of the single scalar \(A(G)\).
\begin{lemma}[Variance of the conditional mean]
\label{lem:otoc1-conditional-mean-variance}
Let \(z=(\ell,t)\) be a physical gate acting on the bond
\((\ell,\ell+1)\) at layer \(t\), with \(\ell,t\) even,
\(2\le \ell\le n-2\), and \(2\le t\le d-2\). For \(G\in U(4)\), let
\(U_d(G)\) be the circuit in which the gate at \(z\) is fixed equal to \(G\),
while all other gates remain Haar-random. Then, for every density matrix
\(\rho\),
\begin{equation}
\label{eq:otoc1-conditional-mean-variance-exact}
\operatorname{Var}_{G\sim\Haar(U(4))}
\!\left(
\E_{\neq z}\mathrm{OTOC}^{(1)}_\rho(U_d(G))
\right)
=
\frac{256}{225}\,
\bigl(\mathsf P_{\ell,t}\mathsf F_{\ell,t}\bigr)^2
\operatorname{Var}_{G\sim\Haar(U(4))}(A(G)).
\end{equation}
Moreover,
\(\operatorname{Var}_{G\sim\Haar(U(4))}(A(G))\) is a strictly positive
numerical constant.
\end{lemma}

\begin{proof}
By the conditional-mean formula
\eqref{eq:otoc1-conditional-mean-even-gate}, for every fixed \(G\),
\[
\E_{\neq z}\mathrm{OTOC}^{(1)}_\rho(U_d(G))
=
\E_{U_d}\mathrm{OTOC}^{(1)}_\rho(U_d)
-
\frac{16}{15}\,
\mathsf P_{\ell,t}\mathsf F_{\ell,t}
\left(A(G)-\frac45\right).
\]
The first term is independent of \(G\), and
\(\mathsf P_{\ell,t}\mathsf F_{\ell,t}\) depends only on the spacetime position
\(z=(\ell,t)\). Therefore, taking the variance over
\(G\sim\Haar(U(4))\),
\[
\operatorname{Var}_{G\sim\Haar(U(4))}
\!\left(
\E_{\neq z}\mathrm{OTOC}^{(1)}_\rho(U_d(G))
\right)
=
\left(\frac{16}{15}\right)^2
\bigl(\mathsf P_{\ell,t}\mathsf F_{\ell,t}\bigr)^2
\operatorname{Var}_{G\sim\Haar(U(4))}
\!\left(A(G)-\frac45\right).
\]
Subtracting a constant does not change the variance, so
\[
\operatorname{Var}_{G\sim\Haar(U(4))}
\!\left(A(G)-\frac45\right)
=
\operatorname{Var}_{G\sim\Haar(U(4))}(A(G)).
\]
This proves \eqref{eq:otoc1-conditional-mean-variance-exact}.

It remains only to note that this local variance is strictly positive. The
function \(A\) is continuous on \(U(4)\) and is not constant: indeed \(A(I)=0\),
whereas \(A(\mathrm{SWAP})=1\). Since Haar measure has full support on
\(U(4)\), a nonconstant continuous function has strictly positive Haar
variance.
\end{proof}

Before continuing, we record a simple lower bound on binomial probabilities that
will be used below.
\begin{lemma}[Crude local binomial lower bound]
\label{lem:crude-local-binomial-lower}
Let \(M\ge1\), \(0\le r\le M\), and
\(b_M(r)\coloneqq \binom Mr(1/5)^r(4/5)^{M-r}\). Then
\begin{equation}
\label{eq:crude-local-binomial-lower}
b_M(r)
\ge
\frac1{10\sqrt{M+1}}
\exp\!\left[
-10\,\frac{\left(r-\frac M5\right)^2}{M+1}
\right].
\end{equation}
\end{lemma}

\begin{proof}
We use Robbins' explicit version of Stirling's formula~\cite{Robbins1955}:
\(m!=\sqrt{2\pi}\,m^{m+1/2}e^{-m}e^{\rho_m}\), with
\((12m+1)^{-1}<\rho_m<(12m)^{-1}\). The endpoint cases \(r=0\) and \(r=M\)
are immediate from the exponential term in
\eqref{eq:crude-local-binomial-lower}, after weakening constants. We assume
therefore that \(1\le r\le M-1\), and set \(\theta=r/M\). Applying Robbins'
formula to \(M!\), \(r!\), and \((M-r)!\), we get
\[
b_M(r)
\ge
\frac{e^{-1/6}}{\sqrt{2\pi M\theta(1-\theta)}}\,
\exp[-M f(\theta)],
\qquad
f(\theta)\coloneqq
\theta\log(5\theta)
+
(1-\theta)\log\!\left(\frac{5(1-\theta)}4\right).
\]
Since \(\theta(1-\theta)\le1/4\), the prefactor is at least \(1/(3\sqrt M)\).
Moreover, using \(\log x\le x-1\), we have
\(f(\theta)\le\theta(5\theta-1)+(1-\theta)(5(1-\theta)/4-1)
=\frac{25}{4}(\theta-\frac15)^2\). Thus
\[
M f(\theta)
\le
\frac{25}{4}\frac{\left(r-\frac M5\right)^2}{M}
\le
10\,\frac{\left(r-\frac M5\right)^2}{M+1},
\]
where the last inequality holds for \(M\ge2\). The case \(M=1\) is checked
directly. Hence
\[
b_M(r)
\ge
\frac1{3\sqrt M}
\exp\!\left[
-10\,\frac{\left(r-\frac M5\right)^2}{M+1}
\right]
\ge
\frac1{10\sqrt{M+1}}
\exp\!\left[
-10\,\frac{\left(r-\frac M5\right)^2}{M+1}
\right],
\]
which proves the claim.
\end{proof}
In the next lemma, we identify a macroscopic region of gates of the circuit for which the variance of the conditional mean is
provably at least of order \(1/n^2\). 
\begin{lemma}[Lower bound on the variance of the conditional mean]
\label{lem:otoc1-PF-exact-front-window}
Assume that \(n\) is even and sufficiently large. Let $d$ even such that \(|\Delta|\le\sqrt n\), where \(\Delta\coloneqq d-\frac53 n\). Let
\[
\mathcal S
\coloneqq
\left\{
(\ell,t):
\ell,t\ \mathrm{even},\quad
\frac n3\le \ell\le \frac{2n}{3},\quad
\left|t-\frac53\ell-\frac{\Delta}{2}\right|\le \sqrt n
\right\}.
\]
Then,
for every \(z=(\ell,t)\in\mathcal S\), the factors \(\mathsf P_{\ell,t}\) and
\(\mathsf F_{\ell,t}\) defined in
Lemma~\ref{lem:otoc1-conditional-mean-even-gate} satisfy
\begin{align}
\label{eq:otoc1-P-binomial-front-window}
\mathsf P_{\ell,t}
&=
\psi_{t-1}\!\left(\frac{t-\ell}{2}\right),\\
\label{eq:otoc1-F-binomial-front-window}
\mathsf F_{\ell,t}
&=
\frac54
\psi_{d-t-1}\!\left(\frac{d-t-(n-\ell)}{2}\right),
\end{align}
where \(\psi_M(r)\coloneqq 4^{M-r}5^{-M}\left[\binom Mr-\binom M{r-1}\right]\).
Moreover, for every density matrix \(\rho\), uniformly for
\(z=(\ell,t)\in\mathcal S\), it holds that
\[
\operatorname{Var}_{G\sim\Haar(U(4))}
\!\left(
\E_{\neq z}\mathrm{OTOC}^{(1)}_\rho(U_d(G))
\right)
=
\Omega(n^{-2}) .
\]
\end{lemma}
\begin{proof}
Recall from Lemma~\ref{lem:otoc1-conditional-mean-even-gate} that, with
\(a=\ell/2\),
\[
\mathsf P_{\ell,t}
=
\frac45 e_a^{\mathsf T}Q^{\frac t2-1}e_1
-
\frac1{20}e_{a+1}^{\mathsf T}Q^{\frac t2-1}e_1,
\qquad
\mathsf F_{\ell,t}
=
e_{n/2}^{\mathsf T}Q^{\frac d2-\frac t2-1}(e_{a+1}-e_a).
\]
The proof has two parts. First, we compute these two quantities exactly by the
same method used in the proof of
Lemma~\ref{thm:endpoint-otoc1-exact-formula}. This gives
\eqref{eq:otoc1-P-binomial-front-window} and \eqref{eq:otoc1-F-binomial-front-window}. Second, we use the definition of
\(\mathcal S\) to show that \(\mathsf P_{\ell,t}=\Omega(n^{-1/2})\) and
\(\mathsf F_{\ell,t}=\Omega(n^{-1/2})\).

As in the proof of
Lemma~\ref{thm:endpoint-otoc1-exact-formula}, let  \(Q\) be the endpoint Markov
matrix in \eqref{eq:otoc1-Q}. By the expansion $(I-\zeta Q)^{-1}=\sum_{s\ge0}\zeta^sQ^s$, we can write the two sensitivity factors $\mathsf P_{\ell,t}$ and $\mathsf F_{\ell,t}$ in terms of the following coefficients:
\begin{align}
\label{eq:otoc1-P-as-resolvent-coefficient}
\mathsf P_{\ell,t}
&=
[\zeta^{\frac t2-1}]
\left[
\frac45 e_a^{\mathsf T}(I-\zeta Q)^{-1}e_1
-
\frac1{20}e_{a+1}^{\mathsf T}(I-\zeta Q)^{-1}e_1
\right],\\
\label{eq:otoc1-F-as-resolvent-coefficient}
\mathsf F_{\ell,t}
&=
[\zeta^{\frac d2-\frac t2-1}]
\left[
e_{n/2}^{\mathsf T}(I-\zeta Q)^{-1}(e_{a+1}-e_a)
\right].
\end{align}
Thus the computation of \(\mathsf P_{\ell,t}\) and \(\mathsf F_{\ell,t}\) is reduced to two steps: first compute the matrix entries appearing in
\eqref{eq:otoc1-P-as-resolvent-coefficient} and
\eqref{eq:otoc1-F-as-resolvent-coefficient}, and then extract the indicated
coefficients.

As in the proof of Lemma~\ref{thm:endpoint-otoc1-exact-formula}, set
\(\zeta=25r/[4(1+r)^2]\). We shall use the determinant identities derived
there. First, by \eqref{eq:otoc1-full-det-solution},
\begin{equation}
\label{eq:full-det-recalled-sensitivity-front-window}
\det(I-\zeta Q)
=
\frac{
(4-r)(1-4r)(1-r^n)
}{
4(1-r)(1+r)^{n+1}
}.
\end{equation}
Second, by \eqref{eq:otoc1-leading-det-solution}, if \(D_m(\zeta)\) denotes the
upper-left \(m\times m\) leading minor of \(I-\zeta Q\), then
\begin{equation}
\label{eq:leading-minor-recalled-sensitivity-front-window}
D_m(\zeta)
=
\frac{
(4-r)+(1-4r)r^{2m+1}
}{
4(1-r)(1+r)^{2m+1}
},
\qquad
0\le m\le n/2-1.
\end{equation}

We now compute the two resolvent entries in
\eqref{eq:otoc1-P-as-resolvent-coefficient} and
\eqref{eq:otoc1-F-as-resolvent-coefficient}. We start with
\(e_m^{\mathsf T}(I-\zeta Q)^{-1}e_1\), for \(1\le m\le n/2\). Using \eqref{eq:otoc1-Q}, the matrix
\(I-\zeta Q\) is tridiagonal, with lower and upper
off-diagonal entries $\left(I-\zeta Q\right)_{j+1,j}=-\frac{16\zeta}{25}$ and $
\left(I-\zeta Q\right)_{j,j+1}=-\frac{\zeta}{25}$. By the cofactor formula,
\[
e_m^{\mathsf T}\left(I-\zeta Q\right)^{-1}e_1
=
(\left(I-\zeta Q\right)^{-1})_{m,1}
=
\frac{
(-1)^{1+m}\det \left(I-\zeta Q\right)^{(1,m)}
}{
\det\left(I-\zeta Q\right)
},
\]
where \(\left(I-\zeta Q\right)^{(1,m)}\) is obtained by deleting row \(1\) and column
\(m\) from $\left(I-\zeta Q\right)$. Since \(\left(I-\zeta Q\right)\) is tridiagonal, this minor has only one nonzero path
connecting the first column to the deleted \(m\)-th column. This path uses the
lower off-diagonal entries
\(\left(I-\zeta Q\right)_{2,1},\left(I-\zeta Q\right)_{3,2},\ldots,
\left(I-\zeta Q\right)_{m,m-1}\), each equal to \(-16\zeta/25\). Hence the path
product is \((-16\zeta/25)^{m-1}\). The cofactor sign cancels this sign. The
remaining independent block is the lower-right block on the indices
\(m+1,m+2,\ldots,n/2\). Denote its determinant by \(E_m(\zeta)\), with the
convention \(E_{n/2}(\zeta)=1\). Therefore
\begin{equation}
\label{eq:cofactor-left-entry-sensitivity-front-window}
e_m^{\mathsf T}\left(I-\zeta Q\right)^{-1}e_1
=
\frac{
\left(\frac{16\zeta}{25}\right)^{m-1}
E_m(\zeta)
}{
\det\left(I-\zeta Q\right)
}.
\end{equation}

We now compute \(E_m(\zeta)\). This determinant satisfies the same bulk
tridiagonal recursion as the leading minors \(D_m(\zeta)\) in the exact
endpoint computation; see \eqref{eq:otoc1-bulk-det-recursion}. The only
difference is that the recursion is now solved from the right boundary. Thus
\[
E_m
=
\left(1-\frac{8\zeta}{25}\right)E_{m+1}
-
\frac{16\zeta^2}{625}E_{m+2}.
\]
The right-boundary conditions are \(E_{n/2}(\zeta)=1\) and
\(E_{{n/2}-1}(\zeta)=1-24\zeta/25\). With the same change of variables
\(\zeta=25r/[4(1+r)^2]\), the solution is
\begin{equation}
\label{eq:trailing-minor-sensitivity-front-window}
E_m(\zeta)
=
\frac{
(1-4r)+(4-r)r^{n-2m+1}
}{
(1-r)(1+r)^{n-2m+1}
}\,,
\end{equation}
as it can be verified by substitution.

Since \(16\zeta/25=4r/(1+r)^2\), substituting
\eqref{eq:trailing-minor-sensitivity-front-window} and
\eqref{eq:full-det-recalled-sensitivity-front-window} into
\eqref{eq:cofactor-left-entry-sensitivity-front-window} gives
\begin{align}
\label{eq:left-resolvent-entry-sensitivity-front-window}
e_m^{\mathsf T}\left(I-\zeta Q\right)^{-1}e_1
&=
\frac{
\left(\frac{4r}{(1+r)^2}\right)^{m-1}
\frac{(1-4r)+(4-r)r^{n-2m+1}}
{(1-r)(1+r)^{n-2m+1}}
}{
\frac{(4-r)(1-4r)(1-r^n)}
{4(1-r)(1+r)^{n+1}}
} =
\frac{
4^m r^{m-1}(1+r)^2
\left[
(1-4r)+(4-r)r^{n-2m+1}
\right]
}{
(4-r)(1-4r)(1-r^n)
}.
\end{align}

We also need \(e_{n/2}^{\mathsf T}\left(I-\zeta Q\right)^{-1}e_m\), where
\(1\le m\le {n/2}\). By the same cofactor formula,
\begin{equation}
\label{eq:cofactor-right-entry-sensitivity-front-window}
e_{n/2}^{\mathsf T}\left(I-\zeta Q\right)^{-1}e_m
=
\frac{
\left(\frac{16\zeta}{25}\right)^{{n/2}-m}
D_{m-1}(\zeta)
}{
\det\left(I-\zeta Q\right)
}.
\end{equation}
Substituting \eqref{eq:leading-minor-recalled-sensitivity-front-window} and
\eqref{eq:full-det-recalled-sensitivity-front-window}, and using again
\(16\zeta/25=4r/(1+r)^2\), we obtain
\begin{equation}
\label{eq:right-resolvent-entry-sensitivity-front-window}
e_{n/2}^{\mathsf T}\left(I-\zeta Q\right)^{-1}e_m
=
\frac{
4^{{n/2}-m}r^{{n/2}-m}(1+r)^2
\left[
(4-r)+(1-4r)r^{2m-1}
\right]
}{
(4-r)(1-4r)(1-r^n)
}.
\end{equation}
We now specialize the two identities in \eqref{eq:left-resolvent-entry-sensitivity-front-window} and \eqref{eq:right-resolvent-entry-sensitivity-front-window} to the sensitivity factors $\mathsf P_{\ell,t}$ and $\mathsf F_{\ell,t}$ in \eqref{eq:otoc1-P-as-resolvent-coefficient} and \eqref{eq:otoc1-F-as-resolvent-coefficient}. From
\eqref{eq:left-resolvent-entry-sensitivity-front-window} and \eqref{eq:right-resolvent-entry-sensitivity-front-window}, using \(2a=\ell\), we get
\begin{align}
\label{eq:resolvent-P-sensitivity-front-window}
\frac45 e_a^{\mathsf T}(I-\zeta Q)^{-1}e_1
-
\frac1{20}e_{a+1}^{\mathsf T}(I-\zeta Q)^{-1}e_1
&=
\frac15
\frac{
4^a r^{a-1}(1+r)^2(1-r^{n-\ell})
}{
1-r^n
},\\
\label{eq:resolvent-F-sensitivity-front-window}
e_{n/2}^{\mathsf T}(I-\zeta Q)^{-1}(e_{a+1}-e_a)&=
\frac{
4^{\frac{n-\ell}{2}-1}
r^{\frac{n-\ell}{2}-1}
(1+r)^2(1-r^\ell)
}{
1-r^n
}.
\end{align}

We can now compute $\mathsf P_{\ell,t}$ and $\mathsf F_{\ell,t}$ by extracting the coefficients as in \eqref{eq:otoc1-P-as-resolvent-coefficient} and \eqref{eq:otoc1-F-as-resolvent-coefficient}. Since \(\zeta=25r/[4(1+r)^2]\), Cauchy's coefficient formula,
\([z^s]H(z)=(2\pi i)^{-1}\oint H(z)z^{-s}\,dz/z\), gives the following
simple rule: if \(H(\zeta(r))=C r^q(1+r)^2K(r)\), with \(K\) regular at
\(r=0\), then
\begin{equation}
\label{eq:coefficient-extraction-rule-sensitivity-front-window}
[\zeta^s]H(\zeta)
=
C\left(\frac4{25}\right)^s
\left[
r^{s-q}
\right]
(1-r)(1+r)^{2s+1}K(r).
\end{equation}
Indeed, \(\zeta^{-s}=(4/25)^s r^{-s}(1+r)^{2s}\), while \(d\zeta/\zeta=(1-r)[r(1+r)]^{-1}dr\). Thus the power \(r^q\) in \(H\) shifts
the coefficient to be extracted from \(r^s\) to \(r^{s-q}\), and the remaining
factors are exactly those in
\eqref{eq:coefficient-extraction-rule-sensitivity-front-window}.

We first calculate the factor $\mathsf P_{\ell,t}$. By
\eqref{eq:otoc1-P-as-resolvent-coefficient} and
\eqref{eq:resolvent-P-sensitivity-front-window}, the coefficient-extraction
rule gives
\begin{equation}
\label{eq:P-coeff-r-form-front-window}
\mathsf P_{\ell,t}
=
\frac15\,4^{\ell/2}
\left(\frac4{25}\right)^{\frac t2-1}
\left[
r^{\frac{t-\ell}{2}}
\right]
(1-r)(1+r)^{t-1}
\frac{1-r^{n-\ell}}{1-r^n}.
\end{equation}
In the region \(\mathcal S\), the fraction \((1-r^{n-\ell})/(1-r^n)\) may be
replaced by its leading term \(1\). Indeed, $\frac{1-r^{n-\ell}}{1-r^n}
=
\sum_{j\ge0}r^{jn}
-
\sum_{j\ge0}r^{n-\ell+jn}$. The target degree in \eqref{eq:P-coeff-r-form-front-window} is
\((t-\ell)/2\). Since \(d<2n\) and \(t\le d\), this target degree is smaller
than \(n\). Hence all terms \(r^{jn}\) with \(j\ge1\) have degree too large.
The terms \(r^{n-\ell+jn}\) also do not contribute: for \(j\ge1\) their
degree is too large, while for \(j=0\) the required coefficient has degree
\((t-\ell)/2-(n-\ell)=(t+\ell-2n)/2\le -n/9+O(\sqrt n)<0\). Therefore only the
leading term \(1\) contributes, and
\[
\mathsf P_{\ell,t}
=
\frac15\,4^{\ell/2}
\left(\frac4{25}\right)^{\frac t2-1}
\left[
r^{\frac{t-\ell}{2}}
\right]
(1-r)(1+r)^{t-1}.
\]
By the definition of \(\psi_M\), for every integer \(u\),
\([r^u](1-r)(1+r)^{t-1}
=\binom{t-1}{u}-\binom{t-1}{u-1}
=5^{t-1}4^{-(t-1-u)}\psi_{t-1}(u)\). Applying this with
\(u=(t-\ell)/2\), the scalar prefactor cancels exactly, and we get $\mathsf P_{\ell,t}
=
\psi_{t-1}\!\left(\frac{t-\ell}{2}\right)$, which proves \eqref{eq:otoc1-P-binomial-front-window}.

The factor $\mathsf F_{\ell,t}$ is treated in the same way. By
\eqref{eq:otoc1-F-as-resolvent-coefficient} and
\eqref{eq:resolvent-F-sensitivity-front-window},
\begin{equation}
\label{eq:F-coeff-r-form-front-window}
\mathsf F_{\ell,t}
=
4^{\frac{n-\ell}{2}-1}
\left(\frac4{25}\right)^{\frac d2-\frac t2-1}
\left[
r^{\frac{d-t-(n-\ell)}{2}}
\right]
(1-r)(1+r)^{d-t-1}
\frac{1-r^\ell}{1-r^n}.
\end{equation}
Again, in the region \(\mathcal S\), the fraction \((1-r^\ell)/(1-r^n)\) may
be replaced by its leading term \(1\). Indeed, $\frac{1-r^\ell}{1-r^n}
=
\sum_{j\ge0}r^{jn}
-
\sum_{j\ge0}r^{\ell+jn}$.
The target degree in \eqref{eq:F-coeff-r-form-front-window} is
\((d-t-(n-\ell))/2<n\), so all terms \(r^{jn}\) with \(j\ge1\) have degree too
large. The terms \(r^{\ell+jn}\) also do not contribute: for \(j\ge1\)
their degree is too large, while for \(j=0\) the required coefficient has
degree
\((d-t-(n-\ell))/2-\ell=(d-t-(n+\ell))/2\le -n/9+O(\sqrt n)<0\). Therefore only the leading term \(1\) contributes, and
\[
\mathsf F_{\ell,t}
=
4^{\frac{n-\ell}{2}-1}
\left(\frac4{25}\right)^{\frac d2-\frac t2-1}
\left[
r^{\frac{d-t-(n-\ell)}{2}}
\right]
(1-r)(1+r)^{d-t-1}.
\]
By the definition of \(\psi_M\), for every integer \(u\),
\([r^u](1-r)(1+r)^{d-t-1}
=\binom{d-t-1}{u}-\binom{d-t-1}{u-1}
=5^{d-t-1}4^{-(d-t-1-u)}\psi_{d-t-1}(u)\). Applying this with
\(u=(d-t-(n-\ell))/2\), the scalar prefactor becomes \(5/4\), and hence $\mathsf F_{\ell,t}
=
\frac54
\psi_{d-t-1}\!\left(\frac{d-t-(n-\ell)}{2}\right)$, which proves \eqref{eq:otoc1-F-binomial-front-window} .

We now prove the lower bounds on $\mathsf P_{\ell,t}$ and $\mathsf F_{\ell,t}$. Write
\(t=\frac53\ell+\frac{\Delta}{2}+\varepsilon\), with
\(|\varepsilon|\le\sqrt n\). Then
\(d-t=\frac53(n-\ell)+\frac{\Delta}{2}-\varepsilon\).

Set \(b_M(r)\coloneqq\binom Mr(1/5)^r(4/5)^{M-r}\). Since
\(\binom M{r-1}/\binom Mr=r/(M-r+1)\), for \(0\le r\le M\) we have $\psi_M(r)=b_M(r)\frac{M+1-2r}{M-r+1}$. By \eqref{eq:otoc1-P-binomial-front-window}, $\mathsf P_{\ell,t}
=
b_{t-1}\!\left(\frac{t-\ell}{2}\right)\frac{2\ell}{t+\ell}$. Also \((t-\ell)/2-(t-1)/5=(3t-5\ell+2)/10\). Hence
Lemma~\ref{lem:crude-local-binomial-lower} gives
\[
\mathsf P_{\ell,t}
\ge
\frac{\ell}{5(t+\ell)\sqrt t}
\exp\!\left[-\frac{(3t-5\ell+2)^2}{10t}\right].
\]
Now \(3t-5\ell+2=\frac32\Delta+3\varepsilon+2=O(\sqrt n)\), while
\(t=\Theta(n)\) and \(\ell/(t+\ell)=\Theta(1)\). Therefore
\(\mathsf P_{\ell,t}=\Omega(n^{-1/2})\), uniformly on \((\ell,t)\in\mathcal S\).

Similarly, by \eqref{eq:otoc1-F-binomial-front-window}, $\mathsf F_{\ell,t}
=
\frac54\,
b_{d-t-1}\!\left(\frac{d-t-(n-\ell)}{2}\right)
\frac{2(n-\ell)}{d-t+n-\ell}$. Also $\frac{d-t-(n-\ell)}{2}-\frac{d-t-1}{5}
=
\frac{3d-3t-5n+5\ell+2}{10}$. Thus Lemma~\ref{lem:crude-local-binomial-lower} gives
\[
\mathsf F_{\ell,t}
\ge
\frac{n-\ell}{4(d-t+n-\ell)\sqrt{d-t}}
\exp\!\left[-\frac{(3d-3t-5n+5\ell+2)^2}{10(d-t)}\right].
\]
Since \(3d-3t-5n+5\ell+2=\frac32\Delta-3\varepsilon+2=O(\sqrt n)\), while
\(d-t=\Theta(n)\) and \((n-\ell)/(d-t+n-\ell)=\Theta(1)\), we get
\(\mathsf F_{\ell,t}=\Omega(n^{-1/2})\), uniformly on \((\ell,t)\in\mathcal S\).

Combining the two bounds, we obtain
\[
\mathsf P_{\ell,t}\mathsf F_{\ell,t}
=
\Omega(n^{-1}),
\]
uniformly for \(z=(\ell,t)\in\mathcal S\).

We now use Lemma~\ref{lem:otoc1-conditional-mean-variance}. It gives, for every
density matrix \(\rho\),
\[
\operatorname{Var}_{G\sim\Haar(U(4))}
\!\left(
\E_{\neq z}\mathrm{OTOC}^{(1)}_\rho(U_d(G))
\right)
=
\frac{256}{225}\,
\bigl(\mathsf P_{\ell,t}\mathsf F_{\ell,t}\bigr)^2
\operatorname{Var}_{G\sim\Haar(U(4))}(A(G)).
\]
The last factor is a strictly positive numerical constant, independent of
\(n,d,\ell,t\), and \(\rho\). Therefore
\[
\operatorname{Var}_{G\sim\Haar(U(4))}
\!\left(
\E_{\neq z}\mathrm{OTOC}^{(1)}_\rho(U_d(G))
\right)
=
\Omega(n^{-2}),
\]
uniformly on \(\mathcal S\). This proves the variance lower bound and completes
the proof.
\end{proof}

We are now ready to pass from the one-gate estimates to the full
circuit-to-circuit variance.

\begin{theorem}[Variance lower bound in the front window]
\label{thm:endpoint-otoc1-variance-front-window}
Let \(n,d\) be even, with \(n\) sufficiently large, and assume
\[
\left|d-\frac53 n\right|\le \sqrt n.
\]
Then, for every density matrix \(\rho\),
\[
\Var_{U_d}\!\left(\mathrm{OTOC}^{(1)}_\rho(U_d)\right)
=
\Omega(n^{-1/2}).
\]
The implicit constant is independent of \(n,d,\rho\).
\end{theorem}

\begin{proof}
Set
\[
\Delta\coloneqq d-\frac53 n,
\]
and let \(\mathcal S\) be the set of gates from
Lemma~\ref{lem:otoc1-PF-exact-front-window}, namely
\[
\mathcal S
\coloneqq
\left\{
(\ell,t):
\ell,t\ \mathrm{even},\quad
\frac n3\le \ell\le \frac{2n}{3},\quad
\left|t-\frac53\ell-\frac{\Delta}{2}\right|\le \sqrt n
\right\}.
\]
For \(n\) sufficiently large, every \(z=(\ell,t)\in\mathcal S\) is a physical
even-depth gate satisfying the assumptions of
Lemma~\ref{lem:otoc1-conditional-mean-even-gate}. There are \(\Theta(n)\)
allowed even values of \(\ell\), and, for each such \(\ell\), there are
\(\Theta(\sqrt n)\) allowed even values of \(t\). Hence
\begin{equation}
\label{eq:otoc1-size-S-front-window}
|\mathcal S|=\Theta(n^{3/2}).
\end{equation}

Let $X\coloneqq \mathrm{OTOC}^{(1)}_\rho(U_d)$. For each \(z\in\mathcal S\), denote by \(G_z\) the Haar gate at \(z\), and set $Y_z \coloneqq \E[X-\E X\mid G_z]$. Then \(Y_z\) depends only on the single gate \(G_z\), and \(\E Y_z=0\). If
\(z\neq z'\), then \(G_z\) and \(G_{z'}\) are independent, so
\begin{equation}
\label{eq:otoc1-one-gate-projections-orthogonal}
\E[\overline{Y_z}Y_{z'}]
=
\E[\overline{Y_z}]\,\E[Y_{z'}]
=
0.
\end{equation}
Moreover, by the defining property of conditional expectation,
\begin{equation}
\label{eq:otoc1-one-gate-projection-inner-product}
\E\!\left[\overline{X-\E X}\,Y_z\right]
=
\E\!\left[
\E[\overline{X-\E X}\mid G_z]\,Y_z
\right]
=
\E|Y_z|^2.
\end{equation}

Using \eqref{eq:otoc1-one-gate-projections-orthogonal} and
\eqref{eq:otoc1-one-gate-projection-inner-product}, we obtain
\[
\begin{aligned}
0
&\le
\E\left|
X-\E X-\sum_{z\in\mathcal S}Y_z
\right|^2 \\
&=
\Var(X)
-
2\operatorname{Re}
\sum_{z\in\mathcal S}
\E\!\left[\overline{X-\E X}\,Y_z\right]
+
\sum_{z,z'\in\mathcal S}
\E[\overline{Y_z}Y_{z'}] \\
&=
\Var(X)
-
\sum_{z\in\mathcal S}\E|Y_z|^2.
\end{aligned}
\]
Therefore
\begin{equation}
\label{eq:otoc1-total-variance-from-one-gate-conditional-means}
\Var_{U_d}(X)
\ge
\sum_{z\in\mathcal S}
\E|Y_z|^2.
\end{equation}

For a fixed \(z\), we have
\[
\E[X\mid G_z]
=
\E_{\neq z}\mathrm{OTOC}^{(1)}_\rho(U_d(G_z)).
\]
and thus
\[
\E|Y_z|^2
=
\Var_{G_z\sim\Haar(U(4))}
\!\left(
\E_{\neq z}\mathrm{OTOC}^{(1)}_\rho(U_d(G_z))
\right)\,.
\]
Consequently, \eqref{eq:otoc1-total-variance-from-one-gate-conditional-means} implies that
\[
    \Var_{U_d}(X)
\ge
\sum_{z\in\mathcal S}\Var_{G_z\sim\Haar(U(4))}
\!\left(
\E_{\neq z}\mathrm{OTOC}^{(1)}_\rho(U_d(G_z))
\right)\,. 
\]
By Lemma~\ref{lem:otoc1-PF-exact-front-window}, uniformly for
\(z\in\mathcal S\) and for every density matrix \(\rho\),
\begin{equation}
\label{eq:otoc1-one-gate-variance-front-window-lower}
\Var_{G_z\sim\Haar(U(4))}
\!\left(
\E_{\neq z}\mathrm{OTOC}^{(1)}_\rho(U_d(G_z))
\right)
=
\Omega(n^{-2}).
\end{equation}
Hence, combining 
\eqref{eq:otoc1-one-gate-variance-front-window-lower} and
\eqref{eq:otoc1-size-S-front-window}, we get
\[
\Var_{U_d}\!\left(\mathrm{OTOC}^{(1)}_\rho(U_d)\right)
\ge
|\mathcal S|\,\Omega(n^{-2})
=
\Theta(n^{3/2})\,\Omega(n^{-2})
=
\Omega(n^{-1/2}).
\]
This proves the claim.
\end{proof}

\subsection{Which gates are responsible for the OTOC fluctuations?}

This section identifies which gates matter, and by how much, for the
fluctuations of \(\mathrm{OTOC}^{(1)}\). This also explains the eye-shaped
region of gates shown in Fig.~\ref{fig:otoc1-gate-sensitivity-lens} in the main
text.

In the previous subsection, the one-gate variance identity
\eqref{eq:otoc1-conditional-mean-variance-exact} showed that the
conditional-mean variance generated by a single gate is
\begin{equation}
\label{eq:otoc1-eye-profile-variance-factorization}
\operatorname{Var}_{G\sim\Haar(U(4))}
\!\left(
\E_{\neq z}\mathrm{OTOC}^{(1)}_\rho(U_d(G))
\right)
=
c_A
\bigl(\mathsf P_{\ell,t}\mathsf F_{\ell,t}\bigr)^2,
\end{equation}
where
\[
c_A
\coloneqq
\frac{256}{225}\,
\operatorname{Var}_{G\sim\Haar(U(4))}(A(G))
\]
is the strictly positive universal constant from
\eqref{eq:otoc1-conditional-mean-variance-exact}. Consequently, to understand which
gates \(z=(\ell,t)\) matter for the fluctuations of the OTOC, it suffices to
estimate the product \(\mathsf P_{\ell,t}\mathsf F_{\ell,t}\).

By estimating the two propagation factors \(\mathsf P_{\ell,t}\) and
\(\mathsf F_{\ell,t}\), we can identify the region of gates whose contribution
to the OTOC fluctuations is inverse-polynomial. The intuition is that both
propagation factors must be appreciable: the endpoint front starting from qubit
\(1\) must have reached the gate by time \(t\), and a perturbation created at
the gate must still be able to reach qubit \(n\) by depth \(d\). The overlap of
these two requirements is the eye-shaped region shown in
Fig.~\ref{fig:otoc1-gate-sensitivity-lens}.

For simplicity, let us consider the critical depth \(d=5n/3\). At this depth,
the geometry becomes particularly clean: the proposition below describes how the
conditional-mean variance of a gate \(z=(\ell,t)\) decays with a Gaussian
profile as the gate moves away from the ballistic line \(t=\frac53\ell\) inside
the light cone.

Let us first spell out the relevant light cone in this case. A gate at position
\(z=(\ell,t)\) can influence the endpoint OTOC only if two deterministic
reachability conditions hold. First, the endpoint starting from \(Z_1\) must be
able to reach the bond \((\ell,\ell+1)\) by time \(t\), which gives
\(\ell\le t\). Second, after the gate acts, there must be enough time for a
perturbation at that bond to reach the final endpoint \(n\), which gives
\(n-\ell\le d-t\). At the critical depth \(d=5n/3\), the second condition is
equivalent to \(t\le \ell+\frac23 n\).

The light cone in the following proposition is exactly the set of gates satisfying
these two constraints. Inside this light cone, the proposition shows that the
conditional-mean variance is largest near the ballistic line
\(t=\frac53\ell\), and decays with a Gaussian-type penalty away from this line.
Outside the light cone, the conditional mean does not depend on the value of the
gate, and hence the variance with respect to that gate is zero. This estimate
is the quantitative origin of the eye-shaped region in
Fig.~\ref{fig:otoc1-gate-sensitivity-lens}.

\begin{proposition}[Eye-shaped profile of the conditional-mean variance]
\label{lem:otoc1-conditional-mean-variance-light-cone}
Assume \(n\in6\mathbb N\) and set \(d=5n/3\). Define the even-depth light cone by
\begin{equation}
\label{eq:otoc1-eye-light-cone-even}
\mathcal L_{\rm even}
\coloneqq
\left\{
(\ell,t):
\ell,t\ {\rm even},\quad
2\le \ell\le n-2,\quad
2\le t\le d-2,\quad
\ell\le t\le \ell+\frac23 n
\right\}.
\end{equation}
For \(z=(\ell,t)\in\mathcal L_{\rm even}\), set
\begin{equation}
\label{eq:otoc1-eye-R-def}
R_{\ell,t}
\coloneqq
\left|t-\frac53\ell\right|
\left(
\frac1t+\frac1{d-t}
\right)^{1/2}.
\end{equation}
Let \(c_A\) be the positive universal constant in
\eqref{eq:otoc1-eye-profile-variance-factorization}. Then there exists a
universal constant \(C\ge1\) such that, for every
\(z=(\ell,t)\in\mathcal L_{\rm even}\) and every density matrix \(\rho\), the
conditional-mean variance satisfies
\begin{equation}
\label{eq:otoc1-conditional-variance-eye-two-sided}
\frac{c_A}{C\,n^2}\,
e^{-C R_{\ell,t}^2}
\le
\operatorname{Var}_{G\sim\Haar(U(4))}
\!\left(
\E_{\neq z}\mathrm{OTOC}^{(1)}_\rho(U_d(G))
\right)
\le
C c_A\,e^{-R_{\ell,t}^2/C}.
\end{equation}
Thus, up to polynomial factors in \(n\), the size of the conditional-mean
variance is governed by the Gaussian penalty \(e^{-R_{\ell,t}^2}\). In particular, for every fixed \(K>0\), if \(R_{\ell,t}^2\le K\), then
\[
\operatorname{Var}_{G\sim\Haar(U(4))}
\!\left(
\E_{\neq z}\mathrm{OTOC}^{(1)}_\rho(U_d(G))
\right)
=
\Omega_K(n^{-2}).
\]
More generally, if \(R_{\ell,t}=O(\sqrt{\log n})\), then the lower bound in
\eqref{eq:otoc1-conditional-variance-eye-two-sided} is still
inverse-polynomial in \(n\). On the other hand, if
\(R_{\ell,t}=\omega(\sqrt{\log n})\), the upper bound is superpolynomially
small, and if \(R_{\ell,t}=\Omega(\sqrt n)\), it is exponentially small in
\(n\).

Geometrically, the condition \(R_{\ell,t}^2\le K\) says that the gate lies
within distance
\[
O_K\!\left[
\left(\frac1t+\frac1{d-t}\right)^{-1/2}
\right]
=
O_K\!\left(
\sqrt{\frac{t(d-t)}{d}}
\right)
\]
from the ballistic line \(t=\frac53\ell\). This width is of order \(\sqrt n\)
in the middle of the circuit and shrinks near the initial and final depths,
giving the eye-shaped region in Fig.~\ref{fig:otoc1-gate-sensitivity-lens}.
\end{proposition} 
\begin{proof}
Let \(C\) denote a universal constant, allowed to change from line to line. Write
\begin{equation}
\label{eq:otoc1-eye-Vz-def}
V_z\coloneqq
\operatorname{Var}_{G\sim\Haar(U(4))}
\!\left(
\E_{\neq z}\mathrm{OTOC}^{(1)}_\rho(U_d(G))
\right).
\end{equation}
By the variance identity \eqref{eq:otoc1-eye-profile-variance-factorization},
\(V_z=c_A(\mathsf P_{\ell,t}\mathsf F_{\ell,t})^2\). Thus the problem is to
estimate the two factors \(\mathsf P_{\ell,t}\) and \(\mathsf F_{\ell,t}\).

We first rewrite these two factors in a form where the estimates are transparent.
Start with \(\mathsf P_{\ell,t}\). In the coefficient formula
\eqref{eq:P-coeff-r-form-front-window}, expand
\[
\frac{1-y^{n-\ell}}{1-y^n}
=
\sum_{j\ge0}y^{jn}
-
\sum_{j\ge0}y^{n-\ell+jn}.
\]
The coefficient extracted there has degree \((t-\ell)/2\). Since
\(z\in\mathcal L_{\rm even}\) as defined in \eqref{eq:otoc1-eye-light-cone-even}, \(0\le (t-\ell)/2\le n/3<n\). Hence no term
with \(j\ge1\) can contribute. Only the two \(j=0\) terms remain. Using the
same coefficient identity as in the derivation of
\eqref{eq:otoc1-P-binomial-front-window}, we get
\begin{equation}
\label{eq:otoc1-eye-P-two-term}
\mathsf P_{\ell,t}
=
\psi_{t-1}\!\left(\frac{t-\ell}{2}\right)
-
4^{-(n-\ell)}
\psi_{t-1}\!\left(\frac{t+\ell-2n}{2}\right).
\end{equation}
Here and below
\begin{equation}
\label{eq:otoc1-eye-psi-def}
\psi_M(r)
\coloneqq
4^{M-r}5^{-M}
\left[
\binom Mr-\binom M{r-1}
\right],
\end{equation}
with binomial coefficients outside \(0\le r\le M\) set to zero.

The same argument gives the corresponding expression for \(\mathsf F_{\ell,t}\).
Indeed, in \eqref{eq:F-coeff-r-form-front-window} the extracted degree is
\((d-t-(n-\ell))/2\). Since \(d=5n/3\), the light-cone condition
\(t\le \ell+2n/3\) in \eqref{eq:otoc1-eye-light-cone-even} gives \(d-t\ge n-\ell\), so this degree is nonnegative.
Also \(t\ge\ell\), again by \eqref{eq:otoc1-eye-light-cone-even}, gives
\[
\frac{d-t-(n-\ell)}2
\le
\frac{d-n}{2}
=
\frac n3
<
n.
\]
Expanding \((1-y^\ell)/(1-y^n)=\sum_{j\ge0}y^{jn}-\sum_{j\ge0}y^{\ell+jn}\),
again only the \(j=0\) terms can contribute. By the coefficient extraction in
\eqref{eq:otoc1-F-binomial-front-window},
\begin{equation}
\label{eq:otoc1-eye-F-two-term}
\mathsf F_{\ell,t}
=
\frac54
\psi_{d-t-1}\!\left(\frac{d-t-(n-\ell)}{2}\right)
-
\frac54\,4^{-\ell}
\psi_{d-t-1}\!\left(\frac{d-t-(n+\ell)}{2}\right).
\end{equation}

We now check, using the two-term formulas \eqref{eq:otoc1-eye-P-two-term} and \eqref{eq:otoc1-eye-F-two-term}, that in both cases the second term is at most a fixed fraction
of the first. We use the elementary inequality
\begin{equation}
\label{eq:otoc1-eye-binomial-ratio}
\frac{\binom M{r-L}}{\binom Mr}
=
\prod_{j=0}^{L-1}
\frac{r-j}{M-r+j+1}
\le
\left(\frac r{M-r+1}\right)^L,
\qquad 0\le L\le r.
\end{equation}
For \(\mathsf P_{\ell,t}\), set \(M=t-1\), \(r=(t-\ell)/2\), and
\(L=n-\ell\). If the second term is zero there is nothing to prove. Otherwise
\(r-L\ge0\), hence \(r\ge L\). Since \(r=(t-\ell)/2\le n/3\), this implies
\(n-\ell\le n/3\), i.e.~\(\ell\ge 2n/3\). Therefore
\((t-\ell)/(t+\ell)\le 1/3\). Using \eqref{eq:otoc1-eye-binomial-ratio} and the
explicit formula for \(\psi_M\) in \eqref{eq:otoc1-eye-psi-def}, we obtain
\[
4^{-(n-\ell)}
\frac{
\psi_{t-1}\!\left(\frac{t+\ell-2n}{2}\right)
}{
\psi_{t-1}\!\left(\frac{t-\ell}{2}\right)
}
\le
\left(\frac{t-\ell}{t+\ell}\right)^{n-\ell}
\frac{(2n-\ell)(t+\ell)}
{2\ell\left(n+\frac{t-\ell}{2}\right)}.
\]
In the present regime, \(t+\ell=2\ell+(t-\ell)\le3\ell\), while
\(2n-\ell\le2n\) and \(n+(t-\ell)/2\ge n\). Thus the last factor is at most
\(3\). Since \(n-\ell\ge2\), the whole ratio is at most
\(3(1/3)^{n-\ell}\le1/3<1/2\). Hence
\begin{equation}
\label{eq:otoc1-eye-P-leading-comparison}
\frac12\,
\psi_{t-1}\!\left(\frac{t-\ell}{2}\right)
\le
\mathsf P_{\ell,t}
\le
\psi_{t-1}\!\left(\frac{t-\ell}{2}\right).
\end{equation}

For \(\mathsf F_{\ell,t}\) the argument is symmetric. If the second term is
nonzero, then \(d-t\ge n+\ell\). Since \(d=5n/3\), this gives
\(t+\ell\le2n/3\), and in particular \(\ell\le n/3\). Applying the same
binomial-ratio estimate \eqref{eq:otoc1-eye-binomial-ratio} with \(M=d-t-1\),
\(r=(d-t-(n-\ell))/2\), and \(L=\ell\), gives
\[
4^{-\ell}
\frac{
\psi_{d-t-1}\!\left(\frac{d-t-(n+\ell)}{2}\right)
}{
\psi_{d-t-1}\!\left(\frac{d-t-(n-\ell)}{2}\right)
}
\le
\left(
\frac{d-t-(n-\ell)}{d-t+n-\ell}
\right)^\ell
\frac{(n+\ell)(d-t+n-\ell)}
{(n-\ell)(d-t+n+\ell)}.
\]
Here \(d-t-(n-\ell)\le d-n=2n/3\), because \(t\ge\ell\), while
\(d-t+n-\ell\ge2n\), because \(d-t\ge n+\ell\). Thus the base is at most
\(1/3\). The remaining factor is at most
\((n+\ell)/(n-\ell)\le2\), since \(\ell\le n/3\). As \(\ell\ge2\), the ratio is
at most \(2(1/3)^\ell\le2/9<1/2\). Therefore
\begin{equation}
\label{eq:otoc1-eye-F-leading-comparison}
\frac58\,
\psi_{d-t-1}\!\left(\frac{d-t-(n-\ell)}{2}\right)
\le
\mathsf F_{\ell,t}
\le
\frac54\,
\psi_{d-t-1}\!\left(\frac{d-t-(n-\ell)}{2}\right).
\end{equation}

It remains to estimate the two leading \(\psi\)-terms. Let
\(b_M(r)=\binom Mr(1/5)^r(4/5)^{M-r}\). Since
\(\binom M{r-1}/\binom Mr=r/(M-r+1)\), the definition \eqref{eq:otoc1-eye-psi-def} gives
\begin{equation}
\label{eq:otoc1-eye-psi-binomial-rewrite}
\psi_M(r)
=
b_M(r)\frac{M+1-2r}{M-r+1}.
\end{equation}
Thus
\begin{equation}
\label{eq:otoc1-eye-leading-psi-values}
\psi_{t-1}\!\left(\frac{t-\ell}{2}\right)
=
b_{t-1}\!\left(\frac{t-\ell}{2}\right)
\frac{2\ell}{t+\ell},
\qquad
\psi_{d-t-1}\!\left(\frac{d-t-(n-\ell)}{2}\right)
=
b_{d-t-1}\!\left(\frac{d-t-(n-\ell)}{2}\right)
\frac{2(n-\ell)}{d-t+n-\ell}.
\end{equation}

We use the local binomial lower bound from
Lemma~\ref{lem:crude-local-binomial-lower}. The matching upper bound follows
from the same Robbins estimate used in the proof of that lemma: for
\(0\le r\le M\),
\begin{equation}
\label{eq:otoc1-eye-local-binomial-two-sided}
\frac1C\frac1{\sqrt{M+1}}
\exp\!\left[
-C\frac{(r-M/5)^2}{M+1}
\right]
\le
b_M(r)
\le
\frac C{\sqrt{M+1}}
\exp\!\left[
-\frac1C\frac{(r-M/5)^2}{M+1}
\right].
\end{equation}
For the past term,
\[
\frac{t-\ell}{2}-\frac{t-1}{5}
=
\frac{3t-5\ell+2}{10},
\]
and for the future term,
\[
\frac{d-t-(n-\ell)}2-\frac{d-t-1}{5}
=
\frac{5\ell-3t+2}{10}.
\]
Using \eqref{eq:otoc1-eye-P-leading-comparison}, \eqref{eq:otoc1-eye-F-leading-comparison}, \eqref{eq:otoc1-eye-leading-psi-values}, and the two-sided binomial estimate \eqref{eq:otoc1-eye-local-binomial-two-sided}, the shifts by \(2\) only change universal constants: since \(t,d-t\ge2\), one
can use \((x+2)^2\le2x^2+8\) and
\((x+2)^2\ge x^2/2-8\). Writing \(\delta\coloneqq 3t-5\ell\), we get
\begin{equation}
\label{eq:otoc1-eye-P-gaussian-bounds}
\frac1C
\frac{\ell}{(t+\ell)\sqrt t}
e^{-C\delta^2/t}
\le
\mathsf P_{\ell,t}
\le
C
\frac{\ell}{(t+\ell)\sqrt t}
e^{-\delta^2/(Ct)},
\end{equation}
and
\begin{equation}
\label{eq:otoc1-eye-F-gaussian-bounds}
\frac1C
\frac{n-\ell}{(d-t+n-\ell)\sqrt{d-t}}
e^{-C\delta^2/(d-t)}
\le
\mathsf F_{\ell,t}
\le
C
\frac{n-\ell}{(d-t+n-\ell)\sqrt{d-t}}
e^{-\delta^2/(C(d-t))}.
\end{equation}

Multiplying the two upper bounds in \eqref{eq:otoc1-eye-P-gaussian-bounds} and \eqref{eq:otoc1-eye-F-gaussian-bounds} gives
\begin{equation}
\label{eq:otoc1-eye-PF-upper-with-B}
\mathsf P_{\ell,t}\mathsf F_{\ell,t}
\le
C B_{\ell,t}
\exp\!\left[
-\frac1C\delta^2
\left(\frac1t+\frac1{d-t}\right)
\right],
\end{equation}
where
\begin{equation}
\label{eq:otoc1-eye-B-def}
B_{\ell,t}
\coloneqq
\frac{\ell(n-\ell)}
{(t+\ell)(d-t+n-\ell)\sqrt{t(d-t)}}.
\end{equation}
Since \(B_{\ell,t}\le1\), \(\delta=3(t-\frac53\ell)\), and \(R_{\ell,t}\) was defined in
\eqref{eq:otoc1-eye-R-def}, \eqref{eq:otoc1-eye-PF-upper-with-B} gives
\(\mathsf P_{\ell,t}\mathsf F_{\ell,t}\le
C\exp[-R_{\ell,t}^2/C]\). Squaring and using
\eqref{eq:otoc1-eye-profile-variance-factorization}, equivalently \(V_z=c_A(\mathsf P_{\ell,t}\mathsf F_{\ell,t})^2\) by \eqref{eq:otoc1-eye-Vz-def}, proves the upper bound in
\eqref{eq:otoc1-conditional-variance-eye-two-sided}.

For the lower bound, the lower sides of \eqref{eq:otoc1-eye-P-gaussian-bounds} and \eqref{eq:otoc1-eye-F-gaussian-bounds} give
\begin{equation}
\label{eq:otoc1-eye-PF-lower-with-B}
\mathsf P_{\ell,t}\mathsf F_{\ell,t}
\ge
\frac1C
B_{\ell,t}
\exp\!\left[
-C\delta^2
\left(\frac1t+\frac1{d-t}\right)
\right].
\end{equation}
We now show that \(B_{\ell,t}\) is at least \(n^{-1}\), up to the same
Gaussian penalty. Put \(a=\ell\) and \(b=n-\ell\). Then \(a,b\ge2\),
\(t=(5a+\delta)/3\), and \(d-t=(5b-\delta)/3\). Hence
\begin{equation}
\label{eq:otoc1-eye-B-exact-ab}
B_{\ell,t}
=
\frac{27ab}
{(8a+\delta)(8b-\delta)
\sqrt{(5a+\delta)(5b-\delta)}}.
\end{equation}
When \(\delta=0\), this equals \(27/(320\sqrt{ab})\), which is at least
\(1/(Cn)\), because \(ab\le n^2/4\).

Assume now \(\delta>0\). The factors \(8b-\delta\) and \(5b-\delta\) are no
larger than their values at \(\delta=0\), so they can only make
\(B_{\ell,t}\) larger. Therefore
\begin{equation}
\label{eq:otoc1-eye-B-lower-positive-shift}
B_{\ell,t}
\ge
\frac{27}{320\sqrt{ab}}
\left(1+\frac{\delta}{8a}\right)^{-1}
\left(1+\frac{\delta}{5a}\right)^{-1/2}
\ge
\frac{1}{Cn}
\left(1+\frac{\delta}{5a}\right)^{-3/2}.
\end{equation}
Since \(t\) and \(\ell\) are even, either \(\delta=0\) or \(|\delta|\ge2\). Thus, in
the present case, \(\delta\ge2\). Also
\begin{equation}
\label{eq:otoc1-eye-local-shift-over-t-lower}
\frac{\delta}{t}
=
\frac{3\delta}{5a+\delta}
\ge
\frac1a,
\end{equation}
because \(\delta(3a-1)\ge5a\) for \(a\ge2\) and \(\delta\ge2\). Hence
\(\delta/(5a)\le \delta^2/(5t)\). Using \(\log(1+x)\le x\), we obtain
\begin{equation}
\label{eq:otoc1-eye-positive-shift-factor}
\left(1+\frac{\delta}{5a}\right)^{-3/2}
\ge
\exp\!\left[
-C\delta^2
\left(\frac1t+\frac1{d-t}\right)
\right].
\end{equation}
Thus \(B_{\ell,t}\ge C^{-1}n^{-1}
\exp[-C\delta^2(1/t+1/(d-t))]\) when \(\delta>0\).

The case \(\delta<0\) is the same, with \(b=n-\ell\) and \(d-t\) in place of
\(a=\ell\) and \(t\). Writing \(E=-\delta\ge2\), one has
\(E/(d-t)=3E/(5b+E)\ge1/b\), hence
\(E/(5b)\le E^2/(5(d-t))\), and the same argument gives
\begin{equation}
\label{eq:otoc1-eye-B-lower-negative-shift}
B_{\ell,t}
\ge
\frac1{Cn}
\exp\!\left[
-C\delta^2
\left(\frac1t+\frac1{d-t}\right)
\right].
\end{equation}
Combining the three cases and using the definition \eqref{eq:otoc1-eye-R-def}, namely
\(\delta^2(1/t+1/(d-t))=9R_{\ell,t}^2\), we have
\begin{equation}
\label{eq:otoc1-eye-B-final-lower}
B_{\ell,t}\ge \frac1{Cn}e^{-C R_{\ell,t}^2}.
\end{equation}
Combining \eqref{eq:otoc1-eye-PF-lower-with-B} with \eqref{eq:otoc1-eye-B-final-lower},
we get \(\mathsf P_{\ell,t}\mathsf F_{\ell,t}\ge C^{-1}n^{-1}
e^{-C R_{\ell,t}^2}\). Squaring and using
\eqref{eq:otoc1-eye-profile-variance-factorization} proves the lower bound in
\eqref{eq:otoc1-conditional-variance-eye-two-sided}.

The stated consequences now follow directly from \eqref{eq:otoc1-conditional-variance-eye-two-sided}. If \(R_{\ell,t}^2\le K\), then the
lower bound gives \(V_z\ge c_K n^{-2}\), i.e.~\(V_z=\Omega_K(n^{-2})\). If
\(R_{\ell,t}=O(\sqrt{\log n})\), the same lower bound is inverse-polynomial in
\(n\). Conversely, the upper bound gives \(V_z\le Cc_A e^{-R_{\ell,t}^2/C}\);
therefore \(V_z=e^{-\omega(\log n)}\) when
\(R_{\ell,t}=\omega(\sqrt{\log n})\), and \(V_z=e^{-\Omega(n)}\) when
\(R_{\ell,t}=\Omega(\sqrt n)\).

Finally, by the definition \eqref{eq:otoc1-eye-R-def} of \(R_{\ell,t}\), the condition
\(R_{\ell,t}^2\le K\) is equivalent to
\begin{equation}
\label{eq:otoc1-eye-tube-width}
\left|t-\frac53\ell\right|
\le
\sqrt K
\left(\frac1t+\frac1{d-t}\right)^{-1/2}
=
\sqrt K\,\sqrt{\frac{t(d-t)}{d}}.
\end{equation}
This is exactly the claimed eye-shaped tube around the ballistic line
\(t=\frac53\ell\).
\end{proof}

\newpage
\section{Classical simulation of the infinite-temperature endpoint
\texorpdfstring{\(\mathrm{OTOC}^{(1)}\)}{OTOC(1)}}
\label{sec:eye-reduction-classical-simulation}

The goal of this section is to provide a classical simulation algorithm with subexponential runtime to estimate the infinite-temperature endpoint OTOC of a given one-dimensional Haar brickwork circuit at the critical depth \(d=5n/3\).
The algorithm is presented in Table~\ref{tab:endpoint-otoc1-explicit-simulation};
its accuracy and runtime are proved in
Theorem~\ref{thm:endpoint-sim-otoc1-runtime}. The idea is to keep the gates
inside a slightly enlarged eye-shaped region depicted in Fig.~\ref{fig:otoc1-gate-sensitivity-lens}(a) of the main text and average over the gates outside it. For the proof, we first control the error of this replacement, and then explain how to estimate the resulting conditional mean.

We use the circuit convention of Section~\ref{subsec:endpoint-otoc1-setup},
with \(n\in6\mathbb N\) and \(U_d=L_1\cdots L_d\). All physical gates are independent and
Haar-random in \(U(4)\). Write the OTOC as 
\[F_\infty(U_d)\coloneqq2^{-n}\Tr[(U_d^\dagger Z_1U_dZ_n)^2].\]
A gate \(z=(\ell,t)\) acts on \((\ell,\ell+1)\) in layer \(t\).
By Lemma~\ref{lem:otoc-causal-reductions}, only gates in the causal overlap
\[
    \mathcal R_d=\{(\ell,t):\ell\le t,\ n-\ell\le d-t\}
\]
can affect \(F_\infty\).
For \(R\ge1\), the set of gates we retain is
\begin{equation}
\label{eq:endpoint-sim-eye-def}
    \mathcal E_R=\{(\ell,t)\in\mathcal R_d:t<d,
                 |t-5\ell/3|\le R\sigma_{\ell,t}\},
    \qquad
    \sigma_{\ell,t}^2=
    \begin{cases}
        t,&t\le5\ell/3,\\
        d-t,&t>5\ell/3.
    \end{cases}
\end{equation}
Let \(W_R\) be the largest number of retained gates in any layer. Since \(\sigma_{\ell,t}\le\sqrt d\), their bonds satisfy
\(|\ell-3t/5|\le3R\sqrt d/5\), so \(W_R=O(R\sqrt n)\). The conditional mean of the OTOC is
\[
\widehat F_R\coloneqq\mathbb E[F_\infty(U_d)\mid U_{\mathcal E_R}],
\]
where the expectation is over the gates outside \(\mathcal E_R\), while the gates in \(\mathcal E_R\) are kept fixed.

To show that \(F_\infty(U_d)\) is well approximated by the conditional mean
\(\widehat F_R\), we first bound how much the OTOC can change when a single
gate is changed. Intuitively, this change must be small if either the butterfly
is unlikely to have reached that gate or the gate is unlikely to affect the
probe before the end of the circuit.

To quantify this, let \(p(u,v)\) denote the probability that the
Haar-averaged Pauli string obtained by evolving \(Z_1\) touches the bond
\((v,v+1)\) just before layer \(u\). We will bound this probability below
using the endpoint Markov chain introduced in
Section~\ref{subsec:endpoint-otoc1-exact-formula}.

\begin{lemma}[Effect of changing one gate]
\label{lem:endpoint-sim-local-gate}
Consider a gate \(z=(\ell,t)\), with \(t<d\), and replace it by either
\(G_1\) or \(G_2\), with \(G_1,G_2\in U(4)\). Then
\begin{equation}
\label{eq:endpoint-sim-local-change}
    \mathbb E_{\ne z}
    \left|
        F_\infty(G_1)-F_\infty(G_2)
    \right|
    \le
    4\min\left\{
        \sqrt{p(t,\ell)},
        \sqrt{p(d-t,n-\ell)}
    \right\}.
\end{equation}
Here \(F_\infty(G_i)\) denotes the OTOC of the circuit in which the gate
at \(z\) is fixed to \(G_i\), and the expectation is over all other Haar
gates.
\end{lemma}

\begin{proof}
Fix all gates except the one at \(z\), and let \(e=(\ell,\ell+1)\) be
its bond. We write the circuit around this gate as
\[
    U_d=U_{<z}G_zU_{>z},
\]
and define the butterfly immediately before the gate and the probe evolved
backwards to the same location by
\[
    B_z=U_{<z}^\dagger Z_1U_{<z},
    \qquad
    M_z=U_{>z}Z_nU_{>z}^\dagger.
\]
With the normalized trace \(\tau(X)=2^{-n}\Tr(X)\), the OTOC with
\(G_z=G_i\) is
\[
    F_\infty(G_i)
    =
    \tau\!\left[
        (G_i^\dagger B_zG_iM_z)^2
    \right].
\]

We first bound the change using the butterfly. Decompose \(B_z\) into the
Pauli strings that are trivial on \(e\) and those that touch \(e\):
\[
    B_z=B_z^{(0)}+B_z^{(e)}.
\]
Let
\[
    w_z\coloneqq
    \|B_z^{(e)}\|_{2,\mathrm{av}}^2,
    \qquad
    \|X\|_{2,\mathrm{av}}^2
    \coloneqq\tau(X^\dagger X).
\]
Since \(B_z^{(0)}\) is identity on \(e\), it commutes with both \(G_1\)
and \(G_2\). Therefore, setting
\[
    A_i=G_i^\dagger B_zG_i,
    \qquad
    D=A_1-A_2,
\]
the part \(B_z^{(0)}\) cancels, and
\[
    \|D\|_{2,\mathrm{av}}
    \le
    \|G_1^\dagger B_z^{(e)}G_1\|_{2,\mathrm{av}}
    +
    \|G_2^\dagger B_z^{(e)}G_2\|_{2,\mathrm{av}}
    =
    2\sqrt{w_z}.
\]

Using
\[
    F_\infty(G_i)=\tau(A_iM_zA_iM_z),
\]
we obtain
\[
\begin{aligned}
    F_\infty(G_1)-F_\infty(G_2)
    &=
    \tau(DM_zA_1M_z)
    +
    \tau(A_2M_zDM_z).
\end{aligned}
\]
Since \(A_i\) and \(M_z\) are unitaries, Cauchy--Schwarz gives
\[
    |F_\infty(G_1)-F_\infty(G_2)|
    \le
    2\|D\|_{2,\mathrm{av}}
    \le
    4\sqrt{w_z}.
\]
Averaging over all gates except \(z\), and using Jensen's inequality,
\[
    \mathbb E_{\ne z}
    |F_\infty(G_1)-F_\infty(G_2)|
    \le
    4\sqrt{\mathbb E_{\ne z}w_z}.
\]
By definition of the Haar-averaged Pauli process,
\(\mathbb E_{\ne z}w_z=p(t,\ell)\). Hence
\[
    \mathbb E_{\ne z}
    |F_\infty(G_1)-F_\infty(G_2)|
    \le
    4\sqrt{p(t,\ell)}.
\]

The same argument can be applied from the probe side, using \(M_z\)
instead of \(B_z\). Reflecting the chain and evolving backwards from
\(Z_n\) gives the same endpoint process, with time \(d-t\) and distance
\(n-\ell\). Therefore
\[
    \mathbb E_{\ne z}
    |F_\infty(G_1)-F_\infty(G_2)|
    \le
    4\sqrt{p(d-t,n-\ell)}.
\]
Taking the smaller of the two bounds proves
Eq.~\eqref{eq:endpoint-sim-local-change}.
\end{proof}

The previous lemma reduces the problem to bounding the probability \(p(u,v)\), namely the probability that the Haar-averaged Pauli string obtained by evolving \(Z_1\) has reached the bond \((v,v+1)\) just before layer \(u\). We now use the endpoint Markov chain introduced in Section~\ref{subsec:endpoint-otoc1-exact-formula} to show that \(p(u,v)\) decays rapidly when the bond lies ahead of the propagating front.
\begin{lemma}[Tail bound for the endpoint process]
\label{lem:endpoint-sim-touching-tail}
For every physical bond \((v,v+1)\) acted on in layer \(u\ge1\),
\begin{equation}
\label{eq:endpoint-sim-touching-tail}
    p(u,v)
    \le
    \frac52
    \exp\!\left[
        -\frac{(5v-3u)_+^2}{200u}
    \right],
    \qquad
    x_+\coloneqq\max\{x,0\}.
\end{equation}
\end{lemma}

\begin{proof}
The idea is that if the Pauli string touches the bond \((v,v+1)\), then its right endpoint must have propagated at least up to that bond. We can therefore use directly the endpoint Markov chain of
Section~\ref{subsec:endpoint-otoc1-exact-formula}.

Recall that \(Y_m\) is the cell containing the right endpoint immediately
after the \(m\)-th odd layer, and that its transition matrix is \(Q\) in
Eq.~\eqref{eq:otoc1-Q}. Define the probability that the endpoint has reached
cell \(a\) or farther by
\[
    q_s(a)
    \coloneqq
    \Pr(Y_{s+1}\ge a)
    =
    \sum_{b=a}^{n/2}e_b^{\mathsf T}Q^se_1 .
\]
A convenient bound on this tail is
\begin{equation}
\label{eq:endpoint-sim-Q-tail}
    q_s(a)
    \le
    \frac54
    \Pr\!\left[
        \operatorname{Bin}(2s+1,1/5)\le s+1-a
    \right].
\end{equation}
To prove it, let us denote the right-hand-side as
\[
    h_s(a)
    \coloneqq
    \frac54
    \Pr\!\left[
        X_s\le s+1-a
    \right],
    \qquad
    X_s\sim\operatorname{Bin}(2s+1,1/5).
\]
From the transition matrix \(Q\) in Eq.~\eqref{eq:otoc1-Q}, the endpoint
tail satisfies
\[
    q_{s+1}(a)
    \le
    \frac{16}{25}q_s(a-1)
    +
    \frac8{25}q_s(a)
    +
    \frac1{25}q_s(a+1).
\]
We now show that \(h_s(a)\) satisfies the same relation with equality.
Indeed, \(X_{s+1}\) has the same distribution as
\(X_s+B_1+B_2\), where \(B_1,B_2\) are independent
Bernoulli\((1/5)\) variables. Since
\[
    \Pr(B_1+B_2=0)=\frac{16}{25},\qquad
    \Pr(B_1+B_2=1)=\frac8{25},\qquad
    \Pr(B_1+B_2=2)=\frac1{25},
\]
conditioning on \(B_1+B_2\) gives
\[
    h_{s+1}(a)
    =
    \frac{16}{25}h_s(a-1)
    +
    \frac8{25}h_s(a)
    +
    \frac1{25}h_s(a+1).
\]
Thus, if \(q_s(a)\le h_s(a)\), the same inequality also holds after one
more step.

At \(s=0\), we have \(q_0(a)=h_0(a)\) for every \(a\). At the right boundary,
\(q_s(n/2+1)=0\le h_s(n/2+1)\). At the left boundary,
\(q_s(1)=1\), while it is easy to check that
\(h_s(1)=\frac54\Pr[\operatorname{Bin}(2s+1,1/5)\le s]\ge1\).
Therefore Eq.~\eqref{eq:endpoint-sim-Q-tail} follows by induction.

We now translate this endpoint bound into a bound on \(p(u,v)\).
If \(u\) is even, then \(v\) is even, and touching \((v,v+1)\) requires
the endpoint to have reached cell \(v/2\). Hence
\[
    p(u,v)\le q_{u/2-1}(v/2).
\]
If \(u\) is odd, then \(v\) is odd. The even layer immediately before
\(u\) can move the endpoint by at most one cell, so
\[
    p(u,v)\le q_{(u-3)/2}((v-1)/2).
\]
Applying Eq.~\eqref{eq:endpoint-sim-Q-tail} in the two cases gives, for
\(u\ge3\) and \(v\ge2\),
\begin{equation}
\label{eq:endpoint-sim-touching-binomial}
    p(u,v)
    \le
    \frac54
    \Pr\!\left[
        \operatorname{Bin}(u-2,1/5)
        \le
        \frac{u-v}{2}
    \right].
\end{equation}
For even \(u\), we have only weakened the bound by dropping one
Bernoulli trial. The remaining boundary cases are immediate.

Finally, let \(X\sim\operatorname{Bin}(u-2,1/5)\). The distance between
its mean and the threshold above is
\[
    \mathbb EX-\frac{u-v}{2}
    =
    \frac{5v-3u-4}{10}.
\]
If \(5v-3u\ge8\), Hoeffding's inequality therefore gives
\[
    p(u,v)
    \le
    \frac54
    \exp\!\left[
        -\frac{(5v-3u)^2}{200u}
    \right].
\]
If \(5v-3u<8\), the bound in
Eq.~\eqref{eq:endpoint-sim-touching-tail} is larger than one and is
therefore trivial. This proves the lemma.
\end{proof}
Equation~\eqref{eq:endpoint-sim-touching-tail} shows that a gate is unlikely
to be reached by the butterfly if it lies sufficiently far ahead of the
ballistic front. The same bound applies from the probe side: evolving
backwards from \(Z_n\), and reflecting the chain from right to left, gives
the same endpoint process with \(u=d-t\) and \(v=n-\ell\). Hence
\(p(d-t,n-\ell)\) controls the probability that a gate at \((\ell,t)\)
can still influence the probe by the final depth \(d\).

Combining these two bounds with
Lemma~\ref{lem:endpoint-sim-local-gate} shows that the influence of gates
outside the eye is exponentially small.\begin{lemma}[Error from averaging outside the eye]
\label{lem:endpoint-sim-eye-error}
For every \(R\ge1\),
\begin{equation}
\label{eq:endpoint-sim-eye-expectation}
    \mathbb E\left|F_\infty(U_d)-\widehat F_R\right|
    \le\frac{100}{3}n^2e^{-R^2/100}.
\end{equation}
The expectation is over the full Haar circuit, including the gates inside
the eye.
\end{lemma}

\begin{proof}
We first bound the influence of a single gate \(z=(\ell,t)\in\mathcal R_d\).
Suppose that \(t\le5\ell/3\). In this case the gate lies ahead of the
butterfly front, so we use the first term in
Lemma~\ref{lem:endpoint-sim-local-gate}. Together with
Lemma~\ref{lem:endpoint-sim-touching-tail}, this gives
\[
    \mathbb E_{\ne z}|F_\infty(G_1)-F_\infty(G_2)|
    \le4\sqrt{\frac52}\exp\!\left[-\frac{(5\ell-3t)^2}{400t}\right].
\]
Since \(5\ell-3t=-3(t-5\ell/3)\) and
\(\sigma_{\ell,t}^2=t\), we obtain
\[
    \mathbb E_{\ne z}|F_\infty(G_1)-F_\infty(G_2)|
    \le4\sqrt{\frac52}\exp\!\left[-\frac{9(t-5\ell/3)^2}{400\sigma_{\ell,t}^2}\right].
\]

If \(t>5\ell/3\), we instead use the second term in
Lemma~\ref{lem:endpoint-sim-local-gate}, which controls whether the gate can
still influence the probe. Applying
Lemma~\ref{lem:endpoint-sim-touching-tail} to \(p(d-t,n-\ell)\), and using
\(3d=5n\), gives the same bound with
\(\sigma_{\ell,t}^2=d-t\). Thus, in both cases,
\begin{equation}
\label{eq:endpoint-sim-one-gate-bound}
    \mathbb E_{\ne z}|F_\infty(G_1)-F_\infty(G_2)|
    \le40\exp\!\left[-\frac{(t-5\ell/3)^2}{100\sigma_{\ell,t}^2}\right].
\end{equation}
Gates outside \(\mathcal R_d\) have no effect on the OTOC by
Lemma~\ref{lem:otoc-causal-reductions}.

Now consider a gate outside the retained region \(\mathcal E_R\).
If it lies in \(\mathcal R_d\), then by definition of \(\mathcal E_R\),
\[
    |t-5\ell/3|>R\sigma_{\ell,t}.
\]
Equation~\eqref{eq:endpoint-sim-one-gate-bound} therefore shows that changing
this gate changes the OTOC, on average, by at most
\(40e^{-R^2/100}\). Gates outside \(\mathcal R_d\) contribute zero.

It remains to average all these gates simultaneously. Let \(V\) denote the
gates in \(\mathcal E_R\), let \(X\) denote the remaining gates, and let
\(X'\) be an independent Haar copy of \(X\). By definition,
\[
    \widehat F_R(V)=\mathbb E_{X'}F_\infty(V,X').
\]
Hence, by Jensen's inequality,
\[
    \mathbb E|F_\infty-\widehat F_R|
    \le\mathbb E_{V,X,X'}|F_\infty(V,X)-F_\infty(V,X')|.
\]

We replace the gates of \(X\) by the corresponding gates of \(X'\) one at
a time. By the triangle inequality, the total change is at most the sum of
the changes produced by the individual replacements. Each replacement has
expected size at most \(40e^{-R^2/100}\) by
Eq.~\eqref{eq:endpoint-sim-one-gate-bound}. Therefore
\[
    \mathbb E|F_\infty-\widehat F_R|
    \le40|\mathcal G_d|e^{-R^2/100},
\]
where \(\mathcal G_d\) is the set of physical gates in the circuit. Since
\[
    |\mathcal G_d|\le\frac{nd}{2}=\frac{5n^2}{6},
\]
we conclude that
\[
    \mathbb E|F_\infty(U_d)-\widehat F_R|
    \le\frac{100}{3}n^2e^{-R^2/100}.
\]
\end{proof}

We are left with the task of estimating the conditional mean \(\widehat F_R\).
As in the proof of Lemma~\ref{thm:endpoint-otoc1-exact-formula}, expand the
evolved butterfly in the Pauli basis as
\[
    U_d^\dagger Z_1U_d=\sum_P c_U(P)P.
\]
Pauli orthogonality then gives
\[
    F_\infty(U_d)=\sum_P |c_U(P)|^2 f(P),
\]
where \(f(P)=1\) if \(P_n\in\{I,Z\}\) and \(f(P)=-1\) if
\(P_n\in\{X,Y\}\). Therefore
\begin{equation}
\label{eq:endpoint-sim-conditional-distribution}
    \widehat F_R=\sum_P \overline\mu_R(P)f(P),
    \qquad
    \overline\mu_R(P)\coloneqq
    \mathbb E[|c_U(P)|^2\mid U_{\mathcal E_R}].
\end{equation}
Since \(\sum_P|c_U(P)|^2=1\), \(\overline\mu_R\) is a probability
distribution. Thus, to estimate \(\widehat F_R\), it is enough to sample a
Pauli string \(P\) from the distribution \(\overline\mu_R\) and average the
corresponding values \(f(P)\). The key point is that the gates outside the
eye \(\mathcal E_R\) are Haar averaged and can therefore be treated through
the classical Pauli transition rule, while the fixed gates inside the eye
must still be evolved coherently. The following lemma shows that one can
sample \(\overline\mu_R\) while keeping coherent amplitudes on only
\(O(W_R)\) qubits at any time, where \(W_R\) denotes the maximum number of
retained gates in the eye in any single layer.

\begin{lemma}[Sampling the conditional mean]
\label{lem:endpoint-sim-conditional-sampler}
For any fixed realization of the eye gates, one can draw an exact sample
\(P\sim\overline\mu_R\), with \(\overline\mu_R\) defined in
Eq.~\eqref{eq:endpoint-sim-conditional-distribution}, using
\(n^2 2^{O(W_R)}\) arithmetic operations.
\end{lemma}

\begin{proof}
Fix the gates in the eye \(\mathcal E_R\). The remaining randomness in
\(\overline\mu_R\) comes only from the Haar gates outside the eye.
Introduce the Pauli coefficient vector
\[
    |c_U\rangle\coloneqq\sum_{P\in\mathcal P_n}c_U(P)|P\rangle
\]
and its conditional second-moment matrix
\begin{equation}
\label{eq:endpoint-sim-second-moment-matrix}
    \rho_R\coloneqq
    \mathbb E\!\left[
        |c_U\rangle\langle c_U|
        \,\middle|\,
        U_{\mathcal E_R}
    \right].
\end{equation}
Its diagonal is exactly the distribution that we want to sample:
\begin{equation}
\label{eq:endpoint-sim-second-moment-diagonal}
    \langle P|\rho_R|P\rangle
    =
    \mathbb E\!\left[
        |c_U(P)|^2
        \,\middle|\,
        U_{\mathcal E_R}
    \right]
    =
    \overline\mu_R(P).
\end{equation}
Thus, by Eq.~\eqref{eq:endpoint-sim-second-moment-diagonal}, it is enough to
sample from the diagonal of \(\rho_R\). The algorithm never constructs
\(\rho_R\) explicitly.

We now describe the sampler. At every stage, it stores a normalized vector
\(\psi\) of Pauli amplitudes on a set of coherent qubits, while every other
qubit has a definite classical Pauli label. Initially all labels are fixed
by the Pauli string \(Z_1\), and \(\psi=1\) on the empty set. Within each
layer, we process the gates outside the eye first and the eye gates
afterwards; this is allowed because gates in the same layer act on disjoint
bonds.

Consider a gate acting on two qubits. Before processing it, any classical
Pauli labels on these two qubits are included in \(\psi\) as basis vectors.
We can then write
\begin{equation}
\label{eq:endpoint-sim-psi-decomposition}
    |\psi\rangle
    =
    \sum_{a,\alpha}
    \psi(a,\alpha)|a\rangle|\alpha\rangle,
\end{equation}
where \(a\in\mathcal P_2\) labels the two qubits of the gate and
\(\alpha\) labels the remaining coherent qubits. We process the gate as
follows.
\begin{itemize}
    \item If the gate \(G\) lies inside the eye, it is fixed. We update the
    Pauli amplitudes exactly according to
    \begin{equation}
    \label{eq:endpoint-sim-fixed-gate-update}
        \psi'(b,\alpha)
        =
        \sum_a T^G_{b,a}\psi(a,\alpha),
        \qquad
        T^G_{b,a}
        =
        \frac14\Tr(bG^\dagger aG).
    \end{equation}
    The two qubits remain part of the coherent vector.

    \item If the gate lies outside the eye, it is Haar averaged. Define
    \begin{equation}
    \label{eq:endpoint-sim-input-distribution}
        q(a)\coloneqq\sum_\alpha|\psi(a,\alpha)|^2.
    \end{equation}
    Since \(\psi\) is normalized, \(q\) is a probability distribution on
    \(\mathcal P_2\). Sample \(a\) according to \(q(a)\), and retain the
    normalized conditional vector
    \begin{equation}
    \label{eq:endpoint-sim-conditional-vector}
        \phi_a(\alpha)
        \coloneqq
        \frac{\psi(a,\alpha)}{\sqrt{q(a)}}.
    \end{equation}
    Next sample an output Pauli \(b\) according to
    \begin{equation}
    \label{eq:endpoint-sim-Pauli-transition}
        K(a\to b)
        =
        \begin{cases}
            1, & a=b=II,\\
            1/15, & a\neq II,\ b\neq II,\\
            0, & \text{otherwise}.
        \end{cases}
    \end{equation}
    Store the two single-qubit labels of \(b\) classically and remove the
    two gate qubits from the coherent vector, which is now \(\phi_a\).
\end{itemize}
After the last gate, we sample the Pauli labels of the qubits that remain
in \(\psi\) according to \(|\psi|^2\), and combine them with the classical
labels already stored. This produces a full Pauli string \(P\).

We now prove that the resulting \(P\) is distributed exactly according to
\(\overline\mu_R\). The invariant is that, after every processed gate, the
average outer product of the random full Pauli coefficient vector represented
by \(\psi\) and the classical labels is exactly the second-moment matrix
obtained by Haar averaging the outside gates processed so far. At the end
of the circuit, this matrix is precisely \(\rho_R\) in
Eq.~\eqref{eq:endpoint-sim-second-moment-matrix}. The invariant is true
initially, since the coefficient vector is the deterministic Pauli string
\(Z_1\).

A fixed eye gate preserves the invariant immediately, since its Pauli
transfer matrix \(T^G\) is applied exactly according to
Eq.~\eqref{eq:endpoint-sim-fixed-gate-update}. Consider instead a gate
outside the eye. Its Haar-averaged Pauli transfer matrix satisfies
\begin{equation}
\label{eq:endpoint-sim-Haar-second-moment}
    \mathbb E_G\!\left[
        T^G_{b,a}\overline{T^G_{b',a'}}
    \right]
    =
    \delta_{a,a'}\delta_{b,b'}K(a\to b),
\end{equation}
with \(K\) defined in Eq.~\eqref{eq:endpoint-sim-Pauli-transition}.
Hence Haar averaging this gate transforms the coefficient products as
\begin{equation}
\label{eq:endpoint-sim-Haar-coefficient-update}
    \mathbb E_G\!\left[
        \psi'_G(b,\alpha)
        \overline{\psi'_G(b',\beta)}
    \right]
    =
    \delta_{b,b'}
    \sum_a K(a\to b)
    \psi(a,\alpha)\overline{\psi(a,\beta)}.
\end{equation}
On the other hand, averaging over the two samples \(a\) and \(b\) produced
by the algorithm gives
\begin{align}
\label{eq:endpoint-sim-sampling-coefficient-update}
    &\delta_{b,b'}\sum_a
    q(a)K(a\to b)
    \phi_a(\alpha)\overline{\phi_a(\beta)}=
    \delta_{b,b'}\sum_a
    K(a\to b)
    \psi(a,\alpha)\overline{\psi(a,\beta)},
\end{align}
where, by Eqs.~\eqref{eq:endpoint-sim-input-distribution} and
\eqref{eq:endpoint-sim-conditional-vector},
\begin{equation}
\label{eq:endpoint-sim-conditional-cancellation}
    q(a)\phi_a(\alpha)\overline{\phi_a(\beta)}
    =
    \psi(a,\alpha)\overline{\psi(a,\beta)}.
\end{equation}
Comparing Eqs.~\eqref{eq:endpoint-sim-Haar-coefficient-update} and
\eqref{eq:endpoint-sim-sampling-coefficient-update}, we see that the
sampling update reproduces exactly the Haar-averaged second-moment update.
In particular, it preserves not only the diagonal probabilities but also
the off-diagonal coefficient products on the remaining coherent qubits,
which are needed by subsequent eye gates.

Induction over all gates therefore shows that, after the full circuit, the
average outer product of the coefficient vector represented by the sampler
is exactly the conditional second-moment matrix \(\rho_R\) in
Eq.~\eqref{eq:endpoint-sim-second-moment-matrix}. Consequently, after the
final sampling step,
\begin{equation}
\label{eq:endpoint-sim-final-sampling-distribution}
    \Pr(P)
    =
    \langle P|\rho_R|P\rangle
    =
    \overline\mu_R(P),
\end{equation}
where the last equality follows from
Eq.~\eqref{eq:endpoint-sim-second-moment-diagonal}. Thus the algorithm draws
an exact sample from \(\overline\mu_R\).

It remains to bound the cost. At the end of each layer, a qubit can remain
in \(\psi\) only if it belongs to an eye gate in that layer or if it is an
idle boundary qubit. Hence \(\psi\) contains at most \(2W_R+2\) qubits.
During the outside-gate updates of the next layer, each gate can temporarily
add at most two qubits, which are removed again after the local Pauli labels
have been sampled. Therefore the number of coherent qubits never exceeds
\(2W_R+4=O(W_R)\).

If \(\psi\) contains \(m=O(W_R)\) qubits, it has \(4^m\) Pauli
coefficients. Computing all sixteen probabilities \(q(a)\) in
Eq.~\eqref{eq:endpoint-sim-input-distribution} requires \(O(4^m)\)
operations, since each coefficient contributes to exactly one of these
sums. Selecting and normalizing the sampled branch according to
Eq.~\eqref{eq:endpoint-sim-conditional-vector} has the same cost up to a
constant factor. A fixed eye gate also costs \(O(4^m)\) operations, since
the \(16\times16\) Pauli transfer matrix in
Eq.~\eqref{eq:endpoint-sim-fixed-gate-update} acts only on the two local
Pauli indices. Thus every gate can be processed using \(2^{O(W_R)}\)
arithmetic operations. Since the circuit contains \(O(n^2)\) gates, one
exact sample from \(\overline\mu_R\) requires $n^2 2^{O(W_R)}$ arithmetic operations.
\end{proof}
With these ingredients in place, we can now state the full classical
algorithm. Lemma~\ref{lem:endpoint-sim-eye-error} shows that averaging the
gates outside \(\mathcal E_R\) introduces only a small error, while
Lemma~\ref{lem:endpoint-sim-conditional-sampler} gives an efficient way to
sample from the resulting conditional distribution. The parameter \(R\)
controls the first error, and the number of samples \(N\) controls the
second.

\begin{table}[!htbp]
\centering
\begin{minipage}{0.94\linewidth}
\small
\hrule
\smallskip
\textbf{Classical estimation of the endpoint OTOC}
\smallskip
\hrule
\medskip
\textbf{Input:} the realized circuit \(U_d\), the eye parameter \(R\), and
the number of samples \(N\), chosen as in
Theorem~\ref{thm:endpoint-sim-otoc1-runtime}.
\begin{enumerate}[leftmargin=1.6em,itemsep=0.4em]
\item Construct the eye \(\mathcal E_R\). Keep the realized gates inside
\(\mathcal E_R\) fixed and Haar-average the gates outside it.
\item Using Lemma~\ref{lem:endpoint-sim-conditional-sampler}, independently
sample \(N\) Pauli strings \(P^{(1)},\ldots,P^{(N)}\) from
\(\overline\mu_R\). For each sample, assign the sign
\[
    f(P^{(j)})=
    \begin{cases}
        +1, & P^{(j)}_n\in\{I,Z\},\\
        -1, & P^{(j)}_n\in\{X,Y\}.
    \end{cases}
\]
\item Return
\[
    \widetilde F\coloneqq
    \frac1N\sum_{j=1}^N f(P^{(j)}).
\]
\end{enumerate}
\hrule
\end{minipage}
\caption{Classical estimator for the infinite-temperature endpoint OTOC.}
\label{tab:endpoint-otoc1-explicit-simulation}
\end{table}

\begin{theorem}[Classical simulation of the endpoint OTOC]
\label{thm:endpoint-sim-otoc1-runtime}
For \(\varepsilon,\delta\in(0,1)\), choose
\[
    R=10\sqrt{\log\frac{400n^2}{3\varepsilon\delta}},
    \qquad
    N=\left\lceil\frac8{\varepsilon^2}\log\frac4\delta\right\rceil.
\]
Then the algorithm in
Table~\ref{tab:endpoint-otoc1-explicit-simulation} satisfies
\[
    \Pr\!\left[
        |\widetilde F-F_\infty(U_d)|\le\varepsilon
    \right]
    \ge1-\delta,
\]
where the probability is over both the Haar-random circuit and the
randomness of the algorithm. Its runtime is
\[
    O\!\left(
        n^2
        2^{O\left(\sqrt{n\log(n/(\varepsilon\delta))}\right)}
        \frac{\log(2/\delta)}{\varepsilon^2}
    \right)
\]
arithmetic operations, assuming sampling from explicitly computed finite
distributions. In particular, for inverse-polynomial \(\varepsilon\) and
\(\delta\), the runtime is
\[
    \operatorname{poly}(n)\,
    2^{O(\sqrt{n\log n})}.
\]
\end{theorem}

\begin{proof}
There are two sources of error. First, replacing the original OTOC by the
conditional mean \(\widehat F_R\) introduces the error controlled by
Lemma~\ref{lem:endpoint-sim-eye-error}. Second, we estimate
\(\widehat F_R\) from finitely many samples.

Conditioned on the eye gates, Lemma~\ref{lem:endpoint-sim-conditional-sampler}
produces independent Pauli strings with distribution \(\overline\mu_R\).
By Eq.~\eqref{eq:endpoint-sim-conditional-distribution}, the corresponding
signs have mean
\[
    \mathbb E[f(P)\mid U_{\mathcal E_R}]
    =
    \widehat F_R
\]
and take values in \(\{-1,+1\}\). Hoeffding's inequality therefore gives
\[
    \Pr\!\left[
        |\widetilde F-\widehat F_R|>\frac{\varepsilon}{2}
        \,\middle|\,
        U_{\mathcal E_R}
    \right]
    \le
    2e^{-N\varepsilon^2/8}
    \le
    \frac{\delta}{2}.
\]

For the error caused by averaging the gates outside the eye,
Lemma~\ref{lem:endpoint-sim-eye-error} and Markov's inequality give
\[
    \Pr\!\left[
        |F_\infty(U_d)-\widehat F_R|>\frac{\varepsilon}{2}
    \right]
    \le
    \frac{200n^2}{3\varepsilon}e^{-R^2/100}
    =
    \frac{\delta}{2}.
\]
With probability at least \(1-\delta\), both errors are therefore at most
\(\varepsilon/2\). The triangle inequality then gives
\[
    |\widetilde F-F_\infty(U_d)|
    \le
    |\widetilde F-\widehat F_R|
    +
    |\widehat F_R-F_\infty(U_d)|
    \le
    \varepsilon.
\]

It remains to bound the runtime. By
Lemma~\ref{lem:endpoint-sim-conditional-sampler}, one sample costs
\(n^2 2^{O(W_R)}\) arithmetic operations. Since
\(W_R=O(R\sqrt n)\), our choice of \(R\) gives
\[
    W_R
    =
    O\!\left(
        \sqrt{n\log\frac{n}{\varepsilon\delta}}
    \right).
\]
Multiplying the cost of one sample by
\(N=O(\varepsilon^{-2}\log(2/\delta))\) gives the stated runtime.
\end{proof}
\end{document}